\newif\ifsubmit     % hide comments?
\newif\ifllncs      % blarf
\newif\ifexabs      % extended abstract
\newif\ifblind      % no author namez

\submittrue
\ifllncs
  \documentclass[runningheads,a4paper]{llncs}

\else
  \documentclass[letterpaper,11pt,pdfa]{article}
  \usepackage[in]{fullpage}
\fi

\usepackage{iftex}
\ifPDFTeX
  \usepackage[utf8]{inputenc}
  \usepackage[noTeX]{mmap}
  \usepackage[T1]{fontenc}
\fi
\ifLuaTeX
  \usepackage{luatex85}
  \usepackage[noTeX]{mmap}
\fi

\ifllncs
\else

\fi

\usepackage{mdframed}
\usepackage{amsfonts}
\usepackage{amssymb}
\usepackage{amsthm}
\usepackage{color}

\usepackage{appendix}
\usepackage{algorithm}
\usepackage{algpseudocode}
\usepackage{braket}
\usepackage{comment}
\usepackage{url}
\usepackage{bbm}
\usepackage{multicol}
\usepackage{mdframed}
\usepackage{multirow}
\usepackage{amsmath}

\usepackage[bookmarks]{hyperref}
\usepackage[nameinlink]{cleveref}
\hypersetup{
  colorlinks,% Defaults to red
  allcolors=blue
}

\ifllncs

  \spnewtheorem{claim}{Claim}{\bfseries}{\rmfamily}
  \crefname{claim}{claim}{claims}
  \Crefname{claim}{Claim}{Claims}
\else
  \newtheorem{theorem}{Theorem}[section]
  \newtheorem{definition}[theorem]{Definition}
  \newtheorem{remark}[theorem]{Remark}
  \newtheorem{lemma}[theorem]{Lemma}
  \newtheorem{corollary}[theorem]{Corollary}

  \newtheorem{claim}[theorem]{Claim}
  
\fi
  
  \newtheorem*{remark*}{Remark}

  \newtheorem*{theorem*}{Theorem}
  \newtheorem*{lemma*}{Lemma}
  \newtheorem{innercustomtheorem}{Theorem}
  \newenvironment{customtheorem}[1]
    {\renewcommand\theinnercustomtheorem{#1}\innercustomtheorem}
    {\endinnercustomtheorem}
  \newtheorem{innercustomlemma}{Lemma}
  \newenvironment{customlemma}[1]
    {\renewcommand\theinnercustomlemma{#1}\innercustomlemma}
    {\endinnercustomlemma}
\usepackage[style=alphabetic,minalphanames=3,maxalphanames=4,maxnames=99,backref=true]{biblatex}

  \DeclareFieldFormat{eprint:iacr}{Cryptology ePrint Archive: \href{https://ia.cr/#1}{\texttt{#1}}}
  \DeclareFieldFormat{eprint:iacrarchive}{Cryptology ePrint Archive: \href{https://eprint.iacr.org/archive/#1}{\texttt{#1}}}
  \AtEveryBibitem{% Clean up the bibtex rather than editing it
    \clearlist{address}
    \clearfield{date}
    \clearfield{isbn}
    \clearfield{issn}
    \clearlist{location}
    \clearfield{month}
    \clearfield{series}
 
    \ifentrytype{book}{}{% Remove publisher and editor except for books
      \clearlist{publisher}
      \clearname{editor}
    }
  }

\usepackage{appendix}
\usepackage{algorithm}
\usepackage{algpseudocode}
\usepackage{braket}
\usepackage{comment}
\usepackage{url}
\usepackage{multicol}
\usepackage{tikz}
\usetikzlibrary{arrows,automata,positioning}
\usepackage{caption}
\usepackage[caption=false]{subfig}
\usepackage[font=small,labelfont=bf]{caption}
\usepackage{comment}
\usepackage{pifont}
\usetikzlibrary{patterns}
\usetikzlibrary{crypto.symbols}

\usepackage{enumitem}
\setlist[description]{noitemsep}
\setlist[enumerate]{noitemsep}
\setlist[itemize]{noitemsep}

\usepackage{soul, xcolor, xparse}
\makeatletter
  \ExplSyntaxOn
    \cs_new:Npn \white_text:n #1
    {
      \fp_set:Nn \l_tmpa_fp {#1 * .01}
      \llap{\textcolor{white}{\the\SOUL@syllable}\hspace{\fp_to_decimal:N \l_tmpa_fp em}}
      \llap{\textcolor{white}{\the\SOUL@syllable}\hspace{-\fp_to_decimal:N \l_tmpa_fp em}}
    }
    \NewDocumentCommand{\whiten}{ m }
    {
      \int_step_function:nnnN {1}{1}{#1} \white_text:n
    }
  \ExplSyntaxOff
  
  \NewDocumentCommand{ \varul }{ D<>{5} O{0.2ex} O{0.1ex} +m } {%
    \begingroup
    \setul{#2}{#3}%
    \def\SOUL@uleverysyllable{%
      \setbox0=\hbox{\the\SOUL@syllable}%
      \ifdim\dp0>\z@
      \SOUL@ulunderline{\phantom{\the\SOUL@syllable}}%
      \whiten{#1}%
      \llap{%
        \the\SOUL@syllable
        \SOUL@setkern\SOUL@charkern
      }%
      \else
      \SOUL@ulunderline{%
        \the\SOUL@syllable
        \SOUL@setkern\SOUL@charkern
      }%
      \fi}%
    \ul{#4}%
    \endgroup
  }
\makeatother

\usepackage{mleftright}

\newcommand{\As}{\mathcal{A}}

\newcommand{\Ks}{\mathcal{K}}

\newcommand{\cA}{\mathcal{A}}

\newcommand{\cB}{\mathcal{B}}
\newcommand{\cC}{\mathcal{C}}

\newcommand{\ketbra}[2]{\ket{#1}\!\bra{#2}}

\mdfdefinestyle{figstyle}{ 
  linecolor=black!7, 
  backgroundcolor=black!7, 
  innertopmargin=10pt, 
  innerleftmargin=25pt, 
  innerrightmargin=25pt, 
  innerbottommargin=10pt 
}

\renewcommand{\kappa}{\ell}

\newcommand{\poly}{{\sf poly}}

\newcommand{\revise}[1]{{\color{red} #1}}

\newcommand{\sponge}{\mathsf{Sp}}

\ifsubmit
    \newcommand{\qipeng}[1]{}
    \newcommand{\minki}[1]{}
    \newcommand{\takashi}[1]{}
    \newcommand{\alex}[1]{}
    \newcommand{\aaram}[1]{}
\else
    \newcommand{\qipeng}[1]{{\color{red} Qipeng: #1}}
    \newcommand{\minki}[1]{{\color{olive} Minki: #1}}
    \newcommand{\takashi}[1]{{\color{orange} Takashi: #1}}
    \newcommand{\alex}[1]{{\color{green} Alex: #1}}
    \newcommand{\aaram}[1]{{\color{cyan} Aaram: #1}}
\fi
\newcommand{\forwardquery}[1]{{\color{lightgray} Related to forward-query case, currently not needed or wrong\\#1}}

\newcommand{\tok}[1]{\underset{#1}{\to}}

\newtheorem{innercustomthm}{Theorem}
{\renewcommand\theinnercustomthm{#1.}\innercustomthm}
{\endinnercustomthm}

\title{
    Quantum Lifting for Invertible Permutations and Ideal Ciphers
}

\ifllncs
\author{
}
\institute{
}
\else
\author{Alexandru Cojocaru \thanks{University of Edinburgh}
\and Minki Hhan \thanks{The University of Texas at Austin}
\and Qipeng Liu \thanks{UC San Diego}
\and Takashi Yamakawa \thanks{NTT Social Informatics Laboratories}
\and Aaram Yun \thanks{Ewha Womans University}
}
\fi

\date{}

\begin{document}

\maketitle

\begin{abstract}
In this work, we derive the first lifting theorems for establishing security in the quantum random permutation and ideal cipher models. These theorems relate the success probability of an arbitrary quantum adversary to that of a classical algorithm making only a small number of classical queries.

By applying these lifting theorems, we improve previous results and obtain new quantum query complexity bounds and post-quantum security results. Notably, we derive tight bounds for the quantum hardness of the double-sided zero search game and establish the post-quantum security for the preimage resistance, one-wayness, and multi-collision resistance of constant-round sponge, as well as the collision resistance of the Davies-Meyer construction.
\end{abstract}

\ifllncs
\else
\tableofcontents
\fi

\section{Introduction}

The random permutation model (RPM) and the ideal cipher model (ICM) are idealized models that provide simplified analyses for cryptographic constructions based on block ciphers and hash function designs,
% \minki{Picky comment; I think hash is usually idealized by ROM, not RPM. The practice of RPM I know is to say that block cipher with random key looks like random permutation so that more reasonable idealization of block cipher; wait I noticed that this is from \cite{coretti2018non}, then it may be okay.}\takashi{I agree that this is inaccurate. I checked \cite{coretti2018non}, but they mean constructions \emph{of} hash functions and not constructions \emph{based on} hash functions.}
% \minki{I slightly changed according to the discussion}
% or block ciphers, 
which often lack rigorous security foundations. Similar to the random oracle model (ROM), both RPM and ICM capture generic attacks --- those that treat the underlying cryptographic primitives as black boxes. Such models often offer insight into the best possible attacks for many natural applications\footnote{While several studies have established separations between idealized models~\cite{goldwasser2003security,canetti2004random,black2006ideal} and the standard model, these separations are often contrived (except for the very recent work \cite{khovratovich2025prove}) and rely on specific structures that can only be exploited by non-black-box attacks. 
% \minki{Another picky comment; a recent work \url{https://eprint.iacr.org/2025/118} shows a non-contrived application can be broken in some cases. But given the current introduction, the current sentence is okay as it uses ``often''} \qipeng{Yes. I thought of this when I wrote it, but then I forogt.. defintely worth mentioning}
}. This approach is commonly referred to as the so-called RPM/ICM methodology~\cite{coretti2018non}:

\begin{center}
{\it \textbf{RPM/ICM methodology.} For “natural” applications of hash functions and block
ciphers, the concrete security proven in the RPM/ICM is the right bound even in the
standard model, assuming the “best possible” instantiation for the idealized component
(permutation or block cipher) is chosen.  }  
\end{center}

In the random permutation model (RPM), every party has access to $\pi$ and $\pi^{-1}$ for a uniformly chosen random permutation $\pi$. In the ideal cipher model (ICM), each party has oracle access to $E_K(\cdot)$ and $E_K^{-1}(\cdot)$, where each $E_K(\cdot)$ is an independent random permutation.  
A wide range of constructions and proofs have been developed within the RPM/ICM framework, leading to significant successes. Many of these constructions have been standardized by the National Institute of Standards and Technology, including the Even-Mansour cipher (AES), 
%\minki{is there any relation betwwen EM and AES?}\takashi{In a very broad sense, I think AES can be seen as (iterated) EM construction, but I'm not sure if this view is standard.}
the Davies-Meyer hash function (SHA-1/2 and MD5), the sponge construction (SHA-3).

\paragraph{Quantum RPM/ICM.} While RPM/ICM provides a precise characterization of generic attacks, it does not account for potential quantum attacks: i.e., the ability to compute the public function in superposition. In the quantum random permutation model (or QRPM), a quantum algorithm can query the following unitaries for one unit of cost:
\begin{align*}
    U_{\pi} \ket {x} \ket y &= \ket {x} \ket {y \oplus \pi(x)}, \text{ and }  \\
    U_{\pi^{-1}} \ket {x} \ket y &= \ket {x} \ket {y \oplus \pi^{-1}(x)}.
\end{align*}
Similarly, in the quantum ideal cipher model (or QICM), a quantum algorithm can query the function $E(\cdot,  \cdot)$ in superposition (i.e., both the key $k$ and the input $x$):
\begin{align*}
    U_{E} \ket K \ket {x} \ket y &= \ket K \ket {x} \ket {y \oplus E_K(x)}, \text{ and } \\
    U_{E^{-1}} \ket K \ket {x} \ket y &= \ket K \ket {x} \ket {y \oplus E^{-1}_K(x)}.
\end{align*}

Since the proposal of the quantum random oracle model (QROM)~\cite{AC:BDFLSZ11}, the quantum idealized models have received a lot of attention because it  characterizes the ``best-possible'' quantum generic attacks. Many tools have been developed in the QROM~\cite{ambainis2019quantum,zhandry2019record,YZ21} and almost all important constructions in the ROM, including the Fiat-Shamir transformation~\cite{DFMS19,liu2019revisiting} and the Fujisaki-Okamoto transformation~\cite{TarghiU16,HHK17,JZCWM18}, 
%\takashi{I replaced the citation with the older papers. (Originally, \cite{bindel2019tighter,don2022online} was cited.)}
have their security shown in the QROM. 

However, the situation becomes far less clear in the QRPM and QICM. There are already two major differences between the QROM and QRPM/QICM:  
\begin{enumerate}
    \item Random functions exhibit perfect independence; that is, the outputs for all inputs are pairwise independent. In contrast, while the outputs of a random permutation are close to independent, the weak correlations between them completely invalidate or complicate almost all methods used in the QROM. This issue persists even when an algorithm has oracle access only to the forward permutation $\pi$ and not to its inverse $\pi^{-1}$.
    \item In both QRPM and QICM, an algorithm has quantum access not only to the original permutation $\pi$ but also to its inverse $\pi^{-1}$. The ability to query the backward oracle $\pi^{-1}$ destroys independence, making it significantly more challenging to extend QROM-based arguments to these settings.
\end{enumerate}

Various approaches have been proposed to establish security in QPRM and QICM, but each comes with its own limitations. When an algorithm has oracle access only to the forward oracle of a random permutation, one can leverage the indistinguishability between random permutations and random functions~\cite{yuen2013quantum,zhandry2013note} to argue security in QROM. This approach introduces a small additive loss in the security analysis but nonetheless preserves the overall security guarantees.  

However, when an algorithm has oracle access to both the forward and inverse oracles, it can immediately distinguish a random permutation from a random function, rendering the above method ineffective. To address this, one line of work attempts to extend Zhandry's compressed oracle technique~\cite{rosmanis2021tight,alagic2022post,DBLP:conf/asiacrypt/Unruh23,alagic2024post,MMW24}, while another introduces novel techniques~\cite{DBLP:conf/asiacrypt/HosoyamadaY18,asiacrypt/Zhandry21,alagic2023two,CP24,CPZ24} for specific constructions.  
%\takashi{I added \cite{CPZ24}.} \minki{I moved \cite{alagic2023two} because it does not use compressed oracle and removed \cite{czajkowski2019quantum,czajkowski2021quantum,unruh2021compressed} because they are not about permutation or unclear if correct.}\takashi{I added \cite{DBLP:conf/asiacrypt/HosoyamadaY18,asiacrypt/Zhandry21}.}

These approaches, however, are often problem-specific and challenging to generalize. For example, \cite{MMW24} employs a strictly monotone factorization to represent permutations, establishing a lower bound for finding $(x, \pi(x))$ in a relation $R$. While their method is theoretically generalizable, it quickly becomes complex as the underlying problem grows in complexity. Similarly, \cite{CP24} introduces a symmetrization technique to prove the one-wayness of the single-round sponge construction, but this method does not easily extend to multiple rounds or to other properties of the sponge construction.

The difficulty described above in analyzing the QRPM/QICM is not merely a technical limitation. Indeed, we can construct cryptographic schemes that are secure in the (classical) RPM/ICM but insecure in the QRPM/QICM.\footnote{\cite{YZ24} gave cryptographic schemes that are secure in the (classical) ROM but insecure in the QROM. This can be extended to a separation between the QRPM/QICM and RPM/ICM by instantiating a random oracle using a permutation-based hash function that is indifferentiable from a random oracle (e.g., sponge construction~\cite{sponge,BDP+08}).} In fact, a similar gap also exists between the QROM and the ROM. Nonetheless, \cite{YZ21} has proposed a ``lifting theorem'' which, albeit with some loss, upgrades a proof in the ROM to one in the QROM.
This leads to the following natural question: 

% \minki{I found the transition between two paragraphs is somewhat sudden; we may say something about \cite{YZ21} here?}\takashi{I added one paragraph.}\minki{Looks good!}

%Given the importance of QRPM and QICM, we pose the following question in this paper:  
\begin{center}  
    {\it Is there a general theorem that seamlessly lifts any classical RPM/ICM proof to a proof in QRPM/QICM?}  
\end{center}  
We answer this question affirmatively, and reprove numerous results from previous works within the QRPM/QICM framework 
as well as obtain new ones
 using simple arguments.

\subsection{Our Results}

\noindent{\bf Quantum Lifting Theorem for Interactive Search Games in RPM/ICM.} Our central result proposes a novel lifting theorem for \emph{interactive search games} in RPM/ICM that relates the success probability of an arbitrary quantum algorithm with the success probability of a classical algorithm performing a much smaller number of queries. More concretely, our main results for QRPM can be stated as follows:

\begin{theorem}[Quantum Lifting Theorem on Random Permutation]
\label{thm:quantum_lifting_informal}
Let $\mathcal{G}$ be an (interactive) search game with a  challenger $\cC$ that performs at most $k$ classical queries to the invertible random permutation $\pi : X \rightarrow X$, and let $\cA$ be an algorithm that performs $q$ quantum queries to $\pi$.
Then there exists an adversary $\cB$ making at most $k$ classical queries to $\pi$
%$k$-query classical adversary $\cB$ 
such that:
%\aaram{`an adversary $\cB$ that makes $k$ classical queries to $\pi$'?}\alex{Sounds good, added that.}\aaram{No my point was that since B runs A internally, calling B as a classical algorithm may be wrong?}\alex{Sorry, you're right, I misunderstood your comment, then yes, we can just go with adversary making $k$ classical queries}
\begin{align*}
\Pr[\cB \textrm{ wins } \mathcal{G}] \geq 
\frac{\left(1 - \frac{k^2}{|X|}\right)}{(8q+1)^{2k}}
\Pr[\cA \text{ wins } \mathcal{G}].
\end{align*} 
\end{theorem}
\noindent Note that in almost all games, the challenger $\cC$ is efficient and thus makes only a polynomial number of queries. Consequently, we can safely assume that $1 - \frac{k^2}{|X|} \geq \frac{1}{2}$,  which does not affect the asymptotic order of the success probability in the search game.  

We demonstrate the power and simplicity of our lifting theorem through the example of double-sided zero search~\cite{DBLP:conf/asiacrypt/Unruh23}. In this game $\mathcal G_{\sf double\text{-}sided}$, a random permutation $\pi: \{0,1\}^{2n} \to \{0,1\}^{2n}$, as well as its inverse $\pi^{-1}$ are given, and the goal is to find $x, y$ such that $\pi(x) = y$, where both $x$ and $y$ have at least $n$ leading zeros. Clearly, $k = 1$, as the challenger makes only a single oracle query to the random permutation.  

For a single-classical-query algorithm $\cA$, the probability that $\cB$ wins $\mathcal G_{\sf double\text{-}sided}$ is at most $2/2^n$.
% \minki{I think this is incorrect in general, probably we may be talking about $\mathcal G$ being the double-sided zero search?} 
This follows from two possibilities: either the algorithm queries a valid pair or it correctly guesses one. The probabilities of these events sum to $1/2^n + 1/2^n = 2/2^n$.  
Applying our main theorem (\Cref{thm:quantum_lifting_informal}), we obtain the following bound for any $q$-query quantum algorithm in QRPM.
\begin{align*}
    \Pr[\cA \text{ wins } \mathcal G_{\sf double\text{-}sided}] \leq 2(8 q + 1)^2 \Pr[\cB \text{ wins } \mathcal G_{\sf double\text{-}sided}] = O(q^2/2^n).
\end{align*}
This bound is tight to that by~\cite{CP24}. 
Furthermore, this result can be easily generalized to the game $\mathcal G_R$ to find a pair $(x, y)$ in an arbitrary relation $R$: 
\begin{align*}
    \Pr[\cA \text{ wins } \mathcal G_R] \leq (8 q + 1)^2 \Pr[\cB \text{ wins } \mathcal G_R] = O(r_{\sf max} \cdot q^2/2^{2 n})
\end{align*}
where $r_{\sf max}/2^{2n}$ is the probability that a single query reveals a pair $(x,y) \in R$. 
This improves the bound $\tilde{O}(r_{\sf max} q^3/2^{2n})$ by~\cite{MMW24} and is also tight to  Grover's search. 

\medskip

We also prove a similar lifting theorem in the QICM. 
\begin{theorem}[Quantum Lifting Theorem on Ideal Ciphers]
\label{thm:quantum_icm_lifting_informal}
Let $\mathcal{G}$ be an (interactive) search game with a challenger $\cC$ that performs at most $k$ classical queries to an ideal cipher oracle $E:\Ks\times X\to X$, and let $\cA$ be an algorithm that performs $q$ quantum queries to the ideal cipher.
Then there exists an adversary $\cB$ making at most $k$ classical queries to $E$
%$k$-query classical adversary $\cB$
such that: 
%\aaram{`an adversary $\cB$ that makes $k$ classical queries to $E$'?}\alex{I added}
\begin{align*}
\Pr[\cB \textrm{ wins } \mathcal{G}] \geq 
\frac{\left(1 - \frac{k^2}{|X|}\right)}{(8q+1)^{2k}}
\Pr[\cA \text{ wins } \mathcal{G}].
\end{align*} 
\end{theorem}

\paragraph{Other Applications in QRPM/QICM.}

Beyond the generalized double-sided search, as previously mentioned, our lifting theorems also have many applications in the random permutation and ideal cipher models.

% \noindent\emph{Hardness of Generalized Double-Sided Search.}
% We first give an improved bound for the generalized double-sided search problem considered in \cite{MMW24} as a generalization of  Unruh's double-sided zero-search conjecture~\cite{DBLP:conf/asiacrypt/Unruh23}.   
% \begin{theorem}[Generalized Double-Sided Search (Informal)]\label{thm:generalized_double_search_informal}
% Let 
% $N$ be a positive integer, $A$ be a quantum algorithm that makes $q$ quantum queries to a uniformly random permutation $\pi$ on $[N]$ and its inverse, and $R\subseteq [N]\times [N]$ be an arbitrary relation. Then it holds that 
% \[
% \Pr_{\pi}\left[
% (x,\pi(x))\in R:
% x \leftarrow A^{\pi,\pi^{-1}}
% \right]
% \le 
% \frac{2\cdot(8q+1)^{2} \cdot r_{\mathrm{max}}}{N-1}
% \]
% where $r_{\mathrm{max}}=\max\{\max_{x} |R_x|,\max_{y}|R_y^{\mathrm{inv}}|\}$ with $R_x=\{y:(x,y)\in R\}$ and $R_y^{\mathrm{inv}}=\{x:(x,y)\in R\}$.
% \end{theorem}

\paragraph{Sponge construction.}
The sponge construction~\cite{sponge} is a permutation-based hashing algorithm that underlies SHA-3. In the classical setting, it is known to be indifferentiable from a random oracle~\cite{BDP+08}, intuitively meaning it is as secure as a random oracle in the RPM. However, little is known about its security in the post-quantum setting. 
The current state of the art is that the \emph{single-round} sponge satisfies preimage-resistance and one-wayness~\cite{CP24,MMW24}, and also achieves \emph{reset indifferentiability} (even with advice) under a certain parameter regime~\cite{asiacrypt/Zhandry21,CPZ24},\footnote{
Specifically, $r\le c$ where
$r$ is the rate and  $c$ is the capacity.} implying that in this regime it is as secure as a random oracle against quantum adversaries. In contrast, no results were previously known for the multi-round sponge.

Using our lifting theorem, we reduce the post-quantum security of sponge to its classical counterpart. As a result, we obtain non-trivial security bounds for preimage-resistance, one-wayness, and (multi-)collision-resistance of \emph{constant-round} sponge. Although these bounds are not tight, this work represents the first non-trivial security result for multi-round sponge constructions.
%See \Cref{sec:sponge} for details. 

\paragraph{Davies-Meyer and PGV hash functions.}
The Davies-Meyer construction~\cite{Winternitz84} is a block-cipher-based hashing algorithm that underlies SHA-1, SHA-2, and MD5. In the classical setting, it is proven to satisfy one-wayness and collision-resistance in the ICM~\cite{Winternitz84,BRS02,BRSS2010}. In the quantum setting, although it is shown to be one-way in the QICM~\cite{DBLP:conf/asiacrypt/HosoyamadaY18}, its collision-resistance remained an open question. 

By applying our lifting theorem, we prove that the Davies-Meyer construction satisfies collision-resistance in the QICM, albeit our result is not tight. 
We remark that our analysis is not specific to the Davies-Mayer construction, and applicable to, say, any of the PGV-hash functions~\cite{PGV93,BRS02,BRSS2010}.

\subsection{Technical Overview}

In this section, we revisit the approaches from \cite{YZ21}, outlining the barriers in QRPM/QICM and the novel ideas to overcome them. Given the similarities between QRPM and QICM, we focus on QRPM in this overview, deferring the details of QICM to the main body.
%\minki{Probably we need to put QICM in appendix for submission version}

\paragraph{The Lifting Theorem in the (Q)ROM.}
% \minki{Q?}
We begin by reviewing the idea behind the lifting theorem in QROM from~\cite{YZ21}, which builds on the measure-and-reprogram lemma first introduced in~\cite{DFMS19} and later improved in~\cite{DFM20}.  
To make the ideas more intuitive, we focus on the function inversion problem: given a random oracle $H: \{0,1\}^m \to \{0,1\}^n$, the goal is to find an $x$ such that $H(x) = 0^n$.  
Before proceeding, we introduce the notation for a reprogrammed oracle: given a random oracle $H$, we define $H[x \to y]$ as the function that behaves identically to $H$ except that it outputs $y$ on input $x$.  

We will first look at a lemma that reduces a classical $q$-query algorithm to another classical query algorithm with a much smaller number of queries; this will provide some intuition for the lifting theorem.
Consider a \emph{classical} $q$-query algorithm $\cA$. Let $x^*$ be the final outcome of $\cA$ under a random oracle $H$. There are two possibilities: (i) $x^*$ is queried by $\cA$ among one of the $q$ queries; or (ii) $x^*$ is never queried.
Let $i^* \in \{1, 2, \ldots, q, \bot\}$ be the index of that query where $\bot$ indicates that $x^*$ is never queried. 
The key observation in~\cite{YZ21} is that, the algorithm cannot distinguish these two cases for $y^*\in \{0,1\}^n$: 
% \begin{enumerate}
%     \item $\cA$ has oracle access to $G [x^* \to H(x^*)]$ for all these $q$ queries \alex{$G$ has not been defined yet. Also, should we use $H^*$ instead of $G$ to be consistent with the notation in the following paragraphs?};
%     \item $\cA$ has oracle access to $G$ (for any $y$) \alex{is this paranthesis really needed here?} for the first $i^*-1$ queries, then $G [x^* \to H(x^*)]$ for the rest of the queries.
% \end{enumerate}
%\minki{
\begin{enumerate}
    \item $\cA$ has oracle access to $H[x^* \to y^*]$ for all these $q$ queries; %\alex{$y^*$ has not been defined before, right?}\minki{I added it above}
    \item $\cA$ has oracle access to $H$ for the first $i^*-1$ queries, then $H [x^* \to y^*]$ for the rest of the queries.
\end{enumerate}
%}
This is simply because without querying on $x^*$, the algorithm will have identical views on the transcript (and thus the computation). With the above observation, \cite{YZ21} defines the following simulator $S[\cA, H, H^*]$ that only makes one query to $H^*$: $S[\cA, H, H^*]$ randomly guesses a uniform $i^* \gets \{1, 2, \ldots, q, \bot\}$. If $i^* = \bot$, it returns what $\cA^H$ returns. Otherwise, it runs $\cA$ with a random oracle $H$ for the first $i^*-1$ steps, and runs the rest of the computation with $H [x \to H^*(x)]$ where $x$ is the input of the $i^*$-th query.

They show the following \emph{classical} measure-and-reprogram lemma: 
\begin{lemma}[Measure and Reprogram Lemma]\label{lem:rom_measure_and_reprogram_informal}
Let $H$ and $H^*$ be any two functions (not random functions).
Let $\cA$ be an arbitrary classical algorithm equipped with $q$ classical
%\minki{classical?} \alex{yes, I think that should be classical} 
queries to the oracle $H$.

Let $x^* \in X$ be an input and $y^* = H^*(x^*)$. 
Then there exists a simulator algorithm $S$ that given oracle access to $H$, $H^*$, making at most one query to $H^*$, such that for any $\cA$, and for any {output} $z$ (can be arbitrarily dependent on $H, H^*, x^*$), simulates the output of $\cA$ having oracle access to  $H[x^*\rightarrow y^*]$ (the reprogrammed version of $H$) with probability: 
\begin{align*}
    \Pr \left[S[\cA, H, H^*] \text{ outputs } z  \right] 
     \ge
 \frac{1}{(q+1)}
    \Pr \left[\cA^{H[x^*\rightarrow y^*]} \text{ outputs } z \right].
\end{align*} 
\end{lemma}
% \minki{Here there's no definition of $\mathcal G$, and if $\mathcal G$ is inversion problem for $0^n$ the following descriptions is somewhat problematic (eg $y^*=0^n$? there may be some other solution?)}
By taking expectation over $H, H^*$,  and summing over all $z$ such that $z = x^*$ and $H^*(z) = 0^n$, the left-hand side is equal to the success probability of a one-classical-query algorithm $\cB$ finding a pre-image: 
\begin{align*}
    \sum_{x^*} \mathbb{E}_{H, H^*}\left[ \Pr[S[\cA, H, H^*] \to x^*\text{ s.t. } H^*(x^*) = 0^n] \right] = \Pr[\cB \text{ finds a pre-image of }0^n].
\end{align*}
Here we simply treat $S$ as the one query algorithm who simulates $H$ itself and makes a single oracle query to $H^*$.

The right-hand side  is equal to the success probability of a $q$-classical-query algorithm $\cA$ finding a pre-image, since 
\begin{align*}
    \sum_{x^*} \mathbb{E}_{H, H^*} \left[ \Pr[\cA^{H [x^* \to y^*]} \to x^* \text{ s.t. } H^*(x^*) = 0^n] \right] = \Pr[\cA \text{ finds a pre-image of }0^n].
\end{align*}
Using the measure-and-reprogram lemma and the linearity of expectation, we have: 
\begin{align*}
    \Pr[\cB \text{ finds a pre-image of }0^n] \geq \frac{1}{(q+1)} \Pr[\cA \text{ finds a pre-image of }0^n]. 
\end{align*}
Finally, since the winning probability of $\cB$ is at most $O(1/N)$, we establish an upper bound for the success probability of $\cA$ as $O(q/N)$. 

The above inequality is the basic form of the lifting theorem.
The similar idea applies to the quantum setting, as well as a general (interactive) search game. 
% For the sake of simplicity, we will only provide these intuitions and ignore the rest of the details.\minki{The last sentence can be removed, and also it may be nice to clearly say it's lifting theorem.}

\paragraph{The Classical Lifting in the RPM.}
% \minki{How about removing ICM in the overview as we never say anything about it}
The above approach fails even in the classical RPM/ICM setting. To illustrate this issue, consider a simple setting where an algorithm has no oracle access to the inverse of a permutation $\pi: \{0,1\}^n \to \{0,1\}^n$ and its goal is to find an $x$ such that both $x$ and $\pi(x)$ have enough leading zeros (the double-sided search problem). 
For any \emph{classical} $q$-query algorithm $\cA$, similar to the argument in ROM, there must exist some $i^* \in \{1, 2, \ldots, q, \bot\}$ such that the final output $x^*$ is either queried as the $i^*$-th query or never queried at all (when $i^* = \bot$). Following the approach of~\cite{YZ21}, we consider the following simulator $S[\cA, \pi, \pi^*]$:  
\begin{itemize}  
    \item It randomly selects $i^* \gets \{1, 2, \ldots, q, \bot\}$.  
    \item If $i^* = \bot$, it returns whatever $\cA^\pi$ outputs.  
    \item Otherwise, it runs $\cA$ with $\pi$ for the first $i^*-1$ steps, then completes the remaining computation using $\pi [x \to \pi^*(x)]$, where $x$ is the input of the $i^*$-th query.  
\end{itemize}  
% \minki{We may want to clearly say that $\pi [x \to \pi^*(x)]$ is not permutation at some point?}
Ideally, we would like to argue that $S[\cA, \pi, \pi^*]$ behaves similarly to $\cA^{\pi^*}$, up to a multiplicative loss of $(q+1)$, which is coming from guessing a correct $i^*$. However, this no longer holds due to the weak dependence inherent in permutations. Even if $S$ correctly guesses $i^*$, it remains possible that an earlier query (before the $i^*$-th query) returns $\pi^*(x)$; i.e., some input $x'$ under permutation $\pi$ evaluates to $\pi(x') = \pi^*(x)$, which makes $\pi [x \to \pi^*(x)]$ is not even a permutation. In this case, even if $x$ itself is only queried at the $i^*$-th step, the algorithm can still detect an inconsistency --- two distinct queries producing the same output contradicts the structure of a permutation. This issue becomes even more pronounced when access to the inverse oracle is provided, making such inconsistencies easier to detect.  

Our first contribution is to identify this issue as well as provide a solution to enable the classical reprogram lemma in the RPM/ICM. For a forward oracle query, we call a query $x$ ``hit'' (or just $x^{\sf hit}$) with respect to the final outcome $x^*$ if $x = x^*$, just as in the random oracle case; we call a query $x$ ``miss'' (or just $x^{\sf miss}$)  with respect to $x^*$ if $\pi(x) = \pi^*(x^*)$. When we reprogram a permutation $\pi$ with $[x^* \to \pi^*(x^*)]$, we will maintain its injective structure. We define the reprogrammed permutation as:
\begin{align*}
    \pi [x^* \to \pi^*(x^*)] (x) = \begin{cases}
        \pi^*(x^*) & \text{ if } x = x^{\sf hit} \\
        \pi(x^*) & \text{ if } x = x^{\sf miss} \\
        \pi(x) & \text{ otherwise}
    \end{cases}.
\end{align*}
In other words, when we only hardcode $[x^* \to \pi^*{(x^*)}]$, it violates the injective structure of the permutation. Thus, we will have to find the element with more than one preimage and reprogram that as well.\footnote{The idea of reprogramming a permutation by swapping two outputs is also used in \cite{alagic2022post}.} This corresponds to \Cref{fig:miss-and-hit}, where to maintain the injectiveness, we have to reprogram both $x^{\sf hit}$ and $x^{\sf miss}$. This corresponds to removing the two solid edges $x^{\sf hit}$ to $\pi^*(x^*)$ and $x^{\sf miss}$ to $\pi(x^*)$ in ~\Cref{fig:miss-and-hit} and adding two edges as in \Cref{fig:miss-and-hit-after}.

\begin{figure}[h]
    \centering
\begin{tikzpicture}[node distance=2cm, thick, ->]
    % Nodes
    \node (A) {$\pi^{-1}(\pi^*(x^*))=x^{\sf miss}$};
    \node (B) [right=of A] {$y^* = y^{\sf hit} = \pi^*(x^{*})$};
    \node (C) [above=of B] {$x^{*} = x^{\sf hit}$};
    \node (D) [right=of C] {$\pi(x^*) = y^{\sf miss}$};
    % Arrows
    \draw (A) -- node[midway, above] {$\pi$} (B);
    \draw[dashed] (C) -- node[midway, right] {$\pi^*$} (B);
    \draw (C) -- node[midway, above] {$\pi$} (D);
    % \draw[dashed] (A) -- (D);
\end{tikzpicture}
    \caption{An illustration of hit and miss inputs regarding $x^*, y^*$}
    \label{fig:miss-and-hit}
\end{figure}
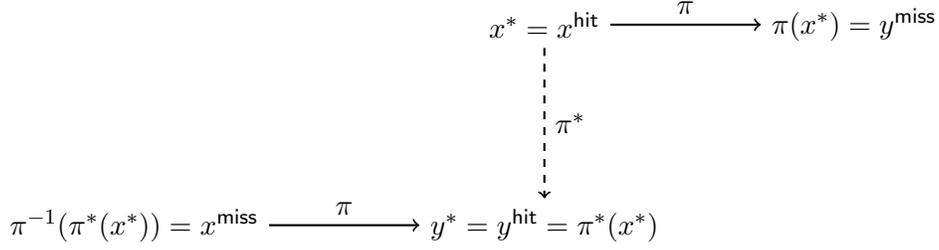

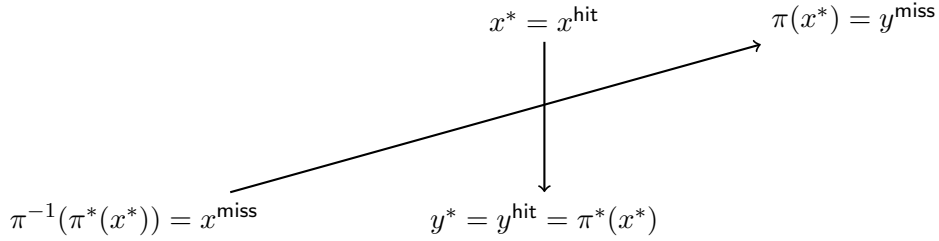
\begin{figure}[h]
    \centering
\begin{tikzpicture}[node distance=2cm, thick, ->]
    % Nodes
    \node (A) {$\pi^{-1}(\pi^*(x^*))=x^{\sf miss}$};
    \node (B) [right=of A] {$y^* = y^{\sf hit} = \pi^*(x^{*})$};
    \node (C) [above=of B] {$x^{*} = x^{\sf hit}$};
    \node (D) [right=of C] {$\pi(x^*) = y^{\sf miss}$};
    % Arrows
    %\draw (A) -- node[midway, above] {$\pi$} (B);
    \draw (C) -- node[midway, right] {} (B);
    %\draw (C) -- node[midway, above] {$\pi$} (D);
    \draw (A) -- (D);
\end{tikzpicture}
    \caption{An illustration of mappings in the reprogrammed permutation $\pi [x^* \to \pi^*(x^*)]$}
    \label{fig:miss-and-hit-after}
\end{figure}

Similarly, for a backward query to $\pi^{-1}$, we call a query $y$ ``hit'' (or just $y^{\sf hit}$) with respect to the final outcome $y^* = \pi^*(x^*)$ if $y^{\sf hit} = y^*$. We call a query $y$ ``miss'' (or just $y^{\sf miss}$) if $y = \pi(x^*)$ as in~\Cref{fig:miss-and-hit}.

Giving the above definition, we consider the following simulator that not only guesses the index $i^*$ but also whether the $i^*$-th query is a hit or a miss query. $S[\cA, \pi, \pi^*]$ is defined as:
\begin{enumerate}
    \item $S$ samples $(i^*, b) \gets (\{1,2,\ldots,q\} \times \{0,1\}) \cup \{(\bot, \bot)\}$
    \item If $i^* = \bot$, it returns whatever $\cA^\pi$ outputs;
    \item Otherwise, it runs $\cA$ using $\pi$ (both the forward and the inverse) for the first $i^*-1$ queries; for the $i^*$-th query:
    \begin{itemize}
        \item If $b = 0$ (indicating a hit query):
            \begin{itemize}
                \item A forward query on input $x$: compute $y = \pi^*(x)$ and reprogram $\pi [x \to y]$.
                \item A backward query on input $y$: compute $x = {\pi^*}^{-1}(y)$ and reprogram $\pi [x \to y]$.
            \end{itemize}
        \item If $b = 1$ (indicating a miss query):
        \begin{itemize}
                \item A forward query on input $x$: compute $y = \pi({\pi^*}^{-1}(\pi(x)))$ and reprogram $\pi [x \to y]$.
                \item A backward query on input $y$: compute $x = \pi^{-1}({\pi^*}(\pi^{-1}(y)))$ and reprogram $\pi [x \to y]$.
            \end{itemize}
    \end{itemize}
    Then it runs the rest of $\cA$'s computation under $\pi[x \to y]$ (and its inverse).
\end{enumerate}

We show that, by defining both hit and miss queries, it captures the first query that ``touches'' the final outcome $x^*$ --- either by querying the miss query or the hit query. 
% \alex{not sure I fully understood the previous sentence and also ''touch'' has not been explained yet, right?}\minki{Yeah it is not explained; probably ``any of $x^*,y^*,x^{\sf hit},$ or $y^{\sf hit}$''?}
This modified simulator allows us to establish a similar reprogram lemma in the classical RPM. For example, we can show that for any $q$-query $\cA$ for the double-sided search problem, there always exists a one-query $\cB$ such that
\begin{align*}
    \Pr[\cB \text{ wins}] \geq \frac{1}{(2q+1)} \Pr[\cA \text{ wins}]. 
\end{align*}
This can be further generalized to any search game, for which we will discuss the quantum setting.

\paragraph{State Decomposition.}

With the above idea in the classical setting, we are now ready to generalize it to the quantum setting and give a lifting theorem for QRPM. 
To explain the high level idea for our measure-and-reprogram lemma in the QRPM setting, we start with the following example.
Let $\cA$ be a $q$-quantum-query algorithm that solves the double-sided search game, and let $x^*, y^*$ be some pair. Consider $\cA$ with oracle access to $\pi [x^* \to y^*]$ and its inverse, whose computation can be written as:
\begin{align*}
    U_{q+1} \, O_{\pi [x^* \to y^*]} \, U_q \, \cdots O_{(\pi [x^* \to y^*])^{-1}} \, U_2 \, O_{\pi [x^* \to y^*]} \, U_1 \, \ket 0.
\end{align*}
% \alex{should we use $T$ or $q$ for the number of queries? I saw later we have both.}
Without loss of generality, we assume the first query is a forward query. We start by considering the state up to the first query: $O_{\pi[{x}^* \rightarrow {y}^*]} U_1 \ket 0$. We insert an additional identity operator and have,
% \begin{align*}
%     O_{\pi[{x}^* \rightarrow {y}^*]} U_1 \ket 0 & = O_{\pi[{x}^*\rightarrow {y}^*]} \, {I} \, U_1 \, \ket 0 \\
%      &= O_{\pi[{x}^* \rightarrow {y}^*]} \left(I -\ketbra {{x^*}^{\sf hit}}{{x^*}^{\sf hit}} +  \ketbra {{x^*}^{\sf hit}}{{x^*}^{\sf hit}} 
%      -  \ketbra {{x^*}^{\sf miss}}{{x^*}^{\sf miss}}  +  \ketbra {{x^*}^{\sf miss}}{{x^*}^{\sf miss}} 
%      \right) U_1 \ket 0 \\
%      &= \underbrace{O_{\pi[{x}^* \rightarrow {y}^*]} \left(I - \ketbra {{x^*}^{\sf hit}}{{x^*}^{\sf hit}} - \ketbra {{x^*}^{\sf miss}}{{x^*}^{\sf miss}}\right) U_1 \ket 0}_{(i)} + \underbrace{O_{\pi[{x}^* \rightarrow {y}^*]} \ketbra {{x^*}^{\sf hit}}{{x^*}^{\sf hit}} U_1 \ket 0}_{(ii)} \\
%      &+ \underbrace{O_{\pi[{x}^* \rightarrow {y}^*]} \ketbra {{x^*}^{\sf miss}}{{x^*}^{\sf miss}} U_1 \ket 0}_{(iii)}.
% \end{align*}
\begin{align*}
    &O_{\pi[{x}^* \rightarrow {y}^*]} U_1 \ket 0 = O_{\pi[{x}^*\rightarrow {y}^*]} \, {I} \, U_1 \, \ket 0 \\
     &= O_{\pi[{x}^* \rightarrow {y}^*]} \left(I -\ketbra {{x^*}^{\sf hit}}{{x^*}^{\sf hit}} +  \ketbra {{x^*}^{\sf hit}}{{x^*}^{\sf hit}} 
     -  \ketbra {{x^*}^{\sf miss}}{{x^*}^{\sf miss}}  +  \ketbra {{x^*}^{\sf miss}}{{x^*}^{\sf miss}} 
     \right) U_1 \ket 0 \\
     &= \underbrace{O_{\pi[{x}^* \rightarrow {y}^*]} \left(I - \ketbra {{x^*}^{\sf hit}}{{x^*}^{\sf hit}} - \ketbra {{x^*}^{\sf miss}}{{x^*}^{\sf miss}}\right) U_1 \ket 0}_{(i)} + \underbrace{O_{\pi[{x}^* \rightarrow {y}^*]} \ketbra {{x^*}^{\sf hit}}{{x^*}^{\sf hit}} U_1 \ket 0}_{(ii)} \\
     &+ \underbrace{O_{\pi[{x}^* \rightarrow {y}^*]} \ketbra {{x^*}^{\sf miss}}{{x^*}^{\sf miss}} U_1 \ket 0}_{(iii)}.
\end{align*}
Here ${x^*}^{\sf hit}$ and ${x^*}^{\sf miss}$ are defined according to $\pi$ as in \Cref{fig:miss-and-hit}.

The first term $(i)$ equals to
\begin{align*}
    & O_{\pi[{x}^* \rightarrow {y}^*]} \left(I - \ketbra {{x^*}^{\sf hit}}{{x^*}^{\sf hit}} - \ketbra {{x^*}^{\sf miss}}{{x^*}^{\sf miss}}\right) U_1 \ket 0 \\
    = & O_{\pi} \left(I - \ketbra {{x^*}^{\sf hit}}{{x^*}^{\sf hit}} - \ketbra {{x^*}^{\sf miss}}{{x^*}^{\sf miss}}\right) U_1 \ket 0 \\
    = & O_{\pi} U_1 \ket 0 - O_{\pi} \ketbra {{x^*}^{\sf hit}}{{x^*}^{\sf hit}} U_1 \ket 0   - O_{\pi} \ketbra {{x^*}^{\sf miss}}{{x^*}^{\sf miss}} U_1 \ket 0. 
\end{align*}
This is because on inputs that are neither hit nor miss inputs, $O_\pi$ and $O_{\pi [x^* \to y^*]}$ are identical. Combining with other terms, we have that the state after the first query is equal to:
\begin{align*}
     O_{\pi[{x}^* \rightarrow {y}^*]} U_1 \ket 0 & = O_{\pi} U_1 \ket 0 - O_{\pi} \ketbra {{x^*}^{\sf hit}}{{x^*}^{\sf hit}} U_1 \ket 0   - O_{\pi} \ketbra {{x^*}^{\sf miss}}{{x^*}^{\sf miss}} U_1 \ket 0 \\ 
     & + O_{\pi[{x}^* \rightarrow {y}^*]} \ketbra {{x^*}^{\sf hit}}{{x^*}^{\sf hit}} U_1 \ket 0  + O_{\pi[{x}^* \rightarrow {y}^*]} \ketbra {{x^*}^{\sf miss}}{{x^*}^{\sf miss}} U_1 \ket 0. 
\end{align*}
These five terms can be interpreted in the following way:
\begin{enumerate}
    \item $O_{\pi} U_1 \ket 0$: the first term corresponds to the case that we do not measure the first query and use $\pi$ for this query.
    \item $O_{\pi} \ketbra {{x^*}^{\sf hit}}{{x^*}^{\sf hit}} U_1 \ket 0$ and $O_{\pi} \ketbra {{x^*}^{\sf miss}}{{x^*}^{\sf miss}} U_1 \ket 0$: both terms correspond to the case that we measure the first query and it is either a hit or miss query; but we still use $\pi$ for the first query. 
    Looking ahead, these two terms will correspond to $\ket{\phi_{1,0,1}}$ and $\ket{\phi_{1,1,1}}$ in the final decomposition defined below. 
    %\takashi{They were written as $\ket{\phi_{1,0,0}}$ and $\ket{\phi_{1,1,0}}$, but I believe they should be $\ket{\phi_{1,0,1}}$ and $\ket{\phi_{1,1,1}}$ because $c=1$ means that we answers the query first and then reprogram.}
    % \minki{It is nice to give the correspondence to $\ket{\phi_{i,b,c}}$? probably ``''}
    \item $O_{\pi[{x}^* \rightarrow {y}^*]} \ketbra {{x^*}^{\sf hit}}{{x^*}^{\sf hit}} U_1 \ket 0$ and $O_{\pi[{x}^* \rightarrow {y}^*]} \ketbra {{x^*}^{\sf miss}}{{x^*}^{\sf miss}} U_1 \ket 0$: both terms correspond to the case that we measure the first query and it is either a hit or miss query; we use the reprogrammed $\pi [x^* \to y^*]$ for the first query.
    Looking ahead, these two terms will correspond to $\ket{\phi_{1,0,0}}$ and $\ket{\phi_{1,1,0}}$ in the final decomposition defined below. 
    %\takashi{They were written as $\ket{\phi_{1,0,1}}$ and $\ket{\phi_{1,1,1}}$, but I believe they should be $\ket{\phi_{1,0,0}}$ and $\ket{\phi_{1,1,0}}$.}
\end{enumerate}

The similar argument extends to the second query, where we will decompose the component $O_\pi U_1 \ket 0$. Similar to the first query case, this decomposition introduces four more terms. By  decomposing the state to the last query, eventually we will have $(4q+1)$ terms. Among them:
\begin{itemize}
    \item There is one term $\ket {\phi_\bot}$ for the case that we do not measure any query and run $\cA$ under $\pi$.
    \item There are $q$ terms $\ket{\phi_{i,0, 0}}$ for the case that we measure the $i$-th query for $i \in \{1, 2, \ldots, q\}$, it is ${x^*}^{\sf hit}$ (or ${y^*}^{\sf hit}$, if it is a backward query). The remaining queries (including the $i$-th query) are under $\pi [x^* \to y^*]$.
    \item There are $q$ terms $\ket{\phi_{i, 1, 0}}$ for the case that we measure the $i$-th query for $i \in \{1, 2, \ldots, q\}$, it is ${x^*}^{\sf miss}$ (or ${y^*}^{\sf miss}$, if it is a backward query). The remaining queries (including the $i$-th query) are under $\pi [x^* \to y^*]$.
    \item Similarly, we have $2q$ terms $\ket {\phi_{i, b, 1}}$ for $b \in \{0, 1\}$. They stand for the case that we measure the $i$-th query for $i \in \{1, 2, \ldots, q\}$, it is either a miss or hit (depending on $b$). The $i$-th query is under $\pi$ and the remaining queries are under $\pi [x^* \to y^*]$.
\end{itemize}

By induction on the state decomposition, we can show that the original computation is equal to the summation of all these $4q+1$ terms:
\begin{align*}
    U_{q+1} \, O_{\pi [x^* \to y^*]} \, U_q \, \cdots O_{(\pi [x^* \to y^*])^{-1}} \, U_2 \, O_{\pi [x^* \to y^*]} \, U_1 \, \ket 0 = \ket{\phi_\bot} +  \sum_{i=1}^q \sum_{b, c \in\{0,1\}} (-1)^c \ket{\phi_{i, b, c}}.
\end{align*}
Thus for a projection $\Pi$ (representing the winning condition), from the above identity and Cauchy-Schwarz, we have
\begin{align*}
    & \|\Pi \, U_{q+1} \, O_{\pi [x^* \to y^*]} \, U_q \, \cdots \, O_{\pi [x^* \to y^*]} \, U_1 \, \ket 0\|^2 \\
    \leq & (4q+1) \cdot \left( \| \Pi \ket{\phi_\bot} \|^2 +  \sum_{i=1}^q \sum_{b, c \in\{0,1\}} \|\Pi \ket{\phi_{i, b, c}} \|^2\right).
\end{align*}
Here each term $\| \Pi \ket{\phi_\bot} \|^2$ or $\|\Pi \ket{\phi_{i, b, c}} \|^2$ has an operational meaning: the probability that the final outcome is in $\Pi$, where the execution will use oracles specified by $\bot$ or $i, b, c$.
We give a formal description of the simulator below.

\paragraph{The Quantum Lifting Theorem in the QRPM.}

% We have 
% \begin{align*}
%      & O_{(\pi[x^*\to y^*])^{-1}} U_2 O_{\pi[{x}^* \rightarrow {y}^*]} U_1 \ket 0 \\
%      & = O_{(\pi[x^*\to y^*])^{-1}} U_2 O_{\pi} U_1 \ket 0 \\
%      & - O_{(\pi[x^*\to y^*])^{-1}} U_2 O_{\pi} \ketbra {{x^*}^{\sf hit}}{{x^*}^{\sf hit}} U_1 \ket 0   -O_{(\pi[x^*\to y^*])^{-1}} U_2 O_{\pi} \ketbra {{x^*}^{\sf miss}}{{x^*}^{\sf miss}} U_1 \ket 0 \\ 
%      & + O_{(\pi[x^*\to y^*])^{-1}} U_2 O_{\pi[{x}^* \rightarrow {y}^*]} \ketbra {{x^*}^{\sf hit}}{{x^*}^{\sf hit}} U_1 \ket 0  + O_{(\pi[x^*\to y^*])^{-1}} U_2 O_{\pi[{x}^* \rightarrow {y}^*]} \ketbra {{x^*}^{\sf miss}}{{x^*}^{\sf miss}} U_1 \ket 0. 
% \end{align*}\alex{resulting in a measurement outcome $x$}

We define the following simulator in the quantum setting for RPM. $S[\cA, \pi, \pi^*]$ is defined as:
\begin{enumerate}
    \item $S$ samples $(i, b, c) \gets (\{1,2,\ldots,q\} \times \{0,1\} \times \{0,1\}) \cup \{(\bot, \bot, \bot)\}$, where $b$ stands for if it is ``hit'' or ``miss'' and $c$ stands for whether the reprogramming happens before or after the query.
    \item If $i = \bot$, it returns whatever $\cA^\pi$ outputs;
    \item Otherwise, it runs $\cA$ using $\pi$ (both the forward and the inverse) for the first $i-1$ queries; for the $i$-th query, $S$ measures the input register :
    \begin{itemize}
        \item If $b = 0$ (indicating a hit query):
            \begin{itemize}
                \item A forward query on input $x$: compute $y = \pi^*(x)$.
                \item A backward query on input $y$: compute $x = {\pi^*}^{-1}(y)$.
            \end{itemize}
        \item If $b = 1$ (indicating a miss query):
        \begin{itemize}
                \item A forward query on input $x$: compute $y = \pi({\pi^*}^{-1}(\pi(x)))$.
                \item A backward query on input $y$: compute $x = \pi^{-1}({\pi^*}(\pi^{-1}(y)))$.
            \end{itemize}
    \end{itemize}
    Then for the remaining queries,
    \begin{itemize}
        \item If $c = 0$, it answers all $\cA$'s remaining queries using $\pi [x \to y]$.
        \item If $c = 1$, it answers $\cA$'s $i$-th query using $\pi$ and the remaining queries using $\pi [x\to y]$.
    \end{itemize}
    $S$ outputs whatever $\cA$ outputs.
\end{enumerate}
If the first step samples $(i,b,c)$, then the probability that the simulator produces $z$ is
at least 
%\takashi{I added "at least" since it may output $z$ even if the guess is incorrect.}
$\|\Pi \ket{\phi_{i, b, c}} \|^2$ (or $\|\Pi \ket{\phi_{\bot}} \|^2$ if $(i,b,c)=(\bot,\bot,\bot)$).
Building on the previous discussion, we conclude that for any outcome \( z \), the probability of the simulator producing \( z \) is at least \( \frac{1}{(4q+1)^2} \) times the probability of \( \mathcal{A} \) producing \( z \).

This simulator can be easily generalized to any number of reprogramming.
In the simulator, we will choose $k$ coordinates instead of one to measure and reprogram, one for each final output $x^*_i$; similarly to the $k=1$ case, there are $(4q+1)$ such possibilities for each $x^*_i$. 
One subtlety is that, when we sequentially reprogram a permutation multiple times, these reprogramming may interfere with each other, which makes the analysis more involved. To avoid this, we introduce the "goodness" condition, which ensures that such interference does not occur.    
(See \Cref{def:good} for the formal definition of the goodness). Fortunately, we show that the goodness condition is satisfied except for an exponentially small probability.

Now we propose the measure-and-reprogram lemma.
In order to describe the formal measure-and-reprogram result, we need to additionally introduce the following notions. We will denote the reprogrammed permutation on $k$ pairs $p_1 = (x_1, y_1)$, ..., $p_k = (x_k, y_k)$ by $\pi[x_1\rightarrow y_1]...[x_k\rightarrow y_k]$.

\begin{lemma}[Measure and Reprogram Lemma, Informal]\label{lem:measure_and_reprogram_informal}
Let $\pi$ and $\pi^*$ be two fixed permutations (not random permutations).
Let $\cA$ be an arbitrary quantum algorithm equipped with $q$ quantum queries to the oracle $\pi, \pi^{-1}$.

Let $\vec{x}^* = (x_1^*, ..., x_k^*) \in X^k$ be any $k$-vector of inputs and $\vec{y}^* = (y_1^*, ..., y_k^*) = (\pi^*(x_1^*), ..., \pi^*(x_k^*))$,
such that the $k$-tuple $(x_1^*, y_1^*), ..., (x_k^*, y_k^*)$ is ``good''\footnote{For simplicity, we do not state the formal definition of ``good'' in the introduction.} with respect to $\pi$. 
Then there exists a simulator algorithm $S$ that given oracle access to $\pi$, $\pi^*$ and $\cA$, for any $z$ (can be arbitrarily dependent on $\pi, \pi^*, \vec{x}^*$), simulates the output of $\cA$ having oracle access to  $\pi[x^*_1\rightarrow y^*_1]...[x^*_k\rightarrow y^*_k]$ (the  reprogrammed version of $\pi$) with probability: 
\begin{align*}
    \Pr_{\pi, \pi^*} \left[S[\cA, \pi, \pi^*] \text{ outputs } z  \right] 
     \ge
 \frac{1}{(8q+1)^{2k}}
    \Pr_{\pi, \pi^*} \left[\cA^{\pi[x^*_1\rightarrow y^*_1]...[x^*_k\rightarrow y^*_k]} \text{ outputs } z \right].
\end{align*} 
\end{lemma}
Here, we have \( (8q+1) \) instead of \( (4q+1) \) because we consider the most general algorithm, which can make superposition queries to \( \pi \) and \( \pi^{-1} \) within a single query. To handle this, we decompose each query into two: one querying only \( \pi \) and the other querying only \( \pi^{-1} \). This transformation introduces an additional multiplicative factor of \( 2 \).

Finally, for any game $\mathcal{G}$, by summing over all valid $\vec{x}^*, \pi(\vec{x}^*) \in X^k \times X^k$, when take expectation over $\pi, \pi^*$, we have for any $q$-quantum-query $\cA$, there exists a $k$-classical-query $\cB$ such that 
\begin{align*}
    \Pr[\cB \text{ wins } \mathcal{G}] \geq \frac{(1 - k^2/|X|)}{(8q+1)^{2k}}\Pr[\cA \text{ wins } \mathcal{G}].
\end{align*}
Here $(1 - k^2/|X|)$ %comes from the probability that a uniform $\pi$ is good for any $(\vec{x}^*, \vec{y}^*) \in X^k \times X^k$. 
comes from the probability that the goodness condition holds for uniform $\pi,\pi^*$. 
This completes the high level idea of our lifting theorem in the QRPM. 

\subsection{Concurrent Work}
A concurrent work by Alagic, Carolan, Majenz, and Tokat~\cite{cryptoeprint:2025/731} establishes quantum indifferentiability for the (multi-round) sponge construction. Although their bounds are still not tight, their results encompass ours in most settings, with a few exceptions, such as preimage-resistance in the single- and two-round cases, and collision-resistance in the single-round case. However, their approach is tailored specifically to sponge, whereas our lifting theorem applies more broadly to any permutation-based construction.

\subsection{Paper Organization}
In \Cref{sec:preparation} we introduce a series of intermediate results that are going to be used to prove the main lifting theorems. Part of the proofs of these results are deferred to \Cref{app:deferred}. \Cref{sec:classical_lifting} contains the classical lifting for permutations as a classical analogue of our main quantum lifting theorem, shown in \Cref{sec:quantum_lifting},
 while the extension of the quantum lifting to the interactive setting is proven in \Cref{app:interactive}.
The quantum lifting theorem in the ideal cipher model is shown in \Cref{app:ideal_cipher}. Finally, the applications of our quantum lifting theorems can be found in \Cref{app:applications}.

% \subsection{Related Work}
% \qipeng{Do we still need related work? Since we discuss other techniques for QRPM/QICM in the intro and \cite{YZ21} in the technical overview.}
% \alex{I agree, we probably dont need this anymore}

%\section{Preliminaries}
%\takashi{IS there any actual preliminaries? If no, we may remove this section.}

\section{Preparation for Lifting Theorem} \label{sec:preparation}

In this section, we introduce notations, definitions and easy lemmas, that are used in the proofs of the (classical and quantum) lifting theorem for permutations in \Cref{sec:classical_lifting,sec:quantum_lifting}. 
%\takashi{I made a new section since this is a new content and not a usual preliminaries.}
\subsection{Algorithms with Permutation Oracles}
\label{sec:prelim_algo_permutation}
We define basic notations for algorithms with oracle access to an invertible permutation. 

Let $\pi: X \to X$ be a permutation and $\pi^{-1}$ be its inverse.  
%\takashi{Is there any particular reason why we are using $T$ for the number of queries? If no, shall we use $q$ similarly to the other parts of the paper?}\alex{switched to $q$}
A $q$-query classical algorithm with permutation oracles $\pi$ can query both $\pi$ and $\pi^{-1}$, but in total $q$ times. A $q$-query quantum algorithm can query the following unitary in total $q$ times:
\begin{align*}
    U_\pi \ket b \ket x \ket y = \begin{cases}
        \ket b \otimes O_\pi \left(\ket x \ket {y}\right) & \text{ if } b = 0 \\
        \ket b \otimes O_{\pi^{-1}} \left(\ket x \ket {y} \right) & \text{ if } b = 1,
    \end{cases} 
\end{align*}
where $O_\pi, O_{\pi^{-1}}$ are the coherent computation for $\pi$ and $\pi^{-1}$, 
%\minki{
% that is, $O_\pi\ket{x}\ket{y} = \ket{x}\ket{y+\pi(x)}$ and $O_{\pi^{-1}}\ket{x}\ket{y}=\ket{x}\ket{y+\pi^{-1}(x)}$.
%}
%\alex{
We will typically denote a quantum (or classical) query algorithm by $\cA$. By $\cA^{\pi}$ we mean that $\cA$ has quantum (or classical) access to the permutation $\pi$, as well as to its inverse, $\pi^{-1}$.
%}
%\minki{We just do not mention anything about the controlled access, which can be done in a black-box way at some multiplicate lose in the number of queries. However, if we say about this some careful reader complain about ignoring that loss in the main results. So, it may be better not to state about it in the paper. If some of you do want to say something about the controlled access, then let us think how to put them later.}

\begin{lemma}[Normal form]
\label{lem:quantum_algo}
    Let $\As$ be a quantum algorithm making at most $q$ quantum queries to a permutation $\pi$ (i.e., $U_\pi$). There always exists a quantum algorithm whose output is identical to that of $\As$, making at most $2q$ quantum queries to $O_\pi$ and $O_{\pi^{-1}}$, and of the following normal form:
    \begin{align*}
        O_{\pi^{-1}} U_{2q} O_{\pi} U_{2q-1} \cdots O_{\pi^{-1}} U_2 O_{\pi} U_1 \ket 0;
    \end{align*}
    i.e., making exactly $q$ queries to $O_\pi$ (for odd-numbered queries) and $q$ queries to $O_{\pi^{-1}}$ (for even-numbered queries).
\end{lemma}
\begin{proof}
    Every oracle access to $U_\pi$ can be replaced with one query access to $O_\pi$ and one query access to $O_{\pi^{-1}}$ (by introducing dummy queries). 
\end{proof}

% \begin{definition}
% For a permutation $\pi:X\rightarrow X$ and $\vec{x}=(x_1,...,x_k)\in X^k$, we define  $\pi(\vec{x})=(\pi(x_1),...,\pi(x_k))$. \takashi{I added this notation, which I think is quite useful. I'm not sure if this is the right place. Maybe we can write it in the preliminaries.}
% \end{definition}

\subsection{Reprogramming of Permutations}
Our lifting theorems are established using simulators that gradually reprogram a permutation given to an algorithm. We introduce notation for reprogramming of permutations and show their basic properties.

\begin{definition}
For a permutation $\pi:X\rightarrow X$ and $\vec{x}=(x_1,...,x_k)\in X^k$, we define  $\pi(\vec{x})=(\pi(x_1),...,\pi(x_k))$. 
\end{definition}

\begin{definition}[Reprogramming permutations]
    \label{def:perm_reprog}
    Let a permutation $\pi : X \rightarrow X$ and $(x, y) \in X \times X$ be an arbitrary pair. Then we denote the reprogramming of $\pi$ by $(x, y)$ as:
    \begin{equation}
        \pi[x \rightarrow y](z) = 
        \begin{cases}
                     y   & \text{ if } z = x\\
                    \pi(x) & \text{ if } z = \pi^{-1}(y) \\
                    \pi(z)       & \text{ if } z \not\in \{x, \pi^{-1}(y)\}
        \end{cases}
    \end{equation}
    If we denote $(x, y)$ by a pair $p$, we also use $\pi[p]$ for $\pi[x \to y]$.
    %We will also define $\pi^{-1}[x^* \to y^*]$ as $(\pi [x^*\to y^*])^{-1}$.
    
    % \aaram{How about simply making a lemma showing that $\pi[x\to y]^{-1}=\pi^{-1}[y\to x]$, instead of $\pi^{-1}[x\to y]$, like the following Lemma~\ref{lem:reprogramming_inverse}?}
    % \takashi{I removed the sentence "we will also define $\pi^{-1}[x^* \to y^*]$ as $(\pi [x^*\to y^*])^{-1}$." since this looks weird.}
    Similarly, for $k$ pairs $p_1 = (x_1, y_1), ..., p_k = (x_k, y_k)$ we can define the reprogramming of $\pi$ by $p_1, ..., p_k$, denoted by $\pi[x_1 \rightarrow y_1]...[x_k \rightarrow y_k]$ (or $\pi [p_1] \ldots [p_k]$) in a recursive manner by $\pi[x_1 \rightarrow y_1]...[x_k \rightarrow y_k]:=(\pi[x_1 \rightarrow y_1]...[x_{k-1}\to y_{k-1}])[x_k \rightarrow y_k]$.
    We often denote $\pi [x_1 \to y_1] \ldots [x_k \to y_k]$ by $\pi [\vec{x} \to \vec{y}]$ for brevity. 
    % \takashi{I moved this notation to here since this can be defined without the commutativity.}
\end{definition}

\begin{lemma}\label{lem:reprogramming_inverse}
% \minki{I don't understand the notations here. In $\pi[x_1\to y_1]\dots[x_k\to y_k]^{-1}$, where the inversion is applied; for example, if $k=1$, isn't it say that $\pi[x_1\to y_1]^-1 = \pi^{-1}[y_1\to x_1] = (\pi [y_1\to x_1])^{-1}$ according to the above definition? May be the above definition is wrong, so that the correct one may be $\pi^{-1}[y^*\to x^*]:= (\pi [x^*\to y^*])^{-1}$?}\takashi{I believe this concern is resolved since I removed the above definition (which I don't think is needed.)}
    For any permutation $\pi$ and pairs $(x_1, y_1), \dots, (x_k, y_k)$, we have
    \[
    (\pi[x_1\to y_1]\dots[x_k\to y_k])^{-1}=\pi^{-1}[y_1\to x_1]\dots [y_k\to x_k].
    \]
\end{lemma}
\begin{proof}
    The proof is trivial: it suffices to check the statement for the case $k=1$, which we skip.
\end{proof}

\begin{definition}[Disjoint pairs]\label{def:disjoint}
    For $k$ pairs $p_1 = (x_1, y_1), ..., p_k = (x_k, y_k)$, we say they are disjoint pairs (or simply disjoint) if there is no duplicated $x$ or $y$ entries: for any $i<j$, we have $x_i\neq x_j$ and $y_i\neq y_j$. 
    %\minki{It means $x_i\neq x_j$ and $y_i\neq y_j$, not about $x_i\neq y_j$, right?}\aaram{Yes, good point.  I added a clarification.}
    We say $\vec{x}, \vec{y} \in X^k \times X^k$ are disjoint if $(x_1, y_1), \ldots, (x_k, y_k)$ are disjoint.
    
    Similarly, we say $\vec{x}$ (or $\vec{y}$) is disjoint if there are no duplicated entries.  
\end{definition}

\begin{lemma}[Commutativity of reprogramming for disjoint pairs]\label{lem:commutativity_disjoint_pairs}
    For every permutation $\pi: X \to X$, disjoint $\vec{x}, \vec{y} \in X^k \times X^k$, and any permutation $\sigma: [k] \to [k]$, we have:
    \begin{align*}
        \pi [x_1 \to y_1] \ldots [x_k \to y_k] = \pi [x_{\sigma(1)} \to y_{\sigma(1)}] \ldots [x_{\sigma(k)} \to y_{\sigma(k)}].
    \end{align*} 
    %\takashi{I don't know if we need this lemma.}\aaram{I'm not sure if absolutely necessary, but certainly useful and good to know.  For example, Lemma~\ref{lem:reprogram_good_tuples} can be proved I think without relying on this, but this lemma makes the proof simpler.} \takashi{Yes, I agree that this is useful. I just wanted to know if we are (even implicitly) using this lemma at any point of the main body.} \alex{Should I then move it in the Appendix?}\takashi{I think we can keep it here. I'm citing this lemma in a footnote of the proof of the quantum case.}
\end{lemma}

The proof of this will be given in the Appendix, page~\pageref{lem:commutativity_disjoint_pairs_proof}.

\if0
Based on the above lemma, we give the following definition.
\begin{definition}
    For every permutation $\pi: X \to X$ and any disjoint $\vec{x}, \vec{y} \in X^k \times X^k$, we denote $\pi [x_1 \to y_1] \ldots [x_k \to y_k]$ by $\pi [\vec{x} \to \vec{y}]$.
\end{definition}
\fi

%\aaram{Not very important thing, but personally I don't really like this definition: since $\vec{x}$ already suggests an `order' of entries, $\pi[\vec{x}\to\vec{y}]$ can be defined without this Lemma 1.4: just apply the reprogramming following the given order.  The fact that the notation is independent of the order is given by Lemma 1.4, of course.}

%\aaram{If we consider the set-theoretic definition of a function $f$, which is just a set of pairs $(x,y)$ where $y=f(x)$, then, when $p_1, \dots, p_k$ are disjoint according to Definition 1.3, then $p=\{p_1, p_2, \dots, p_k\}$ defines a \emph{partial permutation}: a permutation where some $p(x)$ can be undefined.  Lemma 1.4 allows us to use a notation like $\pi[p]$, for a permutation $\pi$ and a partial permutation $p$.  But I'm not sure if this notation could be useful for us or not.}

\subsubsection*{Good Tuples of Permutations}
In our lifting theorems, we consider a simulator that gradually reprograms a permutation $\pi$ according to disjoint pairs $p_1^* = (x_1^*, y_1^*), \ldots, p_k^* = (x_k^*, y_k^*)$. In its analysis, it is useful to ensure that, 
%$x_j^*$ is mapped to $y_j^*$ for all $j \in [n]$ after reprogramming by $p^*_j$ is performed, even if further reprogramming is applied afterward.  
if one reprogramming changes the  value of the permutation on some input, no subsequent reprogramming changes that value. 
%\takashi{I changed the intuition for goodness, (the previous one was inappropriate since that is always true even in the bad case.)}   
To guarantee this property, we define the \emph{goodness condition} for $(p_1^*, \dots, p^*_k)$ as follows:

\begin{definition}[Good tuples]
\label{def:good}
We say that a $k$-tuple of pairs $(p_1^*,...,p^*_k)$, where $p_1^* = (x_1^*, y_1^*), \ldots, p_k^* = (x_k^*, y_k^*)$, is good w.r.t.\ $\pi$ if the following hold:
\begin{itemize}
\item $p_1^*, \ldots, p_k^*$ are disjoint~(\Cref{def:disjoint}), and
\item $\pi(x^*_i)\ne y^*_j$ for any $i,j \in [k]$. 
% \takashi{Changed the definition.}
\end{itemize} 
\end{definition}

\begin{definition}[Good pairs of permutations] \label{def:support_G}
For any distinct $\vec{x}^*=(x_1^*,...,x^*_k)$, 
% \minki{Do we use disjoint for non-pairs? also, it may go to Hit and Miss?}\takashi{I replaced "disjoint" with "distinct"}
let $G[\vec{x}^*]$  be the set consisting of all pairs $(\pi,\pi^*)$ such that the $k$-tuple
$(p_1^*=(x_1^*,\pi^*(x_1^*)),...,p_k^*=(x_k^*,\pi^*(x_k^*)))$ is good w.r.t.\ $\pi$. 
\end{definition}
Then the following lemmas are easy to prove.
\begin{lemma}[Reprogramming on good tuples]\label{lem:reprogram_good_tuples}
    Consider any permutation $\pi$ and $k$ pairs $p_1^\ast,\dots, p_k^\ast$ with $p_j^\ast=(x_j^\ast, y_j^\ast)$ for $j=1, \dots, k$.  Suppose the tuple of pairs $(p_1^\ast, \dots, p_k^\ast)$ is good w.r.t.\ $\pi$. Then we have: 
    \[
    \pi[\vec{x}^* \to \vec{y}^*](z)=
    \begin{cases}
    y_j^\ast & \text{if $z=x_j^\ast$ for some $j\in[k]$,}\\
    \pi(x_j^\ast) & \text{if $z=\pi^{-1}(y_j^\ast)$ for some $j\in[k]$,}\\
    \pi(z) & \text{otherwise.}
    \end{cases}
    \]
\end{lemma}

The proof of this lemma is given in the Appendix, page~\pageref{lem:reprogram_good_tuples_proof}.

%\aaram{For example, using the partial permutation terminologies, this becomes
%\[
%\pi[p](z)=\begin{cases}
%    p(z) & \text{if $z\in\dom(p)$,}\\
%    \pi(p^{-1}(\pi(z))) & \text{if $\pi(z)\in\rng(p)$,}\\
%    \pi(z) & \text{otherwise.}
%\end{cases}
%\]
%The definition of goodness makes sure that these cases are complementary and %exhaustive.}

\begin{lemma}[Bad probability]\label{lem:bad_prob} 
Let $\pi$ be a (fixed) permutation and
$\vec{x}^*=(x^*_1,...,x^*_k)$ be a (fixed) distinct tuple.  
%For $i\in [k]$, let $p^*_i=(x^*_i,\pi^*(x^*_i))$ for a uniformly random permutation $\pi^*$. 
Then we have 
\[
\Pr_{\pi^*}[(\pi,\pi^*)\notin G[\vec{x}^*]]\le \frac{k^2}{|X|}
\]
\end{lemma}
\begin{proof}
Recall that $(\pi,\pi^*)\in G[\vec{x}^*]$ means that  $(p_1^*=(x^*_1,y^*_1=\pi^*(x^*_1)),\ldots, p_k^*=(x^*_k,y^*_k=\pi^*(x^*_k)))$ is good with respect to $\pi$, i.e., 
$(p_1^*,...,p^*_j)$ is disjoint and  
 $\pi(x_i^*) \ne y_j^*$ for all $i,j \in [k]$.
Since $\vec{x}^*$ is assumed to be distinct, the disjointness of $(p_1^*,...,p^*_j)$ is always satisfied. 

The only condition that might fail is the other one. We aim to find an upper bound on the probability that $\pi(x_i^*) = y_j^*$ for some $i,j$.

For each pair $(i, j)$, consider the event $E_{i,j}$ that  $\pi(x_i^*) = y_j^*$ holds.
Since $\pi(x_i^*)$ is a fixed element in $X$, and $y_j^* = \pi^*(x_j^*)$ is uniformly random over $X$, the probability that $\pi(x_i^*) = y_j^*$ is:
\[
\Pr[E_{i,j}] = \frac{1}{|X|}.
\]

There are $k^2$ such pairs $(i, j)$. By the union bound, the probability that at least one of these events occurs is at most the sum of their individual probabilities:
\[
\Pr\left( \bigcup_{i \ne j} E_{i,j} \right) \leq \sum_{i,j} \Pr[E_{i,j}] = \frac{k^2}{|X|}.
\]
Therefore, the probability that $(\pi,\pi^*)\in G[\vec{x}^*]$ is at most $\dfrac{k^2}{|X|}$.
\end{proof}

\begin{lemma}[Uniformity of reprogrammed permutation]\label{lem:uniform} 
For any distinct $\vec{x}^*=(x_1^*,...,x^*_k)$, suppose that we uniformly take $(\pi,\pi^*)\leftarrow G[\vec{x}^*]$ and set $\vec{y}^*=\pi^*(\vec{x}^*)$. 
Then $\pi[\vec{x}^*\rightarrow \vec{y}^*]$ is distributed uniformly randomly. 
%Then $\pi[x^*_1\rightarrow \pi^*(x^*_1)]...[x^*_k\rightarrow \pi^*(x^*_k)]$ is distributed uniformly randomly. 
\end{lemma}
\begin{proof}
Observe that  for any permutation $\sigma$, $((x_1^*,\sigma\pi^*(x_1^*)),...,(x_k^*,\sigma\pi^*(x_k^*)))$ is good w.r.t. $\sigma \pi$ if and only if $((x_1^*,\pi^*(x_1^*)),...,(x_k^*,\pi^*(x_k^*)))$ is good w.r.t. $\pi$.  
This means that  for any $\sigma$, if we take uniform $(\pi,\pi^*)\leftarrow G[\vec{x}^*]$, then $(\sigma\pi,\sigma\pi^*)$ is also distributed uniformly on $G[\vec{x}^*]$. 

For any $\sigma$, we have the following equivalence of distributions: 
\begin{align*}
&\{\sigma\left(\pi[x^*_1\rightarrow \pi^*(x^*_1)]...[x^*_k\rightarrow \pi^*(x^*_k)]\right):(\pi,\pi^*)\leftarrow G[\vec{x}^*]\}\\
&\equiv\{(\sigma\pi)[x^*_1\rightarrow \sigma\pi^*(x^*_1)]...[x^*_k\rightarrow \sigma\pi^*(x^*_k)]:(\pi,\pi^*)\leftarrow G[\vec{x}^*]\}\\
&\equiv\{\pi[x^*_1\rightarrow \pi^*(x^*_1)]...[x^*_k\rightarrow \pi^*(x^*_k)]:(\pi,\pi^*)\leftarrow G[\vec{x}^*]\}
\end{align*}
where the first equivalence easily follows from the definition of reprogramming and the second equivalence follows from the above observation. 
The above equivalence means that the distribution of $\pi[x^*_1\rightarrow \pi^*(x^*_1)]...[x^*_k\rightarrow \pi^*(x^*_k)]$ for $(\pi,\pi^*)\leftarrow G[\vec{x}^*]$ is invariant under left multiplication by any permutation $\sigma$. 
This means that the distribution of $\pi[\vec{x}^*\rightarrow \vec{y}^*]=\pi[x^*_1\rightarrow \pi^*(x^*_1)]...[x^*_k\rightarrow \pi^*(x^*_k)]$  is uniform. 
\end{proof}

\subsubsection*{Hit and Miss Queries}
In the proof of our main theorems, we will need the following definition. %\qipeng{Add some explanations.}
\begin{definition}[Hit and Miss queries]
\label{def:hit_miss_queries}
    Fix any distinct $\vec{x}^*=(x_1^*,...,x_k^*) \in X^k$, permutations $(\pi,\pi^*)\in G[\vec{x}^*]$, %(as defined in \Cref{def:support_G}), 
    $\vec{y}^* = \pi^*(\vec{x}^*)$, %\footnote{where for a vector $\vec{x} = (x_1, ..., x_k) \in X^k$ and a permutation $\pi^*$, the vector $\pi^*(\vec{x}^*)$ is defined as $\pi^*(\vec{x}^*) := (\pi^*(x_1), ..., \pi^*(x_k)) \in X^k$}.
    %For the $j$-th pair $(x^*_j, y^*_j)$, 
    and $j\in [k]$, 
    we define the Hit and Miss input for forward queries as follows:
    \begin{align*}
        x_j^{\sf hit} &= x_j^*, \\
        x_j^{\sf miss} &= \pi^{-1}(y^*_j).
    \end{align*}
    %We define $X^{\sf hit} := \{x_j^{\sf hit}, j = 1, 2, \ldots, k\}$ and $X^{\sf miss} := \{x_j^{\sf miss}\} - X^{\sf hit}$.

    Similarly, %for the $j$-th pair $x^*_j, y^*_j$, 
    we define the Hit and Miss input for backward queries as follows:
    \begin{align*}
        y_j^{\sf hit} &= y_j^*, \\
        y_j^{\sf miss} &= \pi(x_j^*).
    \end{align*}
    %Note that $x_j^{\sf hit}$ (or $y_j^{\sf hit}$) can be equal to $x_j^{\sf miss}$ (or $y_j^{\sf miss}$) if we do not restrict $\pi, \pi^*$. By the definition of ``good'' (\Cref{def:good}), they are always distinct.
    % \minki{The last two sentences are confusing. What does mean by ``if we do not restrict $\pi, \pi^*$''? Also, ``they are always distinct'' what are they?}\takashi{I agree, I just removed it.}
\end{definition}
\begin{remark}
Since we assume $(\pi,\pi^*)\in G[\vec{x}^*]$, there is not duplicate entry in 
$(x_1^{\sf hit},x_1^{\sf miss},...,x_k^{\sf hit},x_k^{\sf miss})$
or 
$(y_1^{\sf hit},y_1^{\sf miss},...,y_k^{\sf hit},y_k^{\sf miss})$.
\end{remark}

The following corollary immediately follows from \Cref{lem:reprogram_good_tuples}. 
\begin{corollary}\label{cor:reprogram_good_tuples} 
%\takashi{I added this corollary. I believe this should be useful (or actually we only had to prove this from the beginning?).}
   Let $\vec{x}^*=(x_1^*,...,x_k^*)\in X^k$ be a distinct tuple, 
$(\pi,\pi^*)\in G[\vec{x}^*]$, 
$\vec{y}^*=(y_1^*,...,y_k^*)=\pi^*(\vec{x}^*)$, 
and 
$(x_j^{\sf hit},x_j^{\sf miss},y_j^{\sf hit},y_j^{\sf miss})$ be as defined in \Cref{def:hit_miss_queries} for $j\in[k]$. 
Then we have: 
    \[
    \pi[\vec{x}^* \to \vec{y}^*](z)=
    \begin{cases}
    y_j^{\sf hit} & \text{if $z=x_j^{\sf hit}$ for some $j\in[k]$,}\\
    y_j^{\sf miss} & \text{if $z=x_j^{\sf miss}$ for some $j\in[k]$,}\\
    \pi(z) & \text{otherwise.}
    \end{cases}
    \]
\end{corollary}

\subsubsection*{Partial Reprogramming} 
%\takashi{I moved this to here because the lemma can be written more clearly if we use the hit and miss queries.}
We define \emph{partial reprogramming}. Looking ahead, this corresponds to a "snapshot" of the oracle simulated by the simulator at some point of its execution.   
\begin{definition}[Partial reprogramming]\label{def:partial}
Let $\vec{x}^*=(x_1^*,...,x_k^*)\in X^k$ and $\vec{y}^*=(y_1^*,...,y_k^*)\in X^k$ be distinct tuples and $\pi:X\rightarrow X$ be a permutation.  %$(\pi,\pi^*)\in G[x^*]$, and $\vec{y}^*=(y_1^*,...,y_k^*)=\pi^*(\vec{x}^*)$. 
We say that $\pi'$ is a partial reprogramming of $\pi$ w.r.t. $(\vec{x}^*,\vec{y}^*)$ if $\pi'=\pi[x^*_{j_1}\rightarrow y^*_{j_1}]...[x^*_{j_\ell}\rightarrow y^*_{j_\ell}]$ for some distinct sequence $j_1,...,j_\ell \in [k]$ and $\ell \leq k$. For $j\in [k]$, we say that $\pi'$ is reprogrammed on $j$ if $j\in \{j_1,...,j_\ell\}$. %\takashi{slightly changed the wording.}
%We say that $\pi'$ is a partial reprogramming of $\pi$ by the pairs $(p_1^*,...,p^*_k)$, where $p_1^* = (x_1^*, y_1^*), \ldots, p_k^* = (x_k^*, y_k^*)$ if $\pi'=\pi[x^*_{j_1}\rightarrow y^*_{j_1}]...[x^*_{j_\ell}\rightarrow y^*_{j_\ell}]$ for some distinct sequence $j_1,...,j_\ell \in [k]$ and $l \leq k$. %\footnote{That is, $\pi'$ is a "snapshot" of the oracle $O$ at some point of the execution of $S[\pi,p^*_1,...,p^*_k]$.}   
%For $j\in [k]$, we say that reprogramming by $p^*_j$ 
%occurs for $\pi'$ if $j\in \{j_1,...,j_\ell\}$. 
\end{definition}

The following lemma immediately follows from  \Cref{cor:reprogram_good_tuples}. 
%\takashi{I removed the proof sketch since this looks really immediate.}
\begin{lemma}[Partial reprogramming on good tuples]\label{cla:claim} 
Let $\vec{x}^*=(x_1^*,...,x_k^*)\in X^k$ be a distinct tuple, 
$(\pi,\pi^*)\in G[\vec{x}^*]$, 
$\vec{y}^*=(y_1^*,...,y_k^*)=\pi^*(\vec{x}^*)$, 
%Suppose that $(p_1^*=(x^*_1,y^*_1),...,p^*_k=(x^*_k,y^*_k))$ is good w.r.t. a permutation $\pi:X\rightarrow X$.
$\pi'$ be a partial reprogramming of $\pi$ w.r.t. $(\vec{x}^*,\vec{y}^*)$, and 
%For $j\in[k]$, $X_j:=\{x_j^{\sf hit},x_j^{\sf miss}\}$ and $Y_j:=\{y_j^{\sf hit},y_j^{\sf miss}\}$ where 
$(x_j^{\sf hit},x_j^{\sf miss},y_j^{\sf hit},y_j^{\sf miss})$ be as defined in \Cref{def:hit_miss_queries} for $j\in[k]$. 
Then the following hold. 
\begin{enumerate}
    \item \label{item:equal_non-touch} For any $x\in X\setminus \bigcup_{j\in [k]} \{x_j^{\sf hit},x_j^{\sf miss}\}$, 
    \[\pi'(x)=\pi[\vec{x}^*\rightarrow \vec{y}^*](x)=\pi(x).\] %\takashi{I added $=\pi(x)$ for clarity}
\item \label{item:equal_touch} 
For any $x \in \{x_j^{\sf hit},x_j^{\sf miss}\}$ for some $j\in[k]$, if $\pi'$ is reprogrammed on $j$, then   
 \[\pi'(x)=
\pi[\vec{x}^*\rightarrow \vec{y}^*](x).\]
   \item For any $y\in X\setminus \bigcup_{j\in [k]} \{y_j^{\sf hit},y_j^{\sf miss}\}$, 
    \[{\pi'}^{-1}(y)={\pi[\vec{x}^*\rightarrow \vec{y}^*]}^{-1}(y)=\pi^{-1}(y).\]
\item For any $y\in \{y_j^{\sf hit},y_j^{\sf miss}\}$ for some $j\in[k]$, if $\pi'$ is reprogrammed on $j$, then
    \[{\pi'}^{-1}(y)=
{\pi[\vec{x}^*\rightarrow \vec{y}^*]}^{-1}(y).\]
\end{enumerate}
\end{lemma}
%\begin{remark}\label{rem:assumption_claim}
%Whenever we invoke \Cref{cla:claim}, we always set $\vec{y}^*=\pi^*(\vec{x}^*)$ for some permutation $\pi^*$. In this case, the assumption of the lemma can be written as $(\pi,\pi^*)\in G[\vec{x}^*]$. 
%\end{remark}
\if0
\begin{proof}(sketch)
\qipeng{expand this proof; I feel we may just use \Cref{lem:reprogram_good_tuples} to prove it, which seems immediate.}
    The first item is obvious since the values on $x\in X\setminus \cup_{j\in [k]} X_j$ does not change by any reprogramming and thus we have 
    \[\pi'(x)=
\pi(x)=\pi[\vec{x}^*\rightarrow \vec{y}^*](x).\]

For the second item,  note that reprogramming by $p_{j'}^*$ for $j'\ne j$ does not change the values on $x\in X_j$, due to the good property (Def.~\ref{def:good}).  
Thus, due to a similar reason to that for the first item, we have 
  \[\pi'(x)=\pi[x^*_j\rightarrow y^*_j](x)=\pi[\vec{x}^*\rightarrow \vec{y}^*](x).\] 

The third and fourth items are similar to the first and second.
\end{proof}
\fi
\section{Classical Lifting Theorem for Permutations}\label{sec:classical_lifting}

In this section, as a warm-up, we  prove a classical analogue of our main lifting theorem. %(as well as the Lemma 4 of \cite{DFM20});
% \minki{The sentence in bracket has no information for general readers so we can remove it, or just say measure-and-reprogramming lemma?}\takashi{I removed.} 
%\revise{
While we encourage the readers to read this section for intuition before moving on to the quantum lifting theorem, it is also possible to skip directly to \Cref{sec:quantum_lifting}.  
%after consulting  \Cref{cla:claim,lem:uniform,lem:bad_prob}, since those are the only parts of this section used there.
%}

%Specifically, we show howto construct a classical $k$-query algorithm $\cB$ whose winning probability is at least $1/(2q+1)^k$ times $\cA$'s winning probability (ignoring an exponentially small multiplicative error).That is, 
We prove the following theorem:

\begin{theorem}[Classical Lifting Theorem]\label{thm:lifting}
Let $\cA$ be an algorithm that makes $q$ classical queries to an (invertible) random permutation oracle on $X$ and $R$ is a relation on $X^k\times X^k\times Z$. 
Then there exists an algorithm $\cB$ making at most 
%\takashi{This was "exactly", but I don't think this is consistent to the current description of $S$ since some of $v_j$ may be $\bot$ in which case it doesn't query. I made a similar fix to all other lifting theorems.}\alex{Yes, this looks right!}
$k$ classical queries such that 
\begin{align*}
&\Pr_{\pi^*} \left[
(x_1,...,x_k,\pi^*(x_1),...,\pi^*(x_k),z)\in R
:(x_1,...,x_k,z)\leftarrow \cB^{\pi^*}\right]\\
&\ge
\frac{\left(1 - \frac{k^2}{|X|} \right)}{(2q+1)^k}\Pr_{\pi^*} \left[
(x_1,...,x_k,\pi^*(x_1),...,\pi^*(x_k),z)\in R
:(x_1,...,x_k,z)\leftarrow \cA^{\pi^*}\right].
\end{align*}
\end{theorem}

\if0
Let $\cA$ be an algorithm (with no input) that makes $q$ classical queries to an (invertible) random permutation oracle $\pi:X\rightarrow X$ and outputs 
$(x_1,...,x_k,z)\in X^k \times Z$ such that $(x_1,...,x_k)$ are pairwise distinct.
We say that $\cA$ wins if $(x_1,...,x_k,\pi(x_1),...,\pi(x_k),z)\in R$ for a relation $R\subseteq X^k \times X^k \times Z$. 

We prove the following theorem. 

\begin{theorem}[Lifting Theorem, Classical]\label{thm:lifting}
Let $\cA,R$ be as above. 
There exists a classical algorithm $\cB$ that makes $k$ classical queries such that 
\begin{align*}
\Pr_{\pi}[
\cB^\pi~\text{wins}]\ge
\frac{\left( 1 - \frac{k^2}{|X|} \right)}{(2q+1)^k}\Pr_{\pi}[
\cA^\pi~\text{wins}].
\end{align*}
\end{theorem}
\fi

\subsection{Measure and Reprogram Lemma for Permutations}
%\minki{The description here is too dry}\takashi{I hope the technical overview explains the intuition.}
For proving \Cref{thm:lifting}, we first prove a lemma which we call the \emph{measure and reprogram lemma}. The lemma is stated using the simulator  $S[\cA,\pi,\pi^*]$ defined below. Looking ahead, the algorithm $\cB$ in the lifting theorem runs $S[\cA,\pi,\pi^*]$ where $\pi^*$ is $\cB$'s own oracle while $\pi$ is internally uniformly chosen.
%We define the algorithm $\cB$ in the lifting theorem by using the simulator $S[\cA,\pi,\pi^*]$ defined below.

\begin{definition}[Permutation Measure-and-Reprogram Experiment, Classical] \label{def:classical_simulator} 
For two permutations $\pi:X\rightarrow X$ , $\pi^*:X\rightarrow X$ let
$S[\cA, \pi, \pi^*]$ be an algorithm 
that has oracle access to $\pi$ and $\pi^*$ and runs $\cA$ with a stateful oracle $O$ as follows:

\begin{enumerate}
\item Pick $\vec{v}\in ([q]\cup \{\bot\})^k$ and $\vec{b} \in \{0,1,\bot\}^k$ uniformly at random, conditioned that
\begin{itemize}
    \item there is no duplicate entry for $\vec{v}$ other than $\bot$, and
    \item for any $j\in[k]$, $b_j=\bot$ if and only if $v_j=\bot$. 
\end{itemize} 
 %   \item Pick a uniformly random subset $\vec{v}$ of $[q]$, of length $k$. We have $1 < v_1 < \cdots < v_k \leq q$. 
  %  \item Pick $\vec{b} \in \{0,1\}^k$ uniformly at random.
    \item Initialize $O:=\pi$; here $O$ provides both forward and backward queries.
       \item Run $\cA^O$ where when $\cA$ makes its $i$-th query, the oracle is simulated as follows:
    \begin{enumerate}
        \item If $i = v_j$ for some $j \in [k]$, 
        we denote the $i$-th query by $x_{v_j}'$ if it is a forward query and by $y_{v_j}'$ if it is a backward query.  Then reprogram $O$ as follows according to the value of $b_j$ and answer the query using the reprogrammed oracle. 

    \begin{enumerate}
    \item If $b_j=0$ (hit), 
        \begin{itemize}
            \item If $\cA$'s $v_j$-th query is a forward query $x_{v_j}'$, then query $x'_{v_j}$ to $\pi^*$ to get $y'_{v_j}=\pi^*(x'_{v_j})$ and  reprogram $O$ to $O[x'_{v_j}\rightarrow y'_{v_j}]$.
            \item If $\cA$'s $v_j$-th query is a backward query $y'_{v_j}$, then query $y'_{v_j}$ to ${\pi^*}^{-1}$ to get $x'_{v_j}={\pi^*}^{-1}(y'_{v_j})$ and reprogram $O$ to $O[x'_{v_j}\rightarrow y'_{v_j}]$.
        \end{itemize}
    \item If $b_j=1$ (miss), 
        \begin{itemize}
            \item if $\cA$'s $v_j$-th query is a forward query $x'_{v_j}$, then query $\pi(x'_{v_j})$ to ${\pi^*}^{-1}$ to get ${\pi^*}^{-1}(\pi(x'_{v_j}))$ and  reprogram $O$ to $O[{\pi^*}^{-1}(\pi(x'_{v_j}))\rightarrow \pi(x'_{v_j})]$.
            \item if $\cA$'s $v_j$-th query is a backward query $y'_{v_j}$, then query $\pi^{-1}(y'_{v_j})$ to $\pi^*$ to get $\pi^*(\pi^{-1}(y'_{v_j}))$ and reprogram $O$ to $O[\pi^{-1}(y'_{v_j})\rightarrow \pi^*(\pi^{-1}(y'_{v_j}))]$.
    \end{itemize}
    \end{enumerate}
    \item Else, answer $\cA$'s $i$-th query by just using the stateful oracle $O$ without any measurement or reprogramming;
    \end{enumerate}
    \item Let $(\vec{x} = (x_1, ..., x_k), z)$ be $\cA$'s output;
    \item Output $(x_1, ..., x_k, z)$.
\end{enumerate}
\end{definition}
First, we observe the following easy lemma about $S[\cA,\pi,\pi^*]$. 
\begin{lemma}\label{lem:classical_correctness_of_reprogramming}
    Let $\vec{x}^*\in X^k$ be a distinct tuple, $\vec{y}^*=\pi^*(\vec{x}^*)$, 
    $(\pi,\pi^*)\in G[\vec{x}^*]$, and  $(x^{\sf hit}_j,x^{\sf miss}_j,y^{\sf hit}_j,y^{\sf miss}_j)$  be as  defined in \Cref{def:hit_miss_queries}.
    In an execution of $S[\cA,\pi,\pi^*]$, suppose that $v_j\ne \bot$ and one of the following holds:
    \begin{itemize}
        \item 
         $b_j=0$ and
        the $v_j$-th query is a forward query $x'_{v_j}=x^{\sf hit}_j$;
          \item 
         $b_j=0$ and
        the $v_j$-th query is a backward query $y'_{v_j}=y^{\sf hit}_j$;
           \item 
         $b_j=1$ and
        the $v_j$-th query is a forward query $x'_{v_j}=x^{\sf miss}_j$;
          \item 
         $b_j=1$ and
        the $v_j$-th query is a backward query $y'_{v_j}=y^{\sf miss}_j$.
    \end{itemize}
Then  $S[\cA,\pi,\pi^*]$ reprograms $O$ to $O[x^*_j \rightarrow y^*_j]$
at the $v_j$-th query. 
\end{lemma}
The proof of the above lemma is straightforward based on \Cref{def:hit_miss_queries} and the definition of $S[\cA,\pi,\pi^*]$.

Then we prove the "measure and reprogram lemma" below: 
% \minki{Perhaps it's nice to synchronize the descriptions with \cref{lemma:quantum_measure_reprogram}, however, this section is described by $(p_1^*,...,p_k^*)$ and the next section is based on $\pi^*$.}\takashi{I agree, I did so.}
\begin{lemma}[Measure and Reprogram Lemma, Classical]\label{lem:measure_and_reprogram}
Let $\cA$ be an algorithm that makes $q$ classical queries,  $\vec{x}^* = (x_1^*, ..., x_k^*)\in X^k$ be a distinct tuple, $(\pi,\pi^*)\in G[\vec{x}^*]$, %\footnote{See \Cref{def:support_G} for the definition of $G[\vec{x}^*]$.}
$\vec{y}^* = (y_1^*,...,y_k^*)=\pi^*(\vec{x}^*)$,
%\alex{shouldnt it be $\pi^*(\vec{x}^*)$?}\takashi{Yes, you are right, I fixed.} 
and 
$R\subseteq X^k \times X^k \times Z$ be a relation.
Then we have 
    \begin{align*}
    &\Pr\left[
\begin{array}{ll}
(x^*_1,...,x^*_k,y^*_1,...,y^*_k,z)\in R \\
~\land ~
\forall j\in[k]~x_j=x^*_j\\
\end{array}
:(x_1,...,x_k,z)\leftarrow S[\cA, \pi,\pi^*]\right]\\
&\ge  
\frac{1}{(2q+1)^{k}}
\Pr\left[
\begin{array}{ll}
(x^*_1,...,x^*_k,y^*_1,...,y^*_k,z)\in R \\
~\land ~
\forall j\in[k]~x_j=x^*_j\\
\end{array}
:(x_1,...,x_k,z)\leftarrow \cA^{\pi[\vec{x}^* \rightarrow \vec{y}^*]}\right]. 
\end{align*}
\end{lemma}
\begin{proof}
In an execution of $\cA^{\pi[\vec{x}^*\rightarrow \vec{y}^*]}$, 
we say that a query $\pi$-touches $p^*_j=(x^*_j,y^*_j)$ if the query is either a forward query $x\in X_j:=\{x_j^{\sf hit},x_j^{\sf miss}\}$ or a backward query $y\in Y_j:=\{y_j^{\sf hit},y_j^{\sf miss}\}$ where $x_j^{\sf hit},x_j^{\sf miss},y_j^{\sf hit},y_j^{\sf miss}$ are defined in \Cref{def:hit_miss_queries}.  
We say that $(\vec{v},\vec{b})$ is a correct guess if for all $j\in[k]$, either of the following hold:
\begin{itemize}
\item $\cA$'s $v_j$-th query is the first query that $\pi$-touches
% \minki{``touch'' is not define. Probably we can define them in Hit and Miss?}\takashi{I defined it here.} 
$p^*_j$, and moreover the way of touching it is as specified by $b_j$ ($b_j = 0$ refers to a Hit query and $b_j = 1$ refers to a Miss query), i.e., if $b_j=0$, then the query is either $x=x_j^{\sf hit}$ in the forward direction or $y=y_j^{\sf hit}$ in the backward direction, and if $b_j=1$, then the query is either $x=x_j^{\sf miss}$ in the forward direction or $y=y_j^{\sf miss}$ in the backward direction; 
%\footnote{That is, ($b_j=0$ AND (the $i_j$th query is $x^*_j$ in the forward direction or $y^*_j$ in the backward direction)) OR ($b_j=1$ AND (the $i_j$th query is $\pi^{-1}(y^*_j)$ in the forward direction or $\pi(x^*_j)$ in the backward direction))} 
\item $\cA$ never makes a query that $\pi$-touches $p^*_j$ and $v_j=\bot$. 
\end{itemize}
Then we have 
\begin{align*}
\Pr_{\vec{v},\vec{b} }\left[
\begin{array}{ll}
(x^*_1,...,x^*_k,y^*_1,...,y^*_k,z)\in R\\
~\land ~
\forall j\in[k]~x_j=x^*_j\\
~\land ~
(\vec{v},\vec{b}) \text{~is~the~correct~guess}
\end{array}
:(x_1,...,x_k,z)\leftarrow \cA^{\pi[\vec{x}^*\rightarrow \vec{y}^*]}\right]\\
\ge
\frac{1}{(2q+1)^k}\Pr\left[
\begin{array}{ll}
(x^*_1,...,x^*_k,y^*_1,...,y^*_k,z)\in R\\
~\land ~
\forall j\in[k]~x_j=x^*_j
\end{array}
:(x_1,...,x_k,z)\leftarrow \cA^{\pi[\vec{x}^*\rightarrow \vec{y}^*]}\right]
\end{align*}
since for an execution of $\cA^{\pi[\vec{x}^*\rightarrow \vec{y}^*]}$, there is a unique %\takashi{Minor comment: we are using here that the miss and hit queries are different for the uniqueness (but the uniqueness is not needed actually).} 
correct guess among at most $(2q+1)^k$ possibilities,\footnote{For each $j\in [k]$, we have 
$(v_j,b_j)\in ([q]\times \{0,1\})\cup \{(\bot,\bot)\}$, and thus
there are at most $(2q+1)$ possibilities.} and the choice of $(\vec{v},\vec{b})$ is independent of the execution of $\cA$. 

We also have 
\begin{align*}
&\Pr\left[
\begin{array}{ll}
(x^*_1,...,x^*_k,y^*_1,...,y^*_k,z)\in R \\
~\land ~
\forall j\in[k]~x_j=x^*_j\\
\end{array}
:(x_1,...,x_k,z)\leftarrow S[\cA,\pi,\pi^*]\right]\\
\ge
&\Pr_{\vec{v},\vec{b}}\left[
\begin{array}{ll}
(x^*_1,...,x^*_k,y^*_1,...,y^*_k,z)\in R\\
~\land ~
\forall j\in[k]~x_j=x^*_j\\
~\land ~
(\vec{v},\vec{b}) \text{~is~the~correct~guess}
\end{array}
:(x_1,...,x_k,z)\leftarrow \cA^{\pi[\vec{x}^*\rightarrow \vec{y}^*]}\right].
\end{align*}
since conditioned on that $(\vec{v},\vec{b})$ is the correct guess, $O$ simulated by $S[\cA,\pi,\pi^*]$ behaves exactly in the same way as $\pi[\vec{x}^*\rightarrow \vec{y}^*]$.
Indeed, if the guess is correct, at the 
$v_j$-th query, 
$S[\cA,\pi,\pi^*]$ reprograms $O$ as $x^*_j\rightarrow y^*_j$ by \Cref{lem:classical_correctness_of_reprogramming}, and 
any query $x\in X_j$ in the forward direction or $y\in Y_j$ in the backward direction is answered after this reprogramming is done. In this case the response from $O$ is identical to that from $\pi[\vec{x}^*\rightarrow \vec{y}^*]$ by \Cref{cla:claim}. %(Note that the assumption of \Cref{cla:claim} is satisfied since we assume $(\pi,\pi^*)\in G[\vec{x}^*]$.) 
Combining the above, we obtain \Cref{lem:measure_and_reprogram}. 
\end{proof}

\subsection{Proof of the Classical Lifting Theorem}

%First, we observe the following:

%By using \Cref{lem:uniform}, we are ready to prove the classical lifting theorem. 
We prove \Cref{thm:lifting} based on \Cref{lem:measure_and_reprogram}.
\begin{proof}[Proof of \Cref{thm:lifting}]

$\cB^{\pi^*}$ runs $S[\cA,\pi,\pi^*]$ for  uniformly random $\pi$.
Then one can see that for any $\pi^*$, 
\begin{align*}
&\Pr \left[
(x_1,...,x_k,\pi^*(x_1),...,\pi^*(x_k),z)\in R
:(x_1,...,x_k,z)\leftarrow \cB^{\pi^*}\right]\\
&=
\sum_{(x^*_1,...,x^*_k)}\Pr_{\pi}\left[
\begin{array}{ll}
(x^*_1,...,x^*_k,y^*_1,...,y^*_k,z)\in R \\
~\land ~
\forall j\in[k]~x_j=x^*_j
\end{array}
:(x_1,...,x_k,z)\leftarrow S[\cA,\pi,\pi^*]\right]
\end{align*}
where 
$y^*_j=\pi^*(x^*_j)$.  
%(The execution of $\cB^*$ \alex{$\cB^{\pi^*}$} is identical to that of $S[\pi,p^*_1,...,p^*_k]$ as long as all the reprogramming done by $\cB$ is $x^*_j\rightarrow y^*_j$. The latter event occurs iff $w_j\in \{(0,x^*_j),(1,y^*_j)\}$ in the corresponding execution of $S[\pi,p^*_1,...,p^*_k]$.)  
By taking an average over random $\pi^*$, we have

\begin{align*}
&\Pr_{\pi^*} \left[
(x_1,...,x_k,\pi^*(x_1),...,\pi^*(x_k),z)\in R
:(x_1,...,x_k,z)\leftarrow \cB^{\pi^*}\right]\\
&=
\sum_{(x^*_1,...,x^*_k)}\Pr_{\pi,\pi^*}\left[
\begin{array}{ll}
(x^*_1,...,x^*_k,y^*_1,...,y^*_k,z)\in R \\
~\land ~
\forall j\in[k]~x_j=x^*_j
\end{array}
:(x_1,...,x_k,z)\leftarrow S[\cA,\pi,\pi^*]\right]\\
&\ge
\sum_{(x^*_1,...,x^*_k)}
\Pr_{\pi,\pi^*}[(\pi,\pi^*)\in G[\vec{x}^*]]\\
&\cdot 
\Pr_{(\pi,\pi^*)\leftarrow G[\vec{x}^*]}\left[
\begin{array}{ll}
(x^*_1,...,x^*_k,y^*_1,...,y^*_k,z)\in R \\
~\land ~
\forall j\in[k]~x_j=x^*_j
\end{array}
:(x_1,...,x_k,z)\leftarrow S[\cA,\pi,\pi^*]\right]\\
&\ge
\sum_{(x^*_1,...,x^*_k)}
\left(1-\frac{k^2}{|X|}\right)\\
&\cdot 
\Pr_{(\pi,\pi^*)\leftarrow G[\vec{x}^*]}\left[
\begin{array}{ll}
(x^*_1,...,x^*_k,y^*_1,...,y^*_k,z)\in R \\
~\land ~
\forall j\in[k]~x_j=x^*_j
\end{array}
:(x_1,...,x_k,z)\leftarrow S[\cA,\pi,\pi^*]\right]\\
&\ge 
\sum_{(x^*_1,...,x^*_k)}
\left(1-\frac{k^2}{|X|}\right)
\frac{1}{(2q+1)^k}\\
&\cdot \Pr_{(\pi,\pi^*)\leftarrow G[\vec{x}^*]}\left[
\begin{array}{ll}
(x^*_1,...,x^*_k,y^*_1,...,y^*_k,z)\in R\\
~\land ~
\forall j\in[k]~x_j=x^*_j
\end{array}
:(x_1,...,x_k,z)\leftarrow \cA^{\pi[x^*_1\rightarrow y^*_1]...[x^*_k\rightarrow y^*_k]}\right]\\
&= \sum_{(x^*_1,...,x^*_k)}\left(1-\frac{k^2}{|X|}\right)
\frac{1}{(2q+1)^k}\Pr_{\pi}\left[
\begin{array}{ll}
(x^*_1,...,x^*_k,\pi(x^*_1),...,\pi(x^*_k),z)\in R\\
~\land ~
\forall j\in[k]~x_j=x^*_j
\end{array}
:(x_1,...,x_k,z)\leftarrow \cA^{\pi}\right]\\
&= \left(1-\frac{k^2}{|X|}\right)
\frac{1}{(2q+1)^k}\Pr_{\pi}\left[
(x_1,...,x_k,\pi(x_1),...,\pi(x_k),z)\in R
:(x_1,...,x_k,z)\leftarrow \cA^{\pi}\right]
\end{align*}
where 
the second inequality follows from \Cref{lem:bad_prob}, 
the third inequality follows from \Cref{lem:measure_and_reprogram}, 
and  the  
second-to-last
equality follows from \Cref{lem:uniform}. 

\end{proof}

\section{Quantum Lifting Theorem for Permutations}\label{sec:quantum_lifting}

We first state the main quantum lifting result:

\begin{theorem}[Quantum Lifting Theorem]
\label{thm:quantum_lifting}
Let $\cA$ be a quantum algorithm that makes $q$ quantum queries to an (invertible) random permutation oracle on $X$ and $R$ is a relation on $X^k \times X^k \times Z$. 
Then there exists an algorithm $\cB$ making at most $k$ classical queries such that 
\begin{align*}
&\Pr_{\pi^*} \left[
(x_1,...,x_k,\pi^*(x_1),...,\pi^*(x_k),z)\in R
:(x_1,...,x_k,z)\leftarrow \cB^{\pi^*}\right]\\
&\ge
\frac{\left(1 - \frac{k^2}{|X|}\right)}{(8q+1)^{2k}}\Pr_{\pi^*} \left[
(x_1,...,x_k,\pi^*(x_1),...,\pi^*(x_k),z)\in R
:(x_1,...,x_k,z)\leftarrow \cA^{\pi^*}\right].
\end{align*}
\end{theorem}

%We can now define the simulator $B$ in the quantum lifting theorem as follows. 
%\revise{
%We define the algorithm $\cB$ in the quantum lifting theorem by using the simulator $S[\cA,\pi,\pi^*]$ defined below.
%}  
Similarly to the classical case in \Cref{sec:classical_lifting}, we first define a simulator $S[\cA,\pi,\pi^*]$. 
This simulator works similarly to the classical setting (\Cref{def:classical_simulator}) except that it \emph{measures} the $v_j$-th query for each $j$ to determine the point on which the oracle is reprogrammed, and incorporates additional randomness $\vec{c}$, which determines whether  reprogramming occurs before or after answering a query. 
%We assume $\cA$ never queries coherently on the forward and the backward oracles. As noted by \Cref{lem:quantum_algo}, it is without loss of generality.

\begin{definition}[Permutation Measure-and-Reprogram Experiment] \label{def:quantum_simulator} 
For two permutations $\pi:X\rightarrow X$ , $\pi^*:X\rightarrow X$ let
$S[\cA, \pi, \pi^*]$ be an algorithm 
that has oracle access to $\pi$ and $\pi^*$ and runs $\cA$ with a stateful oracle $O$ as follows:

\begin{enumerate}
\item \label{item:vbc_condition} Pick $\vec{v}\in ([q]\cup \{\bot\})^k$, $\vec{b} \in \{0,1,\bot\}^k$, and $\vec{c} \in \{0,1,\bot\}^k$ uniformly at random, conditioned that
\begin{itemize}
    \item there is no duplicate entry for $\vec{v}$ other than $\bot$, and
    \item for any $j\in[k]$, if $v_j=\bot$, then $b_j=c_j=\bot$ and otherwise $b_j\ne \bot$ and $c_j\ne \bot$. 
\end{itemize} 
%    \item Pick a uniformly random subset $\vec{v}$ of $[q]$, of length $k$. We have $1 < v_1 < \cdots < v_k \leq q$. 
 %   \item Pick $\vec{b} \in \{0,1\}^k$ uniformly at random.
 %   \item Pick $\vec{c} \in \{0,1\}^k$ uniformly at random.
    \item Initialize $O:=\pi$; here $O$ provides both forward and backward queries.
       \item Run $\cA^O$ where 
       %the oracle $O$ is initialized to be a quantumly accessible classical oracle that computes $\pi$ and 
       when $\cA$ makes its $i$-th query, the oracle is simulated as follows:
    \begin{enumerate}
        \item If $i = v_j$ for some $j \in [k]$, measure $\cA$'s query register. We denote the measurement outcome by $x_{v_j}'$ if it is a forward query and by  $y_{v_j}'$ if it is a backward query.
        %\takashi{Previously, both cases were written as $x_{v_j}'$, which I think was very confusing. So I changed the notation.}\alex{I agree, this is less confusing, we were using it just to be consistent with the old notation from the classical case (which has also been updated).}
        %\takashi{What does this underline mean?}\alex{we just used it to highlight the differences from the classical setting, but we can also remove it.}\takashi{I removed it since this is not the only difference.} %and do either of the following:
        
If $c_j = 0$, first do the following reprogramming (according to the value of $b_j$) and then answer $\cA$'s $v_j$-th query using the reprogrammed oracle. Else if $c_j = 1$, answer $\cA$'s $v_j$-th query using oracle before reprogramming, and then do the following reprogramming. 
    \begin{enumerate}
    \item If $b_j=0$ (hit), 
        \begin{itemize}
            \item If the measurement outcome is a forward query $x'_{v_j}$, then query $x'_{v_j}$ to $\pi^*$ to get $y'_{v_j}=\pi^*(x'_{v_j})$ and  reprogram $O$ to $O[x'_{v_j}\rightarrow y'_{v_j}]$.
            \item If the measurement outcome is a backward query $y'_{v_j}$, then query $y'_{v_j}$ to ${\pi^*}^{-1}$ to get $x'_{v_j}={\pi^*}^{-1}(y'_{v_j})$ and reprogram $O$ to $O[x'_{v_j}\rightarrow y'_{v_j}]$.
        \end{itemize}
    \item If $b_j=1$ (miss), 
        \begin{itemize}
            \item If the measurement outcome is a forward query $x'_{v_j}$, then query $\pi(x'_{v_j})$ to ${\pi^*}^{-1}$ to get ${\pi^*}^{-1}(\pi(x'_{v_j}))$ and  reprogram $O$ to $O[{\pi^*}^{-1}(\pi(x'_{v_j}))\rightarrow \pi(x'_{v_j})]$.
            \item If the measurement outcome is a backward query $y'_{v_j}$, then query $\pi^{-1}(y'_{v_j})$ to $\pi^*$ to get $\pi^*(\pi^{-1}(y'_{v_j}))$ and reprogram $O$ to $O[\pi^{-1}(y'_{v_j})\rightarrow \pi^*(\pi^{-1}(y'_{v_j}))]$.
    \end{itemize}
    \end{enumerate}
    \item Else, answer $\cA$'s $i$-th query by just using the stateful oracle $O$ without any measurement or reprogramming;
    \end{enumerate}
    \item Let $(\vec{x} = (x_1, ..., x_k), z)$ be $\cA$'s output;
    \item Output $(x_1, ..., x_k, z)$.
\end{enumerate}
\end{definition}

% We now state our main theorem. 
% \begin{theorem}[Lifting Theorem, Quantum]
% \label{thm:quantum_lifting}
% Let $A$ be a quantum algorithm that makes $q$ quantum queries to an  (invertible) random permutation oracle and $R$ is a relation on $X^k \times X^k \times Z$. 
% Then there exists an algorithm $B$ making exactly $k$ classical queries such that 
% \begin{align*}
% &\Pr_{\pi^*} \left[
% (x_1,...,x_k,\pi^*(x_1),...,\pi^*(x_k),z)\in R
% :(x_1,...,x_k,z)\leftarrow B^{\pi^*}\right]\\
% &\ge
% \frac{\left(1 - \frac{k^2}{|X|}\right)}{(4q+1)^{2k}}\Pr_{\pi^*} \left[
% (x_1,...,x_k,\pi^*(x_1),...,\pi^*(x_k),z)\in R
% :(x_1,...,x_k,z)\leftarrow A^{\pi^*}\right].
% \end{align*}
% \end{theorem}

% \begin{remark}
%     In the theorem wlog we assume that each query is either forward or backward, but not in superposition of forward/backward. This relaxation can be assumed by doubling the number of queries made by $A$. Thus, in that case, the loss will be $1/(8q+1)^{2k}$. \qipeng{elaborate on this}
% \end{remark}

%\alex{The algorithm $\cB$ in \Cref{{thm:quantum_lifting}} will simply be instantiated as the simulator $S[A, \pi, \pi^*]$.} We will now proceed at 
%analyzing the success probability of the simulator $B$. \takashi{$B$ hasn't been defined yet.}

%\revise{
Similarly to the classical case, we observe the following lemma:
\begin{lemma}\label{lem:quantum_correctness_of_reprogramming} 
%\takashi{I think we were implicitly using this, but I think this should be stated explicitly.}
    Let $\vec{x}^*\in X^k$ be a distinct tuple, $\vec{y}^*=\pi^*(\vec{x}^*)$, 
    $(\pi,\pi^*)\in G[\vec{x}^*]$, and  $(x^{\sf hit}_j,x^{\sf miss}_j,y^{\sf hit}_j,y^{\sf miss}_j)$  be as  defined in \Cref{def:hit_miss_queries}.
    In an execution of $S[\cA,\pi,\pi^*]$, suppose that $v_j\ne \bot$ and one of the following holds:
    \begin{itemize}
        \item 
         $b_j=0$ and
        the measured $v_j$-th query is a forward query $x'_{v_j}=x^{\sf hit}_j$;
          \item 
         $b_j=0$ and
        the measured $v_j$-th query is a backward query $y'_{v_j}=y^{\sf hit}_j$;
           \item 
         $b_j=1$ and
        the measured $v_j$-th query is a forward query $x'_{v_j}=x^{\sf miss}_j$;
          \item 
         $b_j=1$ and
        the measured $v_j$-th query is a backward query $y'_{v_j}=y^{\sf miss}_j$.
    \end{itemize}
Then  $S[\cA,\pi,\pi^*]$ reprograms $O$ to $O[x^*_j \rightarrow y^*_j]$
at the $v_j$-th query (before or after answering the query). 
\end{lemma}
The proof of the above lemma is straightforward based on \Cref{def:hit_miss_queries} and the definition of $S[\cA,\pi,\pi^*]$. 
For completeness, we give a proof in \Cref{correctness_reprogramming}. 
%}

Then we prove the following quantum measure-and-reprogram lemma. 
\begin{lemma}[Quantum Measure-and-Reprogram Lemma] \label{lemma:quantum_measure_reprogram}
Let $\cA$ be an algorithm that makes $q$ quantum queries, $\vec{x}^* = (x_1^*, ..., x_k^*)\in X^k$ be a distinct tuple,  $(\pi,\pi^*)\in G[\vec{x}^*]$, %\footnote{See \Cref{def:support_G} for the definition of $G[\vec{x}^*]$.} 
$\vec{y}^* = (y_1^*,...,y_k^*)=\pi^*(\vec{x}^*)$,
%\alex{shouldnt it be $\pi^*(\vec{x}^*)$?}\takashi{Yes!} 
and 
$R\subseteq X^k \times X^k \times Z$ be a relation.
Then, we have:
    \begin{align*}
    &\Pr\left[
\begin{array}{ll}
(x^*_1,...,x^*_k,y^*_1,...,y^*_k,z)\in R \\
~\land ~
\forall j\in[k]~x_j=x^*_j\\
\end{array}
:(x_1,...,x_k,z)\leftarrow S[\cA, \pi,\pi^*]\right]\\
&\ge  
\frac{1}{(8q+1)^{2k}}
\Pr\left[
\begin{array}{ll}
(x^*_1,...,x^*_k,y^*_1,...,y^*_k,z)\in R \\
~\land ~
\forall j\in[k]~x_j=x^*_j\\
\end{array}
:(x_1,...,x_k,z)\leftarrow \cA^{\pi[\vec{x}^* \rightarrow \vec{y}^*]}\right]. 
\end{align*}
\end{lemma}
\begin{proof}[Proof of \Cref{lemma:quantum_measure_reprogram}]%[Proof of \Cref{thm:quantum_lifting}]

By relying on the normal form in \Cref{lem:quantum_algo} \footnote{Using this definition allows us to  assume w.l.o.g. each query is already fixed as a forward or backward query before the execution (causing only a constant loss).}, we can assume the algorithm is in the normal form and makes $2q$ queries.

We will use the following notation.
Similar to the notations defined in \Cref{sec:prelim_algo_permutation}, we will denote a forward quantum query to the original permutation $\pi$, by the query operator $O_{\pi}$; a backward quantum query to the original permutation $\pi$, by the query operator $O_{\pi^{-1}}$.
A forward (or backward) quantum query to the reprogrammed permutation $\pi[x_j^* \rightarrow y_j^*]$ will be referred to, as standard, by the operator $O_{\pi[x_j^* \rightarrow y_j^*]}$ (or  $O_{\pi[x_j^* \rightarrow y_j^*]^{-1}}$ %\takashi{This was written as $O_{\pi^{-1}[x_j^* \rightarrow y_j^*]}$, but I believe this should be $O_{\pi[x_j^* \rightarrow y_j^*]^{-1}}$. Please let me know if I'm wrong.} \alex{Yes, you are right, it should be $\pi[x_j^* \rightarrow y_j^*]^{-1}$} 
respectively) and a forward quantum query to the reprogrammed permutation $\pi[\vec{x}^* \rightarrow \vec{y}^*]$, by the vectors $\vec{x}^*$ and $\vec{y}^*$ will be denoted by $O_{\pi[\vec{x}^* \rightarrow \vec{y}^*]}$ (or $O_{\pi[\vec{x}^* \rightarrow \vec{y}^*]^{-1}}$ respectively). 

Fix a permutation $\pi$, any $\vec{x}^* \in X^k$ without duplicate entries and $\vec{y}^* = \pi^*(\vec{x}^*) \in X^k$, let $\left| {\psi^{{\pi[\vec{x}^* \rightarrow \vec{y}^*]}}_{2q}} \right\rangle$ 
%\takashi{I changed to $\psi$ since we are using $\phi$ for another meaning later.} \alex{sounds good, I think I only used $\phi$ to have same notations with the sub-normalized states}
denote the state of the algorithm $\As$ after making all its queries to ${\pi[\vec{x}^* \rightarrow \vec{y}^*]}$. 
Now, we will analyze the execution of an algorithm $\As$ that has oracle access to the reprogrammed permutation $\pi[\vec{x}^* \rightarrow \vec{y}^*]$.

Then the final state $\left| {\psi^{{\pi[\vec{x}^* \rightarrow \vec{y}^*]}}_{2q}} \right\rangle$ 
 of the algorithm $\As$ (before applying the final measurement) can be described by the following quantum state: 
\begin{align}\label{eq:def_final_state}
    \left| {\psi^{{\pi[\vec{x}^* \rightarrow \vec{y}^*]}}_{2q}} \right \rangle=  O_{{\pi[ \vec{x}^* \rightarrow \vec{y}^*]^{-1}}} U_{2q}  O_{{\pi[ \vec{x}^* \rightarrow \vec{y}^*]}} \cdots O_{{\pi[ \vec{x}^* \rightarrow \vec{y}^*]^{-1}}} U_2 O_{{\pi[ \vec{x}^* \rightarrow \vec{y}^*]}} U_1 \ket 0.
\end{align}

In the next step, we decompose this quantum state, so that each component in the decomposition corresponds to one of the cases in the quantum simulator; %\revise{
i.e., each component corresponds to a set of possible parameters $\vec{v}, \vec{b}, \vec{c}$ and the simulator with these parameters outputs  $\vec{x}^*$ and some $z$.

%Let us assume $(\pi, \pi^*)$ is in the support of $G[x_1^*, \ldots, x_k^*]$, in the analysis of the first query and the general case. 

%\revise{
More formally, we will show the following:
\begin{align}\label{eq:state_decomposition}
    \left| {\psi^{{\pi[\vec{x}^* \rightarrow \vec{y}^*]}}_{2q}} \right \rangle=  
   \sum_{\vec{v},\vec{b},\vec{c}} (-1)^{\beta_{\vec{v},\vec{b},\vec{c}}} \ket{\phi_{\vec{v}, \vec{b}, \vec{c}}}
\end{align} 
%\takashi{I believe we need a phase of $\pm 1$.}\alex{Yes, that's true, I just wasnt sure whether to add it in the summation or in the recursive definition of $\phi_{\vec{v}, \vec{b}, \vec{c}}$, but since we are fully defining the states, it is better like this.}
where $\beta_{\vec{v},\vec{b},\vec{c}}\in \{0,1\}$ 
and the sum is taken over all $(\vec{v},\vec{b},\vec{c})$ that satisfies the conditions in \Cref{item:vbc_condition} of \Cref{def:quantum_simulator}.  
Here, $\ket{\phi_{\vec{v}, \vec{b}, \vec{c}}}$ is a subnormalized state corresponding to the final state of $S[\cA,\pi,\pi^*]$  (before applying the final measurement) with the fixed choice of $(\vec{v},\vec{b},\vec{c})$,
where we insert the following projections: 
%\takashi{I found that "post-selection" was not a right term since it usually also contains renormalization. Thus, I changed to say "inserting projections".}
for all $j\in [k]$ such that $v_j\ne \bot$, insert the following projections right after receiving the $v_j$-th query:\footnote{Recall that  the $v_j$-th query is a forward query if $v_j$ is odd and 
is a backward query if $v_j$ is even.}
\begin{itemize}
\item If $v_j$ is odd and $b_j=0$, the $v_j$-th query resister is projected onto $\ket{x^{\sf hit}_j}$; 
\item If $v_j$ is odd and $b_j=1$, the $v_j$-th query resister is projected onto  $\ket{x^{\sf miss}_j}$; 
\item If $v_j$ is even and $b_j=0$, the $v_j$-th query resister is projected onto  $\ket{y^{\sf hit}_j}$; 
\item If $v_j$ is even and $b_j=1$, the $v_j$-th query resister is projected onto  $\ket{y^{\sf miss}_j}$; 
\end{itemize}
where $x^{\sf hit}_j$, $x^{\sf miss}_j$ $y^{\sf hit}_j$, and $y^{\sf miss}_j$  are defined in \Cref{def:hit_miss_queries}.

We will proceed to show \Cref{eq:state_decomposition} by induction over the index of the query of the algorithm $\cA$. More specifically, we will show that for any $1\leq t \leq 2q$, the state of the algorithm $\cA$ right after the $t$-th query can be written as:
\begin{align}\label{eq:state_decopose_middle}
    \left| {\psi^{{\pi[\vec{x}^* \rightarrow \vec{y}^*]}}_{t}} \right \rangle= 
    \sum_{\vec{v},\vec{b},\vec{c}} 
    (-1)^{\beta^{(t)}_{\vec{v},\vec{b},\vec{c}}} 
    \ket{\phi^{(t)}_{\vec{v}, \vec{b}, \vec{c}}}
    %\ket{\phi_{\vec{v}, \vec{b}, \vec{c}}}
\end{align}
where 
$\beta^{(t)}_{\vec{v},\vec{b},\vec{c}}\in \{0,1\}$ and
the sum is taken over all $\vec{v}\in ([t]\cup \{\bot\})^k$, $\vec{b}\in \{0,1,\bot\}^k$, and $\vec{c}\in \{0,1,\bot\}^k$ that satisfy the conditions in \Cref{item:vbc_condition} of \Cref{def:quantum_simulator}, and 
$\ket{\phi^{(t)}_{\vec{v}, \vec{b}, \vec{c}}}$ is defined similarly to $\ket{\phi_{\vec{v}, \vec{b}, \vec{c}}}$ except that we consider the state after the $t$-th query rather than the final state. Note that $\ket{\phi^{(2q)}_{\vec{v}, \vec{b}, \vec{c}}}=\ket{\phi_{\vec{v}, \vec{b}, \vec{c}}}$, thus it suffices to prove \Cref{eq:state_decopose_middle} holds for all $1\le t\le 2q$. 

We will first analyze the base case, i.e. the state of the algorithm after the first query.
\paragraph*{The first query. } Without loss of generality, we assume the first query is a forward query. We start by considering the state up to the first query: $\left| {\psi^{{\pi[\vec{x}^* \rightarrow \vec{y}^*]}}_{1}} \right \rangle=O_{\pi[\vec{x}^* \rightarrow \vec{y}^*]} U_1 \ket 0$. We insert an additional identity operator and have,
\begin{align*}
    &O_{\pi[\vec{x}^* \rightarrow \vec{y}^*]} U_1 \ket 0  = O_{\pi[\vec{x}^*\rightarrow \vec{y}^*]} \, {I} \, U_1 \, \ket 0 \\
     &= O_{\pi[\vec{x}^* \rightarrow \vec{y}^*]} \left(I - \sum_{j=1}^k \ketbra {x_j^{\sf hit}}{x_j^{\sf hit}} +  \sum_{j=1}^k \ketbra {x_j^{\sf hit}}{x_j^{\sf hit}}
     -  \sum_{j=1}^k \ketbra {x_j^{\sf miss}}{x_j^{\sf miss}} +  \sum_{j=1}^k \ketbra {x_j^{\sf miss}}{x_j^{\sf miss}}
     \right) U_1 \ket 0 \\
     &= \underbrace{O_{\pi[\vec{x}^* \rightarrow \vec{y}^*]} \left(I - \sum_{j=1}^k \ketbra {x_j^{\sf hit}}{x_j^{\sf hit}} -  \sum_{j=1}^k \ketbra {x_j^{\sf miss}}{x_j^{\sf miss}} \right) U_1 \ket 0}_{(i)} + \underbrace{O_{\pi[\vec{x}^* \rightarrow \vec{y}^*]} \sum_{j=1}^k \ketbra {x_j^{\sf hit}}{x_j^{\sf hit}} U_1 \ket 0}_{(ii)} \\
     &+ \underbrace{O_{\pi[\vec{x}^* \rightarrow \vec{y}^*]} \sum_{j=1}^k \ketbra {x_j^{\sf miss}}{x_j^{\sf miss}} U_1 \ket 0}_{(iii)}.
\end{align*}
%Here $x^{\sf miss}_j$ and $x^{\sf hit}_j$ are defined in \Cref{def:hit_miss_queries} 
% \alex{where we recall that $(\pi,\pi^*)\in G[\vec{x}^*]$}.
% \takashi{I don't see why we need the goodness. At this point, we are just doing mathematical calculation.}\alex{Yes, I agree, I only added that in regards to Aaram's comment on projection, but I am also fine to remove it.}
% \aaram{In this part, if $\pi(x_i^\ast)=y_i^\ast$ is possible, then $x_i^\textsf{hit}=x_i^\textsf{miss}$, so 
% $I - \sum_{j=1}^k \ketbra {x_j^{\sf hit}}{x_j^{\sf hit}} -  \sum_{j=1}^k \ketbra {x_j^{\sf miss}}{x_j^{\sf miss}}$
% cannot be a correct projection to the subspace spanned by inputs which are neither hit nor miss.}
% \takashi{This should be resolved by the change of the definition of goodness.}

The first term $(i)$ equals to
\begin{align*}
    & O_{\pi[\vec{x}^* \rightarrow \vec{y}^*]} \left(I - \sum_{j=1}^k \ketbra {x_j^{\sf hit}}{x_j^{\sf hit}} -  \sum_{j=1}^k \ketbra {x_j^{\sf miss}}{x_j^{\sf miss}}\right) U_1 \ket 0 \\
    =& O_{\pi} \left(I - \sum_{j=1}^k \ketbra {x_j^{\sf hit}}{x_j^{\sf hit}} -  \sum_{j=1}^k \ketbra {x_j^{\sf miss}}{x_j^{\sf miss}}\right) U_1 \ket 0 \\
    =& O_{\pi} U_1 \ket 0 - \sum_{j=1}^k O_\pi \ketbra {x_j^{\sf hit}}{x_j^{\sf hit}} U_1 \ket 0 - \sum_{j=1}^k  O_{\pi} \ketbra {x_j^{\sf miss}}{x_j^{\sf miss}}U_1 \ket 0.
\end{align*}
This is because 
$x_j^{\sf hit}\neq x_j^{\sf miss}$ for all $j\in [k]$ by $(\pi,\pi^*)\in G[\vec{x}^*]$, and
on inputs that are neither hit nor miss inputs, $O_\pi$ and $O_{\pi [\vec{x}^* \to \vec{y}^*]}$ are identical.

The second term $(ii)$ is
\begin{align*}
    O_{\pi[\vec{x}^* \rightarrow \vec{y}^*]} \sum_{j=1}^k \ketbra {x_j^{\sf hit}}{x_j^{\sf hit}} U_1 \ket 0 = \sum_{j=1}^k O_{\pi[x_j^* \rightarrow y_j^*]}  \ketbra {x_j^{\sf hit}}{x_j^{\sf hit}} U_1 \ket 0.
\end{align*}

Similarly, the third term $(iii)$ is
\begin{align*}
    O_{\pi[\vec{x}^* \rightarrow \vec{y}^*]} \sum_{j=1}^k \ketbra {x_j^{\sf miss}}{x_j^{\sf miss}} U_1 \ket 0 = \sum_{j=1}^k O_{\pi[x_j^* \rightarrow y_j^*]}  \ketbra {x_j^{\sf miss}}{x_j^{\sf miss}} U_1 \ket 0.
\end{align*}

Combining everything together, we have,
\begin{align*}
     O_{\pi[\vec{x}^* \rightarrow \vec{y}^*]} U_1 \ket 0 & = O_{\pi} U_1 \ket 0 - \sum_{j=1}^k O_\pi \ketbra {x_j^{\sf hit}}{x_j^{\sf hit}} U_1 \ket 0 - \sum_{j=1}^k  O_{\pi} \ketbra {x_j^{\sf miss}}{x_j^{\sf miss}}U_1 \ket 0 \\ 
     + & \sum_{j=1}^k O_{\pi[x_j^* \rightarrow y_j^*]}  \ketbra {x_j^{\sf hit}}{x_j^{\sf hit}} U_1 \ket 0 + \sum_{j=1}^k O_{\pi[x_j^* \rightarrow y_j^*]}  \ketbra {x_j^{\sf miss}}{x_j^{\sf miss}} U_1 \ket 0.
\end{align*}
Each term corresponds to either (1) do not measure the current query, or (2) measure the current query (which is a hit or miss query) and reprogram before or after the query. 
%\revise{
More specifically, 
we have 
\begin{align*}
    O_\pi U_1\ket{0}=\ket{\phi^{(1)}_{\bot^k,\bot^k,\bot^k}}
\end{align*}
and for any $j\in [k]$, 
% \begin{align*}
%     &O_\pi \ketbra {x_j^{\sf hit}}{x_j^{\sf hit}}U_1\ket{0}=\ket{\phi^{(1)}_{\bot^k_{j\rightarrow 1},\bot^k_{j\rightarrow 0},\bot^k_{j\rightarrow 1}}},~~~
%      O_\pi \ketbra {x_j^{\sf miss}}{x_j^{\sf miss}}U_1\ket{0}=\ket{\phi^{(1)}_{\bot^k_{j\rightarrow 1},\bot^k_{j\rightarrow 1},\bot^k_{j\rightarrow 1}}},\\
%      &O_{\pi[x^*_j\rightarrow y^*_j]} \ketbra {x_j^{\sf hit}}{x_j^{\sf hit}}U_1\ket{0}=\ket{\phi^{(1)}_{\bot^k_{j\rightarrow 1},\bot^k_{j\rightarrow 0},\bot^k_{j\rightarrow 0}}},~~~
%     O_{\pi[x^*_j\rightarrow y^*_j]}  \ketbra {x_j^{\sf miss}}{x_j^{\sf miss}}U_1\ket{0}=\ket{\phi^{(1)}_{\bot^k_{j\rightarrow 1},\bot^k_{j\rightarrow 1},\bot^k_{j\rightarrow 0}}},
% \end{align*}
\begin{align*}
    &O_\pi \ketbra {x_j^{\sf hit}}{x_j^{\sf hit}}U_1\ket{0}=\ket{\phi^{(1)}_{\bot^k_{j\rightarrow 1},\bot^k_{j\rightarrow 0},\bot^k_{j\rightarrow 1}}},~~~ \\
     &O_\pi \ketbra {x_j^{\sf miss}}{x_j^{\sf miss}}U_1\ket{0}=\ket{\phi^{(1)}_{\bot^k_{j\rightarrow 1},\bot^k_{j\rightarrow 1},\bot^k_{j\rightarrow 1}}},\\
     &O_{\pi[x^*_j\rightarrow y^*_j]} \ketbra {x_j^{\sf hit}}{x_j^{\sf hit}}U_1\ket{0}=\ket{\phi^{(1)}_{\bot^k_{j\rightarrow 1},\bot^k_{j\rightarrow 0},\bot^k_{j\rightarrow 0}}},~~~ \\
    &O_{\pi[x^*_j\rightarrow y^*_j]}  \ketbra {x_j^{\sf miss}}{x_j^{\sf miss}}U_1\ket{0}=\ket{\phi^{(1)}_{\bot^k_{j\rightarrow 1},\bot^k_{j\rightarrow 1},\bot^k_{j\rightarrow 0}}},
\end{align*}
where 
for $d\in \{0,1\}$, 
$\bot^k_{j\rightarrow d}$ is the sequence whose $j$-th entry is $d$ and all other entries are $\bot$. 

Thus, $\left| {\psi^{{\pi[\vec{x}^* \rightarrow \vec{y}^*]}}_{1}} \right \rangle$ can be decomposed as in  \Cref{eq:state_decopose_middle} (by trivially observing that $\beta^{(1)}_{\bot^k,\bot^k,\bot^k} = \beta^{(1)}_{{\bot^k_{j\rightarrow 1},\bot^k_{j\rightarrow 0},\bot^k_{j\rightarrow 0}}} = \beta^{(1)}_{{\bot^k_{j\rightarrow 1},\bot^k_{j\rightarrow 1},\bot^k_{j\rightarrow 0}}}  = 0, \beta^{(1)}_{{\bot^k_{j\rightarrow 1},\bot^k_{j\rightarrow 0},\bot^k_{j\rightarrow 1}}} = \beta^{(1)}_{{\bot^k_{j\rightarrow 1},\bot^k_{j\rightarrow 1},\bot^k_{j\rightarrow 1}}} = 1$).
%}  
%It holds similarly for backward queries. 
%\takashi{Technically speaking, I guess we don't need to separately analyze the first query since the initial state can be regarded as a "trivial decomposition"m which can be used as the base case for the induction. Is this true?}\alex{Yes, the decomposition for the first query just serves as intuition on how the decomposition of the algorithm's state relates to the simulator, but we are able to derive the decomposition after the total number of queries without this result.} \qipeng{I just feel that it may be helpful to describe the first query case. But it is definitely true that we can combine all the cases.}
Next, we will show the inductive step.

\paragraph*{The general case.} 
%\revise{
Assume that $\left| {\psi^{{\pi[\vec{x}^* \rightarrow \vec{y}^*]}}_{t}} \right \rangle$ is decomposed as in \Cref{eq:state_decopose_middle}.  Below, we show that 
$\left| {\psi^{{\pi[\vec{x}^* \rightarrow \vec{y}^*]}}_{t+1}} \right \rangle$ can also be decomposed as in \Cref{eq:state_decopose_middle}. 
We only give the proof for even $t$, in which case the $(t+1)$-th query is the forward query, but the proof is similar for the case of odd $t$.\footnote{We only have to replace $O_{\pi[\vec{x}^*\rightarrow \vec{y}^*]}$ with $O_{\pi[\vec{x}^*\rightarrow \vec{y}^*]}^{-1}$
and 
$(x_j^{\sf hit},x_j^{\sf miss})$ with 
$(y_j^{\sf hit},y_j^{\sf miss})$, and then we can see that this yields the desired decomposition by a similar argument.} 

For a fixed subnormalized state $\ket{\phi^{(t)}_{\vec{v}, \vec{b}, \vec{c}}}$ that appears in the decomposition of $\left| {\psi^{{\pi[\vec{x}^* \rightarrow \vec{y}^*]}}_{t}} \right \rangle$  in \Cref{eq:state_decopose_middle}, 
let $J = \{j_1, j_2, \ldots, j_\ell\}\subseteq [k]$ be the set of all indices on which $\vec{v}$ takes a non-$\bot$ entry, i.e., $v_{j_i}\ne \bot$ for $i \in [\ell]$.  
%We write $\pi[\vec{x}^*_J \rightarrow \vec{y}^*_J]$ to mean partial reprogramming $\pi[x^*_{j_1}\rightarrow y^*_{j_1}]...[x^*_{j_\ell}\rightarrow y^*_{j_\ell}]$.  

After applying the next 
internal unitary and the 
(forward) query operator, the overall state is:\footnote{Whenever we write $j\notin J$, it means $j\in [k]\setminus J$.}
%}
%\takashi{The following calculation is unchanged except that I replaced $\ket{\psi}$ with $\ket{\phi^{(t)}_{\vec{v}, \vec{b}, \vec{c}}}$ for notational consistency.}
\begin{align*}
    O_{\pi[\vec{x}^* \rightarrow \vec{y}^*]} U_{t+1} \ket{\phi^{(t)}_{\vec{v}, \vec{b}, \vec{c}}} & = O_{\pi[\vec{x}^* \rightarrow \vec{y}^*]} \, {I} \, U_{t+1} \, \ket{\phi^{(t)}_{\vec{v}, \vec{b}, \vec{c}}}\\
     &= O_{\pi[\vec{x}^* \rightarrow \vec{y}^*]} \left(I - \sum_{j \not\in J} \ketbra {x_j^{\sf hit}}{x_j^{\sf hit}} +  \sum_{j \not\in J} \ketbra {x_j^{\sf hit}}{x_j^{\sf hit}}
     -  \sum_{j \not\in J} \ketbra {x_j^{\sf miss}}{x_j^{\sf miss}} \right. \\
     &\left. +  \sum_{j \not\in J} \ketbra {x_j^{\sf miss}}{x_j^{\sf miss}}
     \right) U_{t+1} \ket{\phi^{(t)}_{\vec{v}, \vec{b}, \vec{c}}} \\
     &= \underbrace{O_{\pi[\vec{x}^* \rightarrow \vec{y}^*]} \left(I - \sum_{j \not\in J} \ketbra {x_j^{\sf hit}}{x_j^{\sf hit}} -  \sum_{j \not\in J} \ketbra {x_j^{\sf miss}}{x_j^{\sf miss}} \right) U_{t+1} \ket{\phi^{(t)}_{\vec{v}, \vec{b}, \vec{c}}}}_{(i)} + \\
     &+ \underbrace{O_{\pi[\vec{x}^* \rightarrow \vec{y}^*]} \sum_{j \not\in J} \ketbra {x_j^{\sf hit}}{x_j^{\sf hit}} U_{t+1} \ket{\phi^{(t)}_{\vec{v}, \vec{b}, \vec{c}}}}_{(ii)} \\
     &+ \underbrace{O_{\pi[\vec{x}^* \rightarrow \vec{y}^*]} \sum_{j \not\in J} \ketbra {x_j^{\sf miss}}{x_j^{\sf miss}} U_{t+1} \ket{\phi^{(t)}_{\vec{v}, \vec{b}, \vec{c}}}}_{(iii)}.
\end{align*}

The term $(i)$ is equal to 
\begin{align*}
&O_{\pi[\vec{x}^*_J \rightarrow \vec{y}^*_J]} U_{t+1} \ket{\phi^{(t)}_{\vec{v}, \vec{b}, \vec{c}}}  - \sum_{j \not\in J} O_{\pi[ \vec{x}^*_J \rightarrow \vec{y}^*_J]}  \ketbra {x_j^{\sf hit}}{x_j^{\sf hit}} U_{t+1} \ket{\phi^{(t)}_{\vec{v}, \vec{b}, \vec{c}}} \\
&-  \sum_{j \not\in J} O_{\pi[\vec{x}^*_J \rightarrow \vec{y}^*_J]}  \ketbra {x_j^{\sf miss}}{x_j^{\sf miss}}  U_{t+1} \ket{\phi^{(t)}_{\vec{v}, \vec{b}, \vec{c}}}
\end{align*}
%\revise{
where $\pi[\vec{x}^*_J \rightarrow \vec{y}^*_J]=\pi[x^*_{j_1}\rightarrow y^*_{j_1}]...[x^*_{j_\ell}\rightarrow y^*_{j_\ell}]$.\footnote{This is well-defined since the sequential reprogramming does not depend on the order by \Cref{lem:commutativity_disjoint_pairs}.}
This is because $x_j^{\sf hit}\neq x_j^{\sf miss}$ for all $j\in [k]$ by $(\pi,\pi^*)\in G[\vec{x}^*]$, and for any $x\notin \bigcup_{j\notin J}\{x_j^{\sf hit},x_j^{\sf miss}\}$, we have 
$
\pi[\vec{x}^*\rightarrow \vec{y}^*](x)
=
\pi[\vec{x}^*_J\rightarrow \vec{y}^*_J](x)
$
by \Cref{item:equal_non-touch,item:equal_touch} of  \Cref{cla:claim}. 
%We note that the assumption of \Cref{cla:claim} is satisfied since we assume $(\pi,\pi^*)\in G[\vec{x}^*]$.
%}

The second term $(ii)$ is
\begin{align*}
    O_{\pi[\vec{x}^* \rightarrow \vec{y}^*]} \sum_{j \not\in J} \ketbra {x_j^{\sf hit}}{x_j^{\sf hit}} U_{t+1} \ket{\phi^{(t)}_{\vec{v}, \vec{b}, \vec{c}}} = \sum_{j\not\in J} O_{\pi[\vec{x}^*_{J \cup \{j\}} \rightarrow \vec{y}_{J \cup \{j\}}^*]}  \ketbra {x_j^{\sf hit}}{x_j^{\sf hit}} U_{t+1} \ket{\phi^{(t)}_{\vec{v}, \vec{b}, \vec{c}}},
\end{align*}
and the third term $(iii)$ is
\begin{align*}
    O_{\pi[\vec{x}^* \rightarrow \vec{y}^*]} \sum_{j \not\in J} \ketbra {x_j^{\sf miss}}{x_j^{\sf miss}} U_{t+1} \ket{\phi^{(t)}_{\vec{v}, \vec{b}, \vec{c}}} = \sum_{j\not\in J} O_{\pi[\vec{x}^*_{J \cup \{j\}} \rightarrow \vec{y}_{J \cup \{j\}}^*]}  \ketbra {x_j^{\sf miss}}{x_j^{\sf miss}} U_{t+1} \ket{\phi^{(t)}_{\vec{v}, \vec{b}, \vec{c}}}
\end{align*}
%\revise{
where $\pi[\vec{x}^*_{J\cup \{j\}} \rightarrow \vec{y}^*_{J\cup \{j\}}]=\pi[x^*_{J}\rightarrow y^*_{J}][x^*_{j}\rightarrow y^*_{j}]$. 
This is because 
$
\pi[\vec{x}^*\rightarrow \vec{y}^*](x_j^{\sf hit})
=
\pi[\vec{x}^*_{J\cup \{j\}}\rightarrow \vec{y}^*_{J\cup \{j\}}](x_j^{\sf hit})
$
and 
$
\pi[\vec{x}^*\rightarrow \vec{y}^*](x_j^{\sf miss})
=
\pi[\vec{x}^*_{J\cup \{j\}}\rightarrow \vec{y}^*_{J\cup \{j\}}](x_j^{\sf miss})
$
by 
\Cref{item:equal_touch} of 
\Cref{cla:claim}. 
%}

Combining all the cases, $O_{\pi[\vec{x}^* \rightarrow \vec{y}^*]} U_{t+1} \ket{\phi^{(t)}_{\vec{v}, \vec{b}, \vec{c}}}$ can be decomposed as: 
\begin{align}\label{eq:state_decomp}\begin{split}
     &O_{\pi[\vec{x}^*_J \rightarrow \vec{y}^*_J]} U_{t+1} \ket{\phi^{(t)}_{\vec{v}, \vec{b}, \vec{c}}} \\  - &\sum_{j \not\in J} O_{\pi[\vec{x}^*_J \rightarrow \vec{y}^*_J]}  \ketbra {x_j^{\sf hit}}{x_j^{\sf hit}} U_{t+1} \ket{\phi^{(t)}_{\vec{v}, \vec{b}, \vec{c}}}  -  \sum_{j \not\in J} O_{\pi[\vec{x}^*_J \rightarrow \vec{y}^*_J]}  \ketbra {x_j^{\sf miss}}{x_j^{\sf miss}}  U_{t+1} \ket{\phi^{(t)}_{\vec{v}, \vec{b}, \vec{c}}}  \\
    + & \sum_{j\not\in J} O_{\pi[\vec{x}^*_{J \cup \{j\}} \rightarrow \vec{y}_{J \cup \{j\}}^*]}  \ketbra {x_j^{\sf hit}}{x_j^{\sf hit}} U_{t+1} \ket{\phi^{(t)}_{\vec{v}, \vec{b}, \vec{c}}} \\
    + & \sum_{j\not\in J} O_{\pi[\vec{x}^*_{J \cup \{j\}} \rightarrow \vec{y}_{J \cup \{j\}}^*]}  \ketbra {x_j^{\sf miss}}{x_j^{\sf miss}} U_{t+1} \ket{\phi^{(t)}_{\vec{v}, \vec{b}, \vec{c}}}.
    \end{split}
\end{align}
We observe that 
each term of \Cref{eq:state_decomp}
can be written as $\ket{\phi^{(t+1)}_{\vec{v}', \vec{b}', \vec{c}'}}$ for some $(\vec{v}', \vec{b}', \vec{c}')$. 
\begin{itemize}
    \item The first term $O_{\pi[\vec{x}^*_J \rightarrow \vec{y}^*_J]} U_{t+1} \ket{\phi^{(t)}_{\vec{v}, \vec{b}, \vec{c}}}$ corresponds to the case that there is no measure-and-reprogram happened for the $(t+1)$-th query, which intuitively means that after the $(t +1)$-th query, the sequences $\vec{v}, \vec{b}, \vec{c}$ are going to remain unchanged.
    %\revise{
    That is, we have
    \[
    O_{\pi[\vec{x}^*_J \rightarrow \vec{y}^*_J]} U_{t+1} \ket{\phi^{(t)}_{\vec{v}, \vec{b}, \vec{c}}}=\ket{\phi^{(t+1)}_{\vec{v}, \vec{b}, \vec{c}}}.
    \]
    %}
    \item The second and the third terms %$\sum_{j \not\in J} O_{\pi[\vec{x}^*_J \rightarrow \vec{y}^*_J]}  \ketbra {x_j^{\sf hit}}{x_j^{\sf hit}} U_{t+1}\ket{\phi^{(t)}_{\vec{v}, \vec{b}, \vec{c}}}$, \\
    %$\sum_{j \not\in J} O_{\pi[\vec{x}^*_J \rightarrow \vec{y}^*_J]}  \ketbra {x_j^{\sf miss}}{x_j^{\sf miss}}  U_{t+1} \ket{\phi^{(t)}_{\vec{v}, \vec{b}, \vec{c}}}$ 
    correspond to the case that a measure-and-reprogram occurs for the $(t+1)$-th query, but the reprogramming is done after the measurement. Specifically, for each $j\notin J$ and $d\in \{0,1\}$, 
let $\vec{v}_{j\rightarrow t+1}$ be the sequence obtained by replacing the $j$-th entry of $\vec{v}$ (which must be $\bot$) with $t+1$, 
let $\vec{b}_{j\rightarrow d}$ be the sequence obtained by replacing the $j$-th entry of $\vec{b}$ (which must be $\bot$) with $d$.
and let $\vec{c}_{j\rightarrow d}$ be the sequence obtained by replacing the $j$-th entry of $\vec{c}$ (which must be $\bot$) with $d$.
%\revise{
Then for any $j\notin J$, 
    we have 
    \begin{align*}
 &O_{\pi[\vec{x}^*_J \rightarrow \vec{y}^*_J]}  \ketbra {x_j^{\sf hit}}{x_j^{\sf hit}} U_{t+1}\ket{\phi^{(t)}_{\vec{v}, \vec{b}, \vec{c}}}=\ket{\phi^{(t+1)}_{\vec{v}_{j\rightarrow t+1}, \vec{b}_{j\rightarrow 0}, \vec{c}_{j\rightarrow 1}}},\\
 & O_{\pi[\vec{x}^*_J \rightarrow \vec{y}^*_J]}  \ketbra {x_j^{\sf miss}}{x_j^{\sf miss}} U_{t+1}\ket{\phi^{(t)}_{\vec{v}, \vec{b}, \vec{c}}}=\ket{\phi^{(t+1)}_{\vec{v}_{j\rightarrow t+1}, \vec{b}_{j\rightarrow 1}, \vec{c}_{j\rightarrow 1}}}.
    \end{align*}
    This is because under the inserted projections corresponding to $\ket{\phi^{(t)}_{\vec{v}, \vec{b}, \vec{c}}}$, the oracle kept by $S[\cA,\pi,\pi^*]$ after the $t$-th query is $\pi[\vec{x}^*_J \rightarrow \vec{y}^*_J]$ by \Cref{lem:quantum_correctness_of_reprogramming}.  
    %\takashi{I adapted the explanation here accordingly.}
 %   }
    \item The fourth and the last terms correspond to the case that a measure-and-reprogram occurs for the $(t+1)$-th query, and the reprogramming is done before the measurement.
   % \revise{
    Similarly to the above item, 
    for any $j\notin J$, 
    we have 
    \begin{align*}
 &O_{\pi[\vec{x}^*_{J\cup\{j\}} \rightarrow \vec{y}^*_{J\cup\{j\}}]}  \ketbra {x_j^{\sf hit}}{x_j^{\sf hit}} U_{t+1}\ket{\phi^{(t)}_{\vec{v}, \vec{b}, \vec{c}}}=\ket{\phi^{(t+1)}_{\vec{v}_{j\rightarrow t+1}, \vec{b}_{j\rightarrow 0}, \vec{c}_{j\rightarrow 0}}},\\
 & O_{\pi[\vec{x}^*_{J\cup\{j\}} \rightarrow \vec{y}^*_{J\cup\{j\}}]}  \ketbra {x_j^{\sf miss}}{x_j^{\sf miss}} U_{t+1}\ket{\phi^{(t)}_{\vec{v}, \vec{b}, \vec{c}}}=\ket{\phi^{(t+1)}_{\vec{v}_{j\rightarrow t+1}, \vec{b}_{j\rightarrow 1}, \vec{c}_{j\rightarrow 0}}}.
    \end{align*}
This is because if $b_j=0$ and the measured $(t+1)$-th query is $x_j^{\sf hit}$
or 
$b_j=1$ and the measured $(t+1)$-th query is $x_j^{\sf miss}$, then $S[\cA,\pi,\pi^*]$ reprograms the oracle as $x_j^*\rightarrow y_j^*$
by \Cref{lem:quantum_correctness_of_reprogramming}.
 %   }
\end{itemize}

%\revise{
By the above argument, we can see that $\left| {\psi^{{\pi[\vec{x}^* \rightarrow \vec{y}^*]}}_{t+1}} \right \rangle$ can be written as a sum of states of the form $\pm \ket{\phi^{(t+1)}_{\vec{v}', \vec{b}', \vec{c}'}}$ for $(\vec{v}', \vec{b}', \vec{c}')$ that satisfy the required conditions ($\vec{v}'\in ([t + 1]\times \{\bot\})^k$, $\vec{b}'\in \{0,1,\bot\}^k$, $\vec{c}'\in \{0,1,\bot\}^k$, and they satisfy the conditions in \Cref{item:vbc_condition} of \Cref{def:quantum_simulator}).\footnote{We can easily see that $(\vec{v}', \vec{b}', \vec{c}')$ 
 satisfies the required condition for the $(t+1)$-th query by the induction hypothesis that 
 $(\vec{v}, \vec{b}, \vec{c})$
 satisfies the required condition for the $t$-th query. 
 } 
Moreover, it is easy to see that each term may appear only at most once in the sum.   
Thus, the induction hypothesis holds for the state after the $(t+1)$-th query, i.e., $\left| {\psi^{{\pi[\vec{x}^* \rightarrow \vec{y}^*]}}_{t+1}} \right \rangle$ can be decomposed as in \Cref{eq:state_decopose_middle}. 
This completes the proof of \Cref{eq:state_decomposition}. 

 \paragraph*{Bounding the loss.}
%\revise{
For $(x_1,...,x_k,z)\in X^k\times Z$, 
let $\Pi_{x_1,...,x_k,z}$ be the projector that projects the output of $\cA$ onto $\ket{x_1,...,x_k,z}$. 
By \Cref{eq:state_decomposition} and Triangle and
Cauchy-Schwarz inequalities, 
\begin{align*}
    \left\|\Pi_{x_1,...,x_k,z}\left| {\phi^{{\pi[\vec{x}^* \rightarrow \vec{y}^*]}}_{2q}} \right \rangle\right\|^2
    \le   
   (8q+1)^{k}\sum_{\vec{v},\vec{b},\vec{c}} \left\|\Pi_{x_1,...,x_k,z}\ket{\phi_{\vec{v}, \vec{b}, \vec{c}}}\right\|^2
\end{align*} 
where we used the fact that there are at most $(8q+1)^k$ possible choices for $(\vec{v},\vec{b},\vec{c})$.\footnote{For each $j\in [k]$, we have 
$(v_j,b_j,c_j)\in ([2q]\times \{0,1\}\times \{0,1\})\cup \{(\bot,\bot,\bot)\}$, and thus
there are at most $(8q+1)$ possibilities.} 
Since $\ket{\phi_{\vec{v}, \vec{b}, \vec{c}}}$ corresponds to the final state of $S[\cA,\pi,\pi^*]$ 
for the fixed choices of $(\vec{v},\vec{b},\vec{c})$ with some inserted projections, 
the above implies that 
\[
\Pr[(x_1,...,x_k,z)\leftarrow \cA^{\pi^*}]
\le (8q+1)^{2k}\Pr[(x_1,...,x_k,z)\leftarrow S[\cA,\pi,\pi^*]]. 
\]
Since this holds for all $(x_1,...,x_k,z)$, 
\Cref{lemma:quantum_measure_reprogram} follows. 
\if0
Consider the state decomposition of the algorithm $\cA$ in \Cref{eq:state_decomp}, which shows that after all 
%\revise{
$2q$
%}
queries, the state of $\cA$ can be decomposed into a summation of subnormalized states; each corresponds to one unique sequence of measure-and-reprogram on $\vec{x}^*, \vec{y}^*$; i.e., with respect to a different set of parameters $\vec{v}, \vec{b}, \vec{c}$ in the execution of the simulator $S[\cA,\pi,\pi^*]$. 
Since each $x^*_j, y^*_j$ can be either not reprogrammed at all, or being reprogrammed in one of these $2q$ forward/backward queries (either before or after the query and as a hit or miss query), there are at most $(8q+1)^{k}$ such possibilities and $(8q+1)^{k}$ states.

The conclusion of \Cref{lemma:quantum_measure_reprogram} follows directly due to the Triangle and Cauchy-Schwarz inequalities. \fi
\end{proof}
We then use \Cref{lemma:quantum_measure_reprogram} to prove \Cref{thm:quantum_lifting}. 
\begin{proof}[Proof of \Cref{thm:quantum_lifting}]
 We define $\cB^{\pi^*}$ as an algorithm that runs $S[\cA,\pi,\pi^*]$ for a uniformly random $\pi$. 
 %Relying on the quantum measure-and-reprogram \Cref{lemma:quantum_measure_reprogram}, 
%Using the same idea in the classical case
 %when summing over all possible $\vec{x}^*$ and $(\pi, \pi^*) \in G[\vec{x}^*]$, we have:
Then we have: 
\begin{align*}
&\Pr_{\pi^*} \left[
(x_1,...,x_k,\pi^*(x_1),...,\pi^*(x_k),z)\in R
:(x_1,...,x_k,z)\leftarrow \cB^{\pi^*}\right]\\
%&\revise{=}
&=
\sum_{(x^*_1,...,x^*_k)}\Pr_{\pi,\pi^*}\left[
\begin{array}{ll}
(x^*_1,...,x^*_k,y^*_1,...,y^*_k,z)\in R \\
~\land ~
\forall j\in[k]~x_j=x^*_j 
\end{array}
:(x_1,...,x_k,z)\leftarrow S[\cA, \pi, \pi^*]\right]\\
&\ge
\sum_{(x^*_1,...,x^*_k)}
\Pr_{\pi,\pi^*}[(\pi,\pi^*)\in G[\vec{x}^*]]\\
&\cdot 
\Pr_{(\pi,\pi^*)\leftarrow G[\vec{x}^*]}\left[
\begin{array}{ll}
(x^*_1,...,x^*_k,y^*_1,...,y^*_k,z)\in R \\
~\land ~
\forall j\in[k]~x_j=x^*_j\} 
\end{array}
:(x_1,...,x_k,z)\leftarrow S[\cA, \pi, \pi^*]\right]\\
&\ge
\sum_{(x^*_1,...,x^*_k)}
\left(1-\frac{k^2}{|X|}\right)\\
&\cdot 
\Pr_{(\pi,\pi^*)\leftarrow G[\vec{x}^*]}\left[
\begin{array}{ll}
(x^*_1,...,x^*_k,y^*_1,...,y^*_k,z)\in R \\
~\land ~
\forall j\in[k]~x_j=x^*_j\} 
\end{array}
:(x_1,...,x_k,z)\leftarrow S[\cA, \pi, \pi^*]\right]\\
&\ge 
\sum_{(x^*_1,...,x^*_k)}
\left(1-\frac{k^2}{|X|}\right)
\frac{1}{(8q+1)^{2k}}\\
&\cdot \Pr_{(\pi,\pi^*)\leftarrow G[\vec{x}^*]}\left[
\begin{array}{ll}
(x^*_1,...,x^*_k,y^*_1,...,y^*_k,z)\in R\\
~\land ~
\forall j\in[k]~x_j=x^*_j
\end{array}
:(x_1,...,x_k,z)\leftarrow \cA^{\pi[\vec{x}^*\rightarrow \vec{y}^*]}\right]\\
&= \sum_{(x^*_1,...,x^*_k)}\left(1-\frac{k^2}{|X|}\right)
\frac{1}{(8q+1)^{2k}}\Pr_{\pi}\left[
\begin{array}{ll}
(x^*_1,...,x^*_k,\pi(x^*_1),...,\pi(x^*_k),z)\in R\\
~\land ~
\forall j\in[k]~x_j=x^*_j
\end{array}
:(x_1,...,x_k,z)\leftarrow \cA^{\pi}\right]\\
&= \left(1-\frac{k^2}{|X|}\right)
\frac{1}{(8q+1)^{2k}}\Pr_{\pi}\left[
(x_1,...,x_k,\pi(x_1),...,\pi(x_k),z)\in R
:(x_1,...,x_k,z)\leftarrow \cA^{\pi}\right]
\end{align*}
%\revise{
where 
the second inequality follows from \Cref{lem:bad_prob}, 
the third inequality follows from \Cref{lemma:quantum_measure_reprogram}, 
and  the second-to-last equality follows from \Cref{lem:uniform}. 
This completes the proof of \Cref{thm:quantum_lifting}.
%} 
\end{proof}

 \section{Quantum Lifting Theorem for Ideal Ciphers} \label{app:ideal_cipher}

In this section, we generalize our lifting theorem to ideal ciphers.  We may show the following theorem given in the Introduction.

\begin{customtheorem}{\ref{thm:quantum_icm_lifting_informal}}[Quantum Lifting Theorem on Ideal Ciphers]
Let $\mathcal{G}$ be an (interactive) search game with a classical challenger $\cC$ that performs at most $k$ queries to an ideal cipher oracle $E:\Ks\times X\to X$, and let $\cA$ be an algorithm that performs $q$ quantum queries to the ideal cipher.
Then there exists an adversary $\cB$ making at most $k$ classical queries to $E$
%$k$-query classical adversary $\cB$
such that: 
\begin{align*}
\Pr[\cB \textrm{ wins } \mathcal{G}] \geq 
\frac{\left(1 - \frac{k^2}{|X|}\right)}{(8q+1)^{2k}}
\Pr[\cA \text{ wins } \mathcal{G}].
\end{align*} 
\end{customtheorem}

Since the proof and the proof methodology are very similar to those of random permutations, instead of providing a formal proof, we will give a brief sketch how we can adapt the proofs for random permutations to ideal ciphers.

Just like the random permutation case, we rely on reprogramming of ciphers to prove Theorem~\ref{thm:quantum_icm_lifting_informal}.  First, let us define some terminologies.

Given two finite sets $\Ks, X$, a \emph{cipher} over $\Ks$ and $X$ is a function $E:\Ks\times X\to X$ such that for each $K\in\Ks$, the function $E_K:X\to X$ defined by $E_K(x):=E(K, x)$ is invertible.  We then define the inverse cipher $E^{-1}$ as the cipher defined by $E^{-1}(K, x):=E_K^{-1}(x)$.

An \emph{ideal cipher} over $\Ks$ and $X$ is a cipher chosen uniform randomly over the set of all ciphers over $\Ks$ and $X$.  When an algorithm has oracle access to a cipher $E$, it may make both forward ($E$) and backward ($E^{-1}$) queries.  Also, a quantum algorithm may make quantum superposition queries to a cipher oracle $E$ via the following unitary $U_E$:
\[
U_E\ket{b}\ket{K}\ket{x}\ket{y}:=
\begin{cases}
    \ket{b} \otimes O_E(\ket{K}\ket{x}\ket{y}) & \text{if $b=0$,}\\
    \ket{b} \otimes O_{E^{-1}}(\ket{K}\ket{x}\ket{y}) & \text{if $b=1$,}
\end{cases}
\]
with $O_E\ket{K}\ket{x}\ket{y}:=\ket{K}\ket{x}\ket{y+E_K(x)}$ and $O_{E^{-1}}\ket{K}\ket{x}\ket{y}:=\ket{K}\ket{x}\ket{y+E_K^{-1}(x)}$.

Just like Lemma~\ref{lem:quantum_algo}, if $\cA$ is a $q$-query algorithm having quantum access to a cipher oracle, then we may assume that it is in the normal form and makes $2q$ queries to $O_E$ and $O_{E^{-1}}$, alternatingly.

Now, we define how we can reprogram a cipher, using reprogramming of a permutation, as follows.

\begin{definition}
    \label{def:cipher_reprog}
    Let $E:\Ks\times X\to X$ be a cipher, and $(K, x, y)\in \Ks\times X\times X$ be an arbitrary tuple.  Then we denote the reprogramming of $E$ by $(K, x, y)$ as:
    \begin{equation}
        E[x \tok{K} y](K', x') = 
        \begin{cases}
                     E_K[x\to y](x')   & \text{if $K'=K$,}\\
                     E(K', x')       & \text{otherwise.}
        \end{cases}
    \end{equation}
\end{definition}

Essentially, we may understand a cipher as a function which maps a key $K$ to a permutation $E_K$, and an ideal cipher as a function which maps a key $K$ to an independent random permutation $E_K$.  Here, we see that reprogramming occurs with respect to a key $K$: $E[x\tok{K} y]$ leaves all `components' of $E$ unchanged, except that it reprograms the permutation $E_K$ to $E_K[x\to y]$.

Also, given $\vec{K}=(K_1, \dots, K_k)$, $\vec{x}=(x_1, \dots, x_k)$, $\vec{y}=(y_1, \dots, y_k)$, we may denote by $E[\vec{x}\tok{\vec{K}}\vec{y}]$ the reprogrammed cipher
\[
E[x_1\tok{K_1}y_1]\dots [x_k\tok{K_k}y_k].
\]

In order to prove the lifting theorem for ideal ciphers, we need to define when a $k$-tuple input $(t_1^\ast, \dots, t_k^\ast)$ of tuples $t_i^\ast=(K_i^\ast, x_i^\ast, y_i^\ast)$ is \emph{good} with respect to a cipher $E$.  

\begin{definition}[Good tuples]
\label{def:good_ic}
Suppose we have a $k$-tuple $t^\ast=(t_1^\ast, \dots, t_k^\ast)$ of tuples $t_i^\ast=(K_i^\ast, x_i^\ast, y_i^\ast)$.
Then we say that the tuple $t^\ast$ is \emph{good} with respect to a cipher $E$, if the following hold:
\begin{itemize}
    \item For any $1\leq i<j\leq k$, if $K_i^\ast=K_j^\ast$, then we have $x_i^\ast\neq x_j^\ast$ and $y_i^\ast\neq y_j^\ast$.  (In other words, $x_i^\ast$ which share the same key are distinct, and so are $y_i^\ast$.)
    \item For any $i, j\in [k]$, $E_{K_i^\ast}(x_i^\ast)\neq y_j^\ast$.
\end{itemize}
\end{definition}

\begin{definition}[Good pairs of ciphers] \label{def:support_G_ic}
For any vectors $\vec{K}^\ast=(K_1^\ast, \dots, K_k^\ast)\in \Ks^k$ and $\vec{x}^\ast=(x_1^\ast, \dots, x_k^\ast)\in X^k$ satisfying
\[
\forall i<j, K_i^\ast=K_j^\ast \implies x_i^\ast\neq x_j^\ast,
\]
we define $G[\vec{K}^\ast, \vec{x}^\ast]$ as the set consisting of all pairs $(E, E^\ast)$ such that the tuple
$(t_1^\ast, \dots, t_k^\ast)$ with $t_i^\ast=(K_i^\ast, x_i^\ast, E^\ast(K_i^\ast, x_i^\ast))$
is good with respect to $E$.
\end{definition}

Just as in the case of permutations, we may show that if we pick ciphers randomly, then with high probability we get good pairs of ciphers.

\begin{lemma}[Bad probability for ciphers]\label{lem:bad_prob_ic}
Let $E$ be a (fixed) permutation and let $\vec{K}^\ast=(K_1^\ast, \dots, K_k^\ast)\in\Ks^k$ and $\vec{x}^\ast=(x_1^\ast, \dots, x_k^\ast)\in X^k$
be (fixed) vectors satisfying
\[
\forall i<j, K_i^\ast=K_j^\ast \implies x_i^\ast\neq x_j^\ast.
\]
Then we have 
\[
\Pr_{E^\ast}[(E, E^\ast)\notin G[\vec{K}^\ast,\vec{x}^*]]\le \frac{k^2}{|X|}
\]
\end{lemma}

In fact, we may show that the bad probability is maximized when all keys are identical, $K_1^\ast=\dots=K_k^\ast$, and in that case we have the same upper bound as in Lemma~\ref{lem:bad_prob}.

When we prove the lifting theorem for ideal ciphers, we follow essentially the same analysis as in the lifting theorem for permutations.  For this, we need to define the hit and miss inputs, incorporating the cipher keys as follows.

\begin{definition}[Hit and Miss queries for ciphers]
\label{def:hit_miss_queries_ic}
Fix any $\vec{K}^\ast=(K_1^\ast, \dots, K_k^\ast)\in\Ks^k$ and $\vec{x}^\ast=(x_1^\ast, \dots, x_k^\ast)\in X^k$ satisfying
\[
\forall i<j, K_i^\ast=K_j^\ast \implies x_i^\ast\neq x_j^\ast,
\]
ciphers $E, E^\ast\in G[\vec{K}^\ast, \vec{x}^\ast]$, and $\vec{y}^\ast = (y_1^\ast, \dots, y_k^\ast)$ with $y_i^\ast=E_{K_i^\ast}^\ast(x_i^\ast)$. 
For the $j$-th tuple $(K_j^\ast, x_j^\ast, y_j^\ast)$, we define the Hit and Miss input for forward queries as follows:
    \begin{align*}
        (K_j^\textsf{hit}, x_j^\textsf{hit}) &:= (K_j^\ast, x_j^\ast), \\
        (K_j^\textsf{miss}, x_j^\textsf{miss}) &:= (K_j^\ast, E_{K_j^\ast}^{-1}(y_j^\ast)).
    \end{align*}
    %We define $X^{\sf hit} := \{x_j^{\sf hit}, j = 1, 2, \ldots, k\}$ and $X^{\sf miss} := \{x_j^{\sf miss}\} - X^{\sf hit}$.

    Similarly, for the $j$-th tuple $(K_j^\ast, x_j^\ast, y_j^\ast)$, we define the Hit and Miss input for backward queries as follows:
    \begin{align*}
        (K_j^\textsf{hit}, y_j^\textsf{hit}) &:= (K_j^\ast, y_j^\ast), \\
        (K_j^\textsf{miss}, y_j^\textsf{miss}) &:= (K_j^\ast, E_{K_j^\ast}(x_j^\ast) ).
    \end{align*}
\end{definition}

Now, we may prove the following lemma, which is essentially Lemma~\ref{lemma:quantum_measure_reprogram} for ciphers.

\begin{lemma}[Quantum Measure-and-Reprogram Lemma for Ciphers]\label{lemma:quantum_measure_reprogram_ic}
Let $\cA$ be a $q$-query quantum algorithm, and $E, E^\ast$ be ciphers over $\Ks$ and $X$.  And let $\vec{K}^\ast=(K_1^\ast, \dots, K_k^\ast)\in\Ks^k$,  $\vec{x}^\ast = (x_1^\ast, \dots, x_k^\ast)\in X^k$ be vectors satisfying
\[
\forall i<j, K_i^\ast=K_j^\ast \implies x_i^\ast\neq x_j^\ast,
\]
such that $(E, E^\ast)\in G[\vec{K}^\ast, \vec{x}^\ast]$, $\vec{y}^\ast := (y_1^\ast, \dots, y_k^\ast)$ with $y_i^\ast:=E_{K_i^\ast}^\ast(x_i^\ast)$, and 
$R\subseteq \Ks^k\times X^k \times X^k \times Z$ be a relation.
Then, we have:
    \begin{align*}
    &\Pr\left[
\begin{array}{ll}
(K_1^\ast, \dots, K_k^\ast, x_1^\ast, \dots, x_k^\ast, y_1^\ast, \dots, y_k^\ast, z)\in R \\
~\land ~
\forall j\in[k]~K_j=K^*_j
~\land ~
\forall j\in[k]~x_j=x^*_j\\
\end{array}
:(K_1, \dots, K_k, x_1,...,x_k,z)\leftarrow S[\cA, E, E^\ast]\right]\\
&\ge  
\frac{1}{(8q+1)^{2k}}
\Pr\left[
\begin{array}{ll}
(K_1^\ast, \dots, K_k^\ast, x_1^\ast, \dots, x_k^\ast, y_1^\ast, \dots, y_k^\ast, z)\in R \\
~\land ~
\forall j\in[k]~K_j=K^*_j
~\land ~
\forall j\in[k]~x_j=x^*_j\\
\end{array} \right. \\
& \left. \ \ \ \ \ \ \ \ \ \ \ \ \ \ \ \ \ \ \ \ \ \ \ \ \ \ \ \ \ \ \ \ :(K_1, \dots, K_k, x_1,...,x_k,z)\leftarrow \cA^{E[\vec{x}^\ast \tok{\vec{K}^\ast} \vec{y}^\ast]}\right]. 
\end{align*}
\end{lemma}

In Lemma~\ref{lemma:quantum_measure_reprogram}, we had a simulator $S[\cA, \pi, \pi^\ast]$ which tries to guess the first query of $\cA$ which touches $(x_j^\ast, y_j^\ast)$, for each $j\in [k]$.  For Lemma~\ref{lemma:quantum_measure_reprogram_ic}, we define a similar simulator $S[\cA, E, E^\ast]$ which works essentially the same as $S[\cA, \pi, \pi^\ast]$ of Lemma~\ref{lemma:quantum_measure_reprogram}, except that cipher keys are also involved, and we use the definition of the hit and miss queries given in Definition~\ref{def:hit_miss_queries_ic}.  Then, the proof is essentially a duplicate of that of Lemma~\ref{lemma:quantum_measure_reprogram}.

Finally, using Lemma~\ref{lemma:quantum_measure_reprogram_ic}, we prove the lifting theorem for ideal ciphers.

\begin{theorem}[Quantum Lifting Theorem for Ideal Ciphers]
\label{thm:quantum_lifting_ic}
Let $\cA$ be a quantum algorithm that makes $q$ quantum queries to an ideal cipher oracle over $\Ks$ and $X$, and $R$ is a relation on $\Ks^k\times X^k \times X^k \times Z$. 
Then there exists an algorithm $\cB$ making at most $k$ classical queries such that 
\begin{align*}
&\Pr_{E^\ast} \left[
(K_1, \dots, K_k, x_1,...,x_k, E_{K_1}^\ast(x_1),...,E_{K_k}^\ast(x_k),z)\in R
:(K_1, \dots, K_k, x_1,...,x_k,z)\leftarrow \cB^{E^\ast}\right]\\
&\ge
\frac{\left(1 - \frac{k^2}{|X|}\right)}{(8q+1)^{2k}}\Pr_{E^\ast} \left[
(K_1, \dots, K_k, x_1,...,x_k, E_{K_1}^\ast(x_1),...,E_{K_k}^\ast(x_k),z)\in R: \right. \\
& \ \ \ \ \ \ \ \ \ \ \ \ \ \ \ \ \ \ \ \ \ \ \ \ \ \ \ \ \ \
\left. (K_1, \dots, K_k, x_1,...,x_k,z)\leftarrow \cA^{E^\ast}\right].
\end{align*}
\end{theorem}

Again, the proof is almost a duplicate of that of Theorem~\ref{thm:quantum_lifting}.  We use Lemma~\ref{lemma:quantum_measure_reprogram_ic} instead of Lemma~\ref{lemma:quantum_measure_reprogram}, and we use the bound for the bad probability for ciphers given in Lemma~\ref{lem:bad_prob_ic}.

Also, the interactive case can be extended to the ideal ciphers, in a similar way.
 \section{Applications} \label{app:applications}
%\takashi{I removed $\pi^{-1}$ from the superscript for consistency to the other sections. We need to make sure that the notation $A^\pi$ means that $A$ has access to both $\pi$ and $\pi^{-1}$.}
In this section, we discuss applications of our lifting theorem (\Cref{thm:quantum_lifting}). 

\subsection{Generalized Double-Sided Search}
We give an improved bound for the generalized double-sided search problem considered in \cite[Theorem 6.11]{MMW24} as a generalization of Unruh's double-sided zero-search conjecture~\cite{DBLP:conf/asiacrypt/Unruh23}.   
\begin{theorem}[Generalized Double-Sided Search]\label{thm:generalized_double_search}
Let 
$N\ge 2$ be an integer, $\cA$ be a quantum algorithm that makes $q$ quantum queries to a uniformly random permutation $\pi$ on $[N]$ and its inverse, and $R\subseteq [N]\times [N]$ be an arbitrary relation. Then it holds that 
\[
\Pr_{\pi}\left[
(x,\pi(x))\in R:
x \leftarrow \cA^{\pi}
\right]
\le 
\frac{8\cdot(8q+1)^{2} \cdot r_{\mathrm{max}}}{N}
\]
where $r_{\mathrm{max}}=\max\{\max_{x} |R_x|,\max_{y}|R_y^{\mathrm{inv}}|\}$ with $R_x=\{y:(x,y)\in R\}$ and $R_y^{\mathrm{inv}}=\{x:(x,y)\in R\}$.
\end{theorem}
\begin{proof}
   By \Cref{thm:quantum_lifting}, there is a one-classical-query algorithm $\cB$ such that 
   \[
   \Pr_\pi\left[
(x,\pi(x))\in R:
x \leftarrow \cA^{\pi}
\right]\\
\le \frac{(8q+1)^{2}}{(1-\frac{1}{N})}\Pr_\pi\left[
(x,\pi(x))\in R:
x \leftarrow \cB^{\pi}
\right].
   \]
To analyze $\cB$'s success probability, suppose that $\cB$'s query forms an input-output pair $(x^*,y^*)$, that is, $\cB$ queries $x^*$ to $\pi$ or queries $y^*$ to $\pi^{-1}$. If $\cB$ outputs $x=x^*$, $(x,\pi(x))\in R$ holds with probability at most $r_{\mathrm{max}}/N$. 
If $\cB$ outputs $x\ne x^*$, then $(x,\pi(x))\in R$ holds with probability at most $\max_{x}|R_x|/(N-1)$. By the union bound, $(x,\pi(x))\in R$ holds with probability at most $2r_{\mathrm{max}}/(N-1)\le 4r_{\mathrm{max}}/N$. 
%\takashi{
% If we could improve the lifting theorem to further ensure that $\cB$'s  query should be equal to the challenger's query, then the probability could be improved to $r_{\mathrm{max}}/N$. 
%}
%$\frac{r_{\mathrm{max}}}{N-1}$. 
  Noting that $1-1/N\ge 1/2$, we obtain the theorem.
\end{proof}
\Cref{thm:generalized_double_search} improves upon the bound of $\frac{914 r_{\mathrm{max}}\cdot q^3\cdot(\ln(N)+2)}{N}$ established by \cite{MMW24}. 
In particular, ours is tight up to a constant factor since it matches a straightforward algorithm based on Grover's search~\cite{grover1996fast}.

%\takashi{I tried to write a multi-input version, but I couldn't see how to upper bound the success probability of $\cB$ in a clean manner. For example, simply defining 
%$r_{\mathrm{max}}$ to be the maximum probability that $(x_1,...,x_k,\pi(x_1),...,\pi(x_k))\in R$ or $(\pi^{-1}(y_1),...,\pi^{-1}(y_k),y_1,...,y_k)\in R$ doesn't work.  
%} 

As an application of \Cref{thm:generalized_double_search}, we provide an alternative proof of the correctness of Unruh's double-sided zero-search conjecture.  
\begin{corollary}[Double-Sided Zero-Search] \label{cor:Unruh_double_sided}
Let $n$ be a positive integer and
$\cA$ be a quantum algorithm that makes $q$ quantum queries to a uniformly random permutation $\pi$ on $\{0,1\}^{2n}$ and its inverse. Then it holds that 
\[
\Pr_{\pi}\left[
\pi(x||0^n)=y||0^n:
(x,y) \leftarrow \cA^{\pi}
\right]
\le 
\frac{8\cdot (8q+1)^{2}}{2^{n}}.
\]
\end{corollary}
\begin{proof}
    Define relation $R:=\{x\in \{0,1\}^n: \exists~y\in \{0,1\}^n~s.t.~\pi(x||0^n)=y||0^n\}$. 
    Then we clearly have $r_{\mathrm{max}}=2^n$. 
    Substituting this into \Cref{thm:generalized_double_search} yields \Cref{cor:Unruh_double_sided}.
\end{proof}
The bound in \Cref{cor:Unruh_double_sided} is tight up to a constant factor, though \cite{CP24} provides a slightly better constant factor; their bound is   $50(q+1)^2/2^{n}$.

As another application of \Cref{thm:generalized_double_search}, we show an improved bound for the fixed point finding problem studied in \cite{DBLP:conf/asiacrypt/HosoyamadaY18}.
\begin{corollary}[Fixed Point Finding]\label{cor:fixed_point} 
Let 
$N$ be a positive integer and $\cA$ be a quantum algorithm that makes $q$ quantum queries to a uniformly random permutation $\pi$ on $[N]$ and its inverse. Then it holds that 
\[
\Pr_\pi\left[
\pi(x)=x:
x \leftarrow \cA^{\pi}
\right]
\le 
\frac{8\cdot(8q+1)^{2}}{N}.
\]
\end{corollary}
\begin{proof}
    Define relation $R:=\{x\in \{0,1\}^n: x=\pi(x)\}$. 
    Then we clearly have $r_{\mathrm{max}}=1$. 
    Substituting this into \Cref{thm:generalized_double_search} yields \Cref{cor:fixed_point}.
\end{proof}
The bound in \Cref{cor:fixed_point} is tight up to a constant factor, and improves the existing bound of $O(q/\sqrt{N})$ shown by \cite{DBLP:conf/asiacrypt/HosoyamadaY18}. 

Similarly to \cite{MMW24}, we can also use \Cref{thm:generalized_double_search} to show preimage-resistance of single-round sponge, but we omit this since we discuss (not necessarily single-round) sponge in detail in \Cref{sec:sponge}.

\subsection{Sponge Construction}\label{sec:sponge}
%We show that \Cref{thm:quantum_lifting} enables us to lift \emph{classical} indifferentiability into some form of \emph{quantum} security  up to a certain security loss. We demonstrate it for the case of the sponge construction, but a similar argument works for any construction that satisfies classical indifferentiability in the invertible permutation model (or ideal cipher model). \takashi{For example. sum of permutations \cite{GBJ+23}.}

The sponge construction~\cite{sponge} is a hashing algorithm that underlies SHA-3.

The sponge construction makes use of a permutation $\pi$ on $\{0,1\}^{r+c}$, where $r$ and $c$ are positive integers called \emph{rate} and \emph{capacity}, respectively. We fix $r$ and $c$ throughout this subsection.
While the sponge construction is designed to support unbounded-length inputs and outputs, we focus on fixed input and output lengths, as our security bounds degrade with them.  
For input length $m$ and output length $n$, 
%we define the number of input blocks $\ell_a:= \lceil \frac{m+1}{r}\rceil$ and the number of output blocks $\ell_s:= \lceil \frac{n}{r}\rceil$. 
%Then 
the sponge construction, %with $n$-bit inputs and $m$-bit outputs, 
which we denote by $\sponge^{\pi}[m,n]$, works as follows (see also \Cref{fig:sponge}):\footnote{The sponge construction supports more general padding functions. We adopt the simplest one that simply appends $1||0...0$.}\footnote{\Cref{fig:sponge} is made based on the one taken from \cite{TikZ:for:Cryptographers}.}

\begin{figure}
\caption{Sponge construction}
\label{fig:sponge}
\begin{center}
    
\begin{tikzpicture}[scale=0.4]

\tikzset{SpongePerm/.style=rounded corners=4pt,};
\tikzset{edge/.style=-latex new, arrow head=8pt, thick};
\tikzset{edgee/.style=latex new-latex new, arrow head=8pt, thick};

\path (8.5, -1) node {\large\bf Absorbing phase};
\path (23.5, -1) node {\large\bf Squeezing phase};

\begin{scope}[xshift=0cm]
  \draw[thick] (0,0) rectangle ++(1,10);
  \draw[thick] (0,3) -- ++(1,0);

  \node[XOR,thick] (xm1) at (1+1.5,8) {};
  \draw[edge,thick] (1,8) -- (xm1);
  \draw[edge,thick] (1,2) -- ++(3,0);
  \draw[edge,thick] (1+1.5,10.5) node[above] {\large $x_{1}$} -- (xm1);  
  \draw[edge,thick] (xm1) -- ++(1.5,0);

	\draw[edgee,anchor=east] (-1,0) -- node[left] {$c$ bits} ++(0,3);
	\draw[edgee,anchor=east] (-1,3) -- node[left] {$r$ bits} ++(0,7);

\end{scope}
  
  \begin{scope}[xshift=4cm]
    \draw[SpongePerm] 
    		(0,0) rectangle node {\large$\pi$} ++(1,10);

    \node[XOR,thick] (xm2) at (1+1.5,8) {};
    \draw[edge,thick] (1,8) -- (xm2);
    \draw[edge,thick] (1,2) -- ++(3,0);  
    \draw[edge,thick] (1+1.5,10.5) node[above] {\large $x_{2}$} -- (xm2);
    \draw[edge,thick] (xm2) -- ++(1.5,0);
  \end{scope}

  \begin{scope}[xshift=8cm]
    \draw[SpongePerm] 
    		(0,0) rectangle node {\large$\pi$} ++(1,10);

    \draw[edge,thick] (1,8) -- (1+1.0,8);
    \draw[edge,thick] (1,2) -- ++(1.0,0);  
    \draw[edge,thick] (3,8) -- (3+1.0,8);
    \draw[edge,thick] (3,2) -- ++(1.0,0); 
    \node at (2.5,5) {$\ldots$};
  \end{scope}

    \begin{scope}[xshift=12cm]
    \draw[SpongePerm] 
    		(0,0) rectangle node {\large$\pi$} ++(1,10); 

            \node[XOR,thick] (xmlast) at (1+1.5,8) {};
    \draw[edge,thick] (1,8) -- (xmlast);
    \draw[edge,thick] (1,2) -- ++(3,0);  
    \draw[edge,thick] (1+1.5,10.5) node[above] {\large $x_{\ell_a}$} -- (xmlast);
    \draw[edge,thick] (xmlast) -- ++(1.5,0);
  \end{scope}

\begin{scope}[xshift=16cm]
  \draw[SpongePerm] 
  		(0,0) rectangle node {\large$\pi$} ++(1,10);

  \draw[edge,thick] (1,2) -- ++(3,0);
  \draw[edge,thick] (1,8) -- ++(3,0);  
\end{scope}

\begin{scope}[xshift=20cm]
 \draw[SpongePerm] 
  		(0,0) rectangle node {\large$\pi$} ++(1,10);

  \draw[edge,thick] (1,2) -- ++(3,0);
  \draw[edge,thick] (1,8) -- ++(3,0);  
  \draw[edge,thick] (-1,8) -- ++(0,2.5) node[above] {\large $z_{1}$};
  
\draw[dashed,thick] (-2,-2) -- ++(0,13);
\end{scope}

\begin{scope}[xshift=24cm]
  \draw[SpongePerm] 
  		(0,0) rectangle node {\large$\pi$} ++(1,10);

  \draw[edge,thick] (1,2) -- ++(1.5,0);
  \draw[edge,thick] (1,8) -- ++(1.5,0);  
  \draw[edge,thick] (-1,8) -- ++(0,2.5) node[above] {\large $z_{2}$};
  \node at (2.5,5) {$\ldots$};
\end{scope}

\end{tikzpicture}
\end{center}

\end{figure}
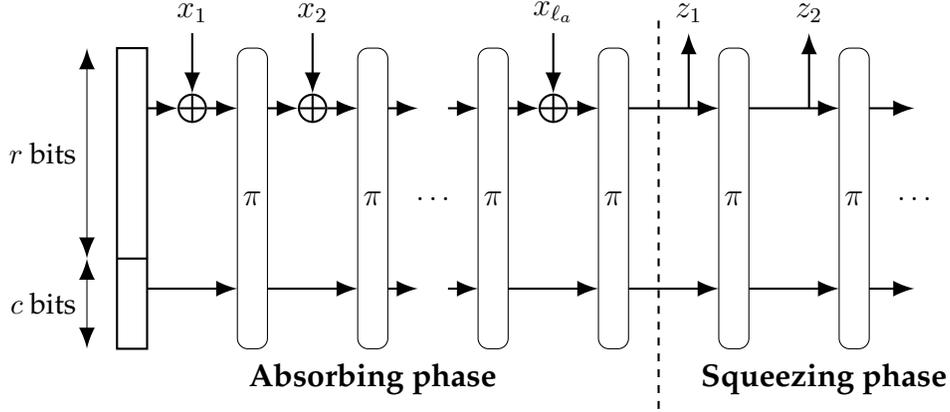

\begin{itemize}
\item[Input:] $x\in \{0,1\}^m$ 
\item Let $\ell_a:= \lceil \frac{m+1}{r}\rceil$ and  $\ell_s:= \lceil \frac{n}{r}\rceil$.
\item Parse $x||1||0^{\ell_a r-m-1}$ to $x_1||x_2||\ldots||x_{\ell_a}$ where each $x_i$ is a $r$-bit string. 
\item $s\gets 0^{r+c}$.
\item For $i=1$ to $\ell_a$, do the following:
\begin{itemize}
\item $s\gets s\oplus (x_i||0^{c})$.
\item $s\gets \pi(s)$.
\end{itemize}
\item For $i=1$ to $\ell_s$, do the following:
\begin{itemize}
\item Let $z_i$ be the first $r$-bit of $s$. 
\item $s\gets \pi(s)$.
\end{itemize}
\item Output the first $n$-bit of $z_1||z_2||...||z_{\ell_s}$.
\end{itemize}
%In the rest of this subsection, 
%$\ell_a$ and $\ell_s$ are defined as in the above algorithm (that is, $\ell_a:= \lceil \frac{m+1}{r}\rceil$ and  $\ell_s:= \lceil \frac{n}{r}\rceil$) without further notice. 
We note that $\sponge^\pi[m,n]$ calls $\pi$ a total of $\ell:=\ell_a+\ell_s-1$ times.

We prove the following lifting theorem for the sponge construction.  
\begin{theorem}[Lifting Theorem for Sponge] \label{thm:lifting_sponge}  
Let $R\subseteq \left(\{0,1\}^{m}\right)^k\times \left(\{0,1\}^{n}\right)^{k}$ be any relation.
% \takashi{I omit the $z$ part since it is redundant in the context of query-complexity.}\minki{I can't understand this comment, probably it is related to something removed?}\takashi{In our lifting theorem, we consider a relation over $(\vec{x},\pi(\vec{x}),z)$, so we could also prove a similar lifting theorem including the $z$ part here, but I omitted it for simplciity.}
Define $P_{\mathrm{max}}^{R}$ as 
\[
P_{\mathrm{max}}^{R}:= \max_{B}\Pr_{H}\left[
(x_1,...,x_k,H(x_1),...,H(x_k))\in R:
(x_1,...,x_k) \leftarrow B^H 
\right]
\]
where 
$H:\{0,1\}^m\rightarrow \{0,1\}^n$ is a uniformly random function and
the $\max$ is taken over all $k$-classical-query algorithms $B$.  
Let $A$ be a quantum algorithm that makes $q$  quantum queries to a uniformly random permutation $\pi$ on $\{0,1\}^{r+c}$ and its inverse. 
Then it holds that 
\begin{align*}
&\Pr_{\pi}\left[
(x_1,...,x_k,\sponge^\pi[m,n](x_1),...,\sponge^\pi[m,n](x_k))\in R:
(x_1,...,x_k) \leftarrow A^{\pi}
\right]\\
&\le2\cdot (8q+1)^{2k\ell}
\left(P_{\mathrm{max}}^{R}
+
\frac{(k\ell+k+1)^2}{2^c}
\right)
% = O\left(
% {(8q+1)^{2k\ell}}
% \left(P_{\mathrm{max}}^{R}
% +
% \frac{(k\ell+k+1)^2}{2^c}
% \right)
% \right)
\end{align*}
%\minki{I think the upper bound of $2{(8q+1)^{2k\ell}}\left(P_{\mathrm{max}}^{R}+\frac{(k\ell+k+1)^2}{2^c}\right)$ always holds because $(k\ell)^2 \le 2^c$which is more cleaner.}\takashi{Why does that always hold?}
% \minki{We may assume $r,c \ge 1$. Also to obtain a meaningful bound we have $(k\ell + k +1)^2/2^c\le 1$, which gives $k^2\ell^2 \le 2^c$. Then $k^2 \ell^2/2^{r+c}\le 1/2^r\le 1/2$, so the claim.}\takashi{I did so!}
where we recall that $\ell=\ell_a+\ell_s-1$ is the number of calls to $\pi$ made by $\sponge^\pi[m,n]$. 
\end{theorem}
\begin{remark}
 The above theorem is meaningful only if $(8q+1)^{2k\ell}(k\ell+k+1)^2\ll 2^c$. However, this condition does not generally hold in typical usage scenarios for hash functions, where the input length can be arbitrarily large. On the other hand, as discussed e.g., in \cite{Lefevre23},
 % \minki{I feel this is too new paper, why should we cite this one here? the sentence ``sufficient for ... Fiat-Shamir transforms and password hashing'' holds regardless of this paper I think?}\takashi{That's because I just learned those applications from that paper, and also want to emphasize that the restricted input length setting has also been studied.}
 hash functions with fixed input lengths are sufficient for certain applications, such as Fiat-Shamir transforms and password hashing.
 % \minki{And hash-based signatures like sphincs? (I honestly don't know their designs and usage)}\takashi{But I guess we anyway need unbounded-input-length hashing since we may want to first hash the message before signing?}.  
 In particular, by using the security parameter $\lambda$, if 
 $c=\Omega(\lambda)$, 
$q=\poly(\lambda)$, and 
$\ell$ and $k$ are constant, then $(8q+1)^{2k\ell}(k\ell+k+1)^2\cdot 2^{-c}=2^{-\Omega(\lambda)}$, providing a meaningful bound.       
\end{remark}
We defer the proof of \Cref{thm:lifting_sponge} to the end of this subsection. 

It is typically easy to upper bound $P_{\mathrm{max}}^{R}$ since it is just a bound in the classical random oracle model. 
In particular, we show two examples where $P_{\mathrm{max}}^{R}$ can be upper bounded by a simple term. 
\begin{lemma}[Upper Bounds of $P_{\mathrm{max}}^R$]\label{lem:upper_bound_P_max}
Let $R\subseteq \left(\{0,1\}^{m}\right)^k\times \left(\{0,1\}^{n}\right)^{k}$ be any relation, and $P_{\mathrm{max}}^{R}$ be as in \Cref{thm:lifting_sponge}. Then the following hold: 
\begin{itemize}
 \item If $k=1$, it holds that   $P_{\mathrm{max}}^{R}\le 2\max_{x}\Pr_y[(x,y)\in R]$. 
 \item If $R$ does not depend on $(x_1,...,x_k)$, i.e., $R$ can be written as $R=\left(\{0,1\}^{m}\right)^k\times R^{out}$ by using some relation $R^{out} \subseteq \left(\{0,1\}^{n}\right)^{k}$, then it holds that: \\ 
 $P_{\mathrm{max}}^{R}\le {2k \choose k}\Pr_{y_1,...,y_k}[\exists \pi~s.t.~(y_{\pi(1)},...,y_{\pi(k)})\in R^{out}]$ where $\pi$ is a permutation on $[k]$. %\takashi{It might be possible to remove ${2k \choose k}$? }
\end{itemize}
\end{lemma}
\begin{proof}~
\paragraph{First item.}
Let $B$ be an algorithm that makes a single classical query to a random function $H:\{0,1\}^m\rightarrow \{0,1\}^n$. Let $x^* \in \{0,1\}^m$ be $B$'s query. The probability that $(x^*,H(x^*))\in R$ is $\Pr_{y^*}[(x^*,y^*)\in R]$. 
If does not occur, the best choice for $B$ is to output some $x\ne x^*$. Again, the probability that $(x,H(x))\in R$ is $\Pr_{y}[(x,y)\in R]$. Thus, by the union bound, the probability that $B$'s output $x$ satisfies $(x,H(x))\in R$ is at most $2\max_{x}\Pr_y[(x,y)\in R]$. This completes the proof of the first item. 

\paragraph{Second item.} 
Let $B$ be an algorithm that makes $k$ classical queries to a random function $H:\{0,1\}^m\rightarrow \{0,1\}^n$. By modifying $B$ to query its outputs to $H$ before outputting them, we can assume that the outputs $(x_1,...,x_k)$ of $B$ are queried to $H$ if we increase the number of queries to $2k$. In this case, $B$'s output satisfies  $(x_1,...,x_k,H(x_1),...,H(x_k))\in R^{out}$ only if a permutation of a $k$-subset of its queries satisfies it. For any fixed input $(x_1,...,x_k)$, the probability that it holds is $\Pr_{y_1,...,y_k}[\exists \pi~s.t.~(y_{\pi(1)},...,y_{\pi(k)})\in R^{out}]$. Thus, the second item follows from the union bound.
\end{proof}
%Putting the first item of \Cref{lem:upper_bound_P_max} to \Cref{thm:lifting_sponge} gives correlation intractability of the sponge construction when $\ell_a$ and $\ell_s$ are constant. 

As examples, we show that the above immediately implies 
preimage-resistance and  (multi-)collision resistance. 
\begin{corollary}[Preimage-Resistance of Sponge]\label{cor:preimage_sponge}
Let $A$ be a quantum algorithm that makes $q$  quantum queries to a uniformly random permutation $\pi$ on $\{0,1\}^{r+c}$ and its inverse. 
For any 
$y\in \{0,1\}^{n}$,  
it holds that 
\[
\Pr_\pi\left[\sponge^\pi[m,n](x)=y:
x \leftarrow A^{\pi}
\right]\le (8q+1)^{2\ell}\left(\frac{4}{2^{n}}
+
\frac{2(\ell+2)^2}{2^c}
\right).
\]
\end{corollary}
\begin{proof}
    Define $R:=\{0,1\}^m\times \{y\}$. By the first (or second) item of \Cref{lem:upper_bound_P_max}, we have $P_{\mathrm{max}}^{R}\le 2\cdot 2^{-n}$. Substituting this and $k=1$ to \Cref{thm:lifting_sponge} gives the corollary. 
\end{proof}

For the single-round case where $\ell=1$, 
\Cref{cor:preimage_sponge} gives a bound of 
 $(8q+1)^{2} \left(\frac{4}{2^{n}}
+
\frac{18}{2^c}
\right).$ 
In particular, $\Omega(2^{\min\{n,c\}})$ quantum queries are needed to succeed with constant probability, which is tight due to Grover's algorithm~\cite{grover1996fast}. 
This matches the bound shown in \cite{CP24} up to a constant factor. 
\takashi{Strictly speaking, \cite{CP24} doesn't consider truncation or padding. 
\cite{CPZ24} shows indifferentiability, so it should work with truncation, but only in the case of $r\le c$. 
\cite{MMW24} works for truncated single-round construction, but theirs is non-tight. So perhaps, ours is the first to get a tight bound for the truncated single-round with $r> c$.  (But \cite{CP24} might work with truncation.)
}
More generally, when $\ell=O(1)$, $n=\Omega(\lambda)$,  $c=\Omega(\lambda)$, and $q=\poly(\lambda)$, the RHS is $2^{-\Omega(\lambda)}$.  This is the first non-trivial bound for the preimage-resistance of sponge with more than one round though it is non-tight.

Note that the above corollary does not imply one-wayness of sponge since the target output $y$ is fixed independently of $\pi$ rather than being set to be $\pi(x)$ for a random input $x$. For obtaining a bound for one-wayness, we can invoke the interactive version of the lifting lemma (\Cref{thm:lifting_interactive}), albeit with a looser bound.
\begin{corollary}[One-Wayness of Sponge]\label{cor:oneway_sponge}
Let $A$ be a quantum algorithm that makes $q$  quantum queries to a uniformly random permutation $\pi$ on $\{0,1\}^{r+c}$ and its inverse.  
Then it holds that 
\[
\Pr_\pi\left[\sponge^\pi[m,n](x')=y:
\begin{array}{ll}
x\leftarrow \{0,1\}^{m}\\
y:=\sponge^\pi[m,n](x)\\
x' \leftarrow A^{\pi}(y)
\end{array}
\right]\le (8q+1)^{4\ell} \left(\frac{12}{2^{\min\{m,n\}}}
+
\frac{2(2\ell+3)^2}{2^c}
\right).
\]
\end{corollary}
We omit its proof since this can be proven similarly by using \Cref{thm:lifting_interactive} instead of \Cref{thm:quantum_lifting}. Note that this is non-tight even in the single-round case due to the quartic loss in $q$. We note that a tight bound in the single-round setting is  shown in \cite{CP24}. 
%In Appendix ?? we show an alternative bound that only incurs quadratic factor in $q$, while making another factor worse, making it incomparable to the above corollary. 
\takashi{
Here is an idea for another incomparable bound.  
Observe that the sponge construction is almost pair-wise independent. Thus, we can use the leftover hash lemma to reduce one-wayness to preimage-resistance. That would give something like 
$(8q+1)^{2(\ell_a+\ell_s-1)} \left(\frac{2}{2^{\min\{m,n\} /3}}
+
\frac{6}{2^c}
\right).$
}

Next, we show the collision-resistance. 
\begin{corollary}[Collision-Resistance of Sponge]\label{cor:collision_sponge}
Let $A$ be a quantum algorithm that makes $q$  quantum queries to a uniformly random permutation $\pi$ on $\{0,1\}^{r+c}$ and its inverse. 
Then it holds that 
\[
\Pr_\pi\left[\sponge^\pi[m,n](x_1)=\sponge^\pi[m,n](x_2):
(x_1,x_2) \leftarrow A^{\pi}
\right]\le (8q+1)^{4\ell}\left(\frac{12}{2^{n}}
+
\frac{2(2\ell+3)^2}{2^c}
\right).
\]
\end{corollary}
\begin{proof}
    Define $R:=\left(\{0,1\}^{n}\right)^2\times \cup_{y\in \{0,1\}^n}\ \{(y,y)\}$. By the second item of \Cref{lem:upper_bound_P_max}, we have $P_{\mathrm{max}}^{R}\le 6\cdot 2^{-n}$. Substituting this and $k=2$ to \Cref{thm:lifting_sponge} gives the corollary. 
\end{proof}
For the single-round case where $\ell=1$,  
\Cref{cor:collision_sponge} gives a bound of
$(8q+1)^{4} \left(\frac{12}{2^{n}}
+
\frac{50}{2^c}
\right).$ 
This is unlikely to be tight since the BHT algorithm~\cite{BHT98} would find a collision with $O(2^{n/3})$ queries.  
% \takashi{Does it work for any function?}\minki{No, for example if there is a unique collision for $f:[N]\to [N]$ then the quantum query complexity may be $N^{2/3}$ as I remember.}\takashi{I changed the sentence.}
\cite{CPZ24} showed that single-round sponge satisfies quantum indifferentiability (even with precomputation) when $r\le c$, which would give a tight bound for collision-resistance. However, to our knowledge, ours is the first non-trivial bound in the regime of $r>c$ even for the single-round case. 
More generally, when $\ell=O(1)$, $n=\Omega(\lambda)$,  $c=\Omega(\lambda)$, and $q=\poly(\lambda)$, the RHS is $2^{-\Omega(\lambda)}$. 
This is the first non-trivial bound for the collision-resistance of sponge with more than one round though it is non-tight. 

The above corollary can be easily generalized to the multi-collision case. 
\begin{corollary}[Multi-Collision-Resistance of Sponge]\label{cor:multi_collision_sponge}
Let $k$ be a positive integer and  
 $A$ be a quantum algorithm that makes $q$  quantum queries to a uniformly random permutation $\pi$ on $\{0,1\}^{r+c}$ and its inverse.  
Then it holds that 
\begin{align*}
    \Pr_\pi&\left[\sponge^\pi(x_1)=...=\sponge^\pi(x_k):
(x_1,...,x_k) \leftarrow A^{\pi}
\right] \\
& \le 2\cdot (8q+1)^{2k\ell}\left(\frac{{2k \choose k}}{2^{(k-1)n}}
+
\frac{(k\ell+k+1)^2}{2^c}
\right).
\end{align*}
where  
we write $\sponge^\pi$ to mean $\sponge^\pi[m,n]$. 
\end{corollary}
Since its proof is similar to that of \Cref{cor:collision_sponge}, we omit it.

\paragraph{\bf Proof of \Cref{thm:lifting_sponge}}
Finally, we prove \Cref{thm:lifting_sponge}. We prove it by using \Cref{thm:lifting} to reduce quantum security to classical security. To analyze classical security of sponge, we recall the following lemma. 
\begin{lemma}[Classical Indifferentiability of Sponge~\cite{BDP+08}]\label{lem:indiff_sponge}
There exists a stateful algorithm $S$ such that for any algorithm $D$ that makes $q_1$ classical queries to the sponge function $\sponge^\pi[m,n]$ and $q_2$ classical queries to $\pi$ (in either forward or backward direction), it holds that 
\[
\left|\Pr_{\pi}[D^{\sponge^\pi[m,n],\pi}=1]-\Pr_H[D^{H,S^{H}}=1]\right|
\le \frac{N(N+1)}{2^c}
\]
where $\pi$ is a uniformly random permutation on $\{0,1\}^{r+c}$, $H$ is a uniformly random function from $\{0,1\}^m$ to  $\{0,1\}^n$, and $N:=q_1\ell+q_2$. Moreover $S$ makes at most one classical query to $H$ whenever it is queried.
\end{lemma}
\begin{remark}
\cite[Theorem 1]{BDP+08} shows a bound
   $1-\prod_{i=1}^N\left(\frac{1-\frac{i+1}{2^c}}{1-\frac{i}{2^r2^c}}\right)$ when $N<2^c$. 
   This can be further upper bounded as follows: 
     \[1-\prod_{i=1}^N\left(\frac{1-\frac{i+1}{2^c}}{1-\frac{i}{2^r2^c}}\right)
     \le 
     1-
\left(1-\frac{N+1}{2^c}\right)^N
\le
1-\left(1-\frac{N(N+1)}{2^c}\right)
=\frac{N(N+1)}{2^c}.
     \]
Moreover, the bound in the lemma  trivially holds when $N\ge 2^c$ since in this case $\frac{N(N+1)}{2^c}\ge 1$. Thus, we obtain the above lemma. 
  %  \takashi{Is the above  correct? I believe it's correct, but it's weird that the original paper didn't give this simplified upper bound.}\minki{I need to check it carefully but I believe you are right, given the approximate inequality is given in (6) of the original paper.}\minki{I checked it and I agreed that you are right.}
\end{remark}
Then we prove \Cref{thm:lifting_sponge}.
\begin{proof}[Proof of \Cref{thm:lifting_sponge}]
When $\frac{(k\ell+k+1)^2}{2^c}\ge 1$, the inequality trivially holds. Thus, we only have to prove the theorem assuming $\frac{(k\ell+k+1)^2}{2^c}\le 1$. In this case, we have $\frac{k^2\ell^2}{2^c}\le 1$ and thus $\frac{k^2\ell^2}{2^{r+c}}\le  \frac{1}{2}$. 

Noting that each invocation of $\sponge^\pi[n,m]$ makes $\ell$ queries to $\pi$, \Cref{thm:lifting} directly gives the following:
There exists a $k$-classical-query algorithm $B$ such that 
\begin{align*}
&\Pr_{\pi}\left[
(x_1,...,x_k,\sponge^\pi[m,n](x_1),...,\sponge^\pi[m,n](x_k))\in R:
(x_1,...,x_k) \leftarrow A^{\pi}
\right]\\
&\le 2\cdot (8q+1)^{2k\ell}
\Pr_{\pi}\left[
(x_1,...,x_k,\sponge^\pi[m,n](x_1),...,\sponge^\pi[m,n](x_k))\in R:
(x_1,...,x_k) \leftarrow B^{\pi}
\right]
\end{align*}
where we used $\frac{k^2\ell^2}{2^{r+c}}\le  \frac{1}{2}$. 
Let $P_B$ the probability in the second line. 
To upper bound $P_B$, consider the following algorithm $D$:
\begin{description}
\item[$D^{\sponge^\pi[m,n],\pi}$:]
It runs $B^{\pi}$ to obtain an output $(x_1,...,x_k)$. 
Then it outputs $1$ if and only if $$(x_1,...,x_k,\sponge^\pi[m,n](x_1),...,\sponge^\pi[m,n](x_k))\in R.$$ 
\end{description}
Then we clearly have 
\[
P_B=\Pr[D^{\sponge^\pi[m,n],\pi}=1].
\]
Note that $D$ makes at most $k$ classical queries to $\sponge^{\pi}[m,n]$ and at most $k$ queries to $\pi$ and $\pi^{-1}$. 
Thus, by \Cref{lem:indiff_sponge}, there is $S=(S_0,S_1)$ that makes at most one classical query such that 
\[
\left|\Pr_{\pi}[D^{\sponge^\pi[m,n],\pi}=1]-\Pr_H[D^{H,S^{H}}=1]\right|
\le \frac{(k\ell+k)(k\ell+k+1)}{2^c}\le \frac{(k\ell+k+1)^2}{2^c}.
\] 
By the definition of $D$ and $P_{\mathrm{max}}^R$, we have 
$\Pr_H[D^{H,S^{H}}=1]\le P_{\mathrm{max}}^R$. Combining the above we obtain \Cref{thm:lifting_sponge}.  
\end{proof}
\begin{remark}
The above proof only relies on the fact that the sponge construction satisfies classical indifferentiability. 
Thus, a similar argument works for any construction that satisfies classical indifferentiability in the invertible permutation model or ideal cipher model (e.g., \cite{GBJ+23}). \takashi{I believe there should be more.}
\end{remark}
\subsection{Davies-Meyer and PGV hash functions}\label{sec:davies-meyer}

The Merkle-Damg{\aa}rd hash function is a class of hash functions which uses a `compression function' $f(h, m)$ and use it iteratively to process a long message $M$ block-by-block to compute the hash of the message.  One benefit of the Merkle-Damg{\aa}rd construction is that, as long as $f(h, m)$ is collision-resistant, the Merkle-Damg{\aa}rd hash function constructed from $f$ is also collision-resistant.  Hence, the job of designing a collision-resistant hash function is reduced to designing a secure compression function, which is a fixed-input-length hash function.

The Davies-Meyer construction~\cite{Winternitz84} is a block-cipher-based compression function that underlies SHA-1, SHA-2, and MD5.  When $E$ is a block cipher, then the Davies-Meyer construction turns $E$ into a compression function $f$ by
\[
f(h, m):=E_m(h)\oplus h.
\]

In other words, in the Davies-Meyer construction, the previous state $h$ is first `encrypted' by $E$, using the message block $m$ as the cipher key, and the resulting ciphertext is XORed with $h$ (feed-forward).

Since we cannot expect the cipher key $m$ to be uniform random in this scenario, it would be difficult to prove the security of Davies-Meyer by modeling the block cipher $E$ as a pseudorandom permutation.  But, in the classical setting, it is proven to satisfy one-wayness and collision-resistance in the ideal cipher model~\cite{Winternitz84,BRS02,BRSS2010}. 

In the quantum setting, although it is shown to be one-way in the QICM~\cite{DBLP:conf/asiacrypt/HosoyamadaY18}, its collision-resistance remained an open question. 

In fact, the Davies-Meyer construction can be considered as a special case of a family of block-cipher-based compression functions.  Preneel, Govaerts, and Vandewalle~\cite{PGV93} studied block-cipher-based compression functions of the following form:
\[
f(h, m)=E_k(x)\oplus s,\quad\text{where $k, x, s\in\{c, h, m, h\oplus m\}$.}
\]
Here, $c$ is an arbitrarily fixed constant.  We may call the resulting 64 functions as PGV functions.  PGV studied them from the point of view of cryptanalysis.  They regarded 12 of them as secure, and showed that 39 of them are insecure by exhibiting damaging attacks.  The rest of 13 functions are subject to not very severe potential attacks.

Black, Rogaway, and Shrimpton~\cite{BRS02} studied the classical security of the PGV hash functions from the point of view of provable security.  Let $f_1, \dots, f_{64}$ be the PGV compression functions, and $H_1, \dots, H_{64}$ be the corresponding Merkle-Damg{\aa}rd hash functions.  They classified these hash functions into three groups.  Group-1 consists of 12 hash functions $H_1, \dots, H_{12}$, for which they have proved optimal one-wayness and collision-resistance in the ideal cipher model.  In fact, for group-1, the underlying compression functions $f_{1}, \dots, f_{12}$ have optimal one-wayness and collision-resistance.  Davies-Meyer is $H_5$ in group-1, according to this classification.

Their group-2 consists of 8 hash functions $H_{13}, \dots, H_{20}$, for which they have also proved one-wayness and collision-resistance in the ideal cipher model, but group-2 hash functions have suboptimal one-wayness: the adversarial advantage for breaking one-wayness is bounded by $\Theta(q^2/2^n)$, where $q$ is the total number of queries.  In case of group-2, the underlying compression functions \emph{do not} have one-wayness or collision-resistance: they are easily invertible and also allow quick collision-finding.  But, using the Merkle-Damg{\aa}rd construction, the compression functions can be turned into secure hash functions.

The rest, group-3, consists of 44 hash functions $H_{21}, \dots, H_{64}$, for which there are damaging collision-finding attacks.  

We may lift the security proofs of Black et al.\ on PGV functions to quantum adversaries in the quantum ideal cipher model.  For example, consider the following classical collision-resistance of the group-1 compression functions given in~\cite{BRS02}.

\begin{lemma}
Let $f$ be a group-1 PGV compression function of $n$-bit output.  Then, if $\cA$ is a collision-finding algorithm making $q$ classical queries to the ideal cipher $E$, then the collision-finding probability of $\cA$ is bounded above by $q(q+1)/2^n$.
\end{lemma}

Given a potential collsion pair $(h, m)\neq (h', m')$, it takes $k=2$ classical queries to the cipher $E$ to compute $f(h, m)$ and $f(h', m')$ and verify whether it is indeed a collision pair or not.  Hence, using Theorem~\ref{thm:quantum_icm_lifting_informal}, we obtain the following theorem.

\begin{theorem}
Let $f$ be a group-1 PGV compression function of $n$-bit output.  Then, if $\cA$ is a collision-finding algorithm making $q$ quantum queries to the ideal cipher $E$, then the collision-finding probability of $\cA$ is bounded above by
\[
\frac{6(8q+1)^4}{2^n-4}=O\left(\frac{q^4}{2^n}\right).
\]  
\end{theorem}

While the bound is not tight (when we compare it with the BHT quantum collsion-finding algorithm), still the collsion-resistance of Davies-Meyer or other secure PGV functions in QICM was not known before.

\printbibliography

\appendix

\section{Deferred Proofs} \label{app:deferred}

In this section, we present some proofs deferred for better readability.

\begin{customlemma}{\ref{lem:commutativity_disjoint_pairs}}[Commutativity of reprogramming for disjoint pairs]\label{lem:commutativity_disjoint_pairs_proof}
    For every permutation $\pi: X \to X$, disjoint $\vec{x}, \vec{y} \in X^k \times X^k$, and any permutation $\sigma: [k] \to [k]$, we have:
    \begin{align*}
        \pi [x_1 \to y_1] \ldots [x_k \to y_k] = \pi [x_{\sigma(1)} \to y_{\sigma(1)}] \ldots [x_{\sigma(k)} \to y_{\sigma(k)}].
    \end{align*}
\end{customlemma}
\begin{proof}
We need only to show that
\[
\pi[x\to y][x'\to y']=\pi[x'\to y'][x\to y],
\]
for any $x\neq x'$, $y\neq y'$, and $\pi$.

For this, let us denote by $(a_{1}\ a_{2}\  \dots\, a_{k})$ the cycle, which is the permutation which maps $a_{i}$ to $a_{i+1}$ for $i=1, \dots, k-1$, and also maps $a_{k}$ to $a_{1}$, and leaves everything else.  The special case $(a\ b)$, is the transposition of $a$ and $b$, and $(a\ a)$ is the identity permutation for any $a$.

Now, observe that $\pi[x\to y]=\pi \circ (x\ \pi^{-1}(y))$.  Then, $\pi[x\to y]^{-1}=(x\ \pi^{-1}(y))\circ\pi^{-1}$.

Then,
\[
\begin{aligned}
\pi[x\to y][x'\to y']&=\pi \circ (x\ \pi^{-1}(y))\circ (x'\ \pi[x\to y]^{-1}(y'))\\
&=\pi\circ (x\ \pi^{-1}(y))\circ (x'\ (x\ \pi^{-1}(y))(\pi^{-1}(y'))).
\end{aligned}
\]
Similarly, we have
\[
\pi[x'\to y'][x\to y]=\pi\circ (x'\ \pi^{-1}(y'))\circ (x\ (x'\ \pi^{-1}(y'))(\pi^{-1}(y))).
\]
So, we need only to show that
\[
(x\ \pi^{-1}(y))\circ (x'\ (x\ \pi^{-1}(y))(\pi^{-1}(y')))=(x'\ \pi^{-1}(y'))\circ (x\ (x'\ \pi^{-1}(y'))(\pi^{-1}(y))).
\]
For this, we perform case analysis.  We will consider four cases:

First, when $y'=\pi(x)$ and $y=\pi(x')$, then
\begin{align*}
\text{LHS}&=(x\ x')\circ (x'\ \pi^{-1}(y))=(x\ x')\circ (x'\ x')=(x\ x'),\\
\text{RHS}&=(x'\ x)\circ  (x\ \pi^{-1}(y'))=(x'\ x)\circ (x\ x)=(x'\ x).
\end{align*}
Here, $(x\ \pi^{-1}(y))(\pi^{-1}(y'))=\pi^{-1}(y)=x'$, because $y'=\pi(x) \implies \pi^{-1}(y')=x$.  Similarly we have $(x'\ \pi^{-1}(y'))(\pi^{-1}(y))=\pi^{-1}(y')$.  So we see that the two sides are equal in this case.

Next, consider the case when $y'=\pi(x)$ but $y\neq \pi(x')$.  Then,
\begin{align*}
\text{LHS}&=(x\ \pi^{-1}(y))\circ (x'\ \pi^{-1}(y))=(\pi^{-1}(y)\ x'\ x),\\
\text{RHS}&=(x'\ x)\circ (x\ \pi^{-1}(y))=(\pi^{-1}(y)\ x'\ x).
\end{align*}
Here, $(x'\ \pi^{-1}(y'))(\pi^{-1}(y))=\pi^{-1}(y)$, because $y\neq y'$ and $y\neq \pi(x')$.  Again, the two sides are equal.

The other case when $y'\neq \pi(x)$ but $y=\pi(x')$ can be done symmetrically.

Finally, when $y'\neq \pi(x)$ and $y\neq\pi(x')$, 
\begin{align*}
\text{LHS}&=(x\ \pi^{-1}(y))\circ (x'\ \pi^{-1}(y')),\\
\text{RHS}&=(x'\ \pi^{-1}(y'))\circ (x\ \pi^{-1}(y)).
\end{align*}
Since it is easy to check that $(a\ b)\circ(c\ d)=(c\ d)\circ(a\ b)$, if $\{a, b\}\cap\{c, d\}=\emptyset$, and in this case we have $\{x,\pi^{-1}(y)\}\cap\{x',\pi^{-1}(y')\}=\emptyset$, we see that the two sides agree.

Therefore, we see that $\pi[x\to y][x'\to y']=\pi[x'\to y'][x\to y]$ holds, if $x\neq x'$ and $y\neq y'$.   
\end{proof}

In order to prove Lemma~\ref{lem:reprogram_good_tuples}, we first prove a few simple properties of the reprogramming.

\begin{lemma}\label{lem:good_tuple_properties}
For any permutation $\pi$ and any tuple $(p_1^\ast, \dots, p_k^\ast)$ of pairs $p_i^\ast=(x_i^\ast, y_i^\ast)$ which is good w.r.t.\ $\pi$, the following are all true:
\begin{enumerate}
    \item $\pi[x_1^\ast\to y_1^\ast]\dots[x_k^\ast\to y_k^\ast](x_k^\ast)=y_k^\ast$.
    \item $\pi[x_1^\ast\to y_1^\ast]\dots[x_{k-1}^\ast\to y_{k-1}^\ast](x_k^\ast)=\pi(x_k^\ast)$.
    \item $\pi[x_1^\ast\to y_1^\ast]\dots[x_{k-1}^\ast\to y_{k-1}^\ast]^{-1}(y_k^\ast)=\pi^{-1}(y_k^\ast)$.
    \item $\pi[x_1^\ast\to y_1^\ast]\dots[x_{k}^\ast\to y_{k}^\ast](\pi^{-1}(y_k^\ast))=\pi(x_k^\ast)$.
\end{enumerate}
\end{lemma}
\begin{proof}
Let us prove the above statements one by one.
\begin{itemize}
    \item Statement~1: It directly follows from the definition of the permutation reprogramming.
    \item Statement~2: It can be proved by mathematical induction.
    
    The base case is where $k=1$, and in that case the equality holds trivially. 
    
    Now, suppose that this statement holds for some $k-1$, and we prove it for $k$.  
    
    Consider the permutation $\pi'=\pi[x_1^\ast\to y_1^\ast]$.  Then, we can see that the tuple $(p_2^\ast, \dots, p_k^\ast)$ is good w.r.t.\ $\pi'$: for any $i, j>1$, since $x_i^\ast\neq x_1^\ast$ and also $x_i^\ast\neq \pi^{-1}(y_1^\ast)$ (due to goodness), we get $\pi'(x_i^\ast)=\pi[x_1^\ast\to y_1^\ast](x_i^\ast)=\pi(x_i^\ast)\neq y_j^\ast$ by the definition of reprogramming and also by goodness.

    Then, according to the induction hypothesis applied to $\pi'$ and $(p_2^\ast, \dots, p_k^\ast)$, a tuple of length $k-1$, we have 
    \[
    \pi'[x_2^\ast\to y_2^\ast]\dots[x_{k-1}^\ast\to y_{k-1}^\ast](x_k^\ast)=\pi'(x_k^\ast).
    \]
    This shows that $\pi[x_1^\ast\to y_1^\ast]\dots[x_{k-1}^\ast\to y_{k-1}^\ast](x_k^\ast)=\pi'(x_k^\ast)$, which we have already seen to be equal to $\pi(x_k^\ast)$, finishing the mathematical induction.
    \item Statement~3: This is a direct consequence of Statement~2, after we invert the reprogramming using Lemma~\ref{lem:reprogramming_inverse}.
    \item Statement~4: Let us define $\pi'=\pi[x_1^\ast\to y_1^\ast]\dots[x_{k-1}^\ast\to y_{k-1}^\ast]$.  Then, Statement~2 says that $\pi'(x_k^\ast)=\pi(x_k^\ast)$, and Statement~3 says that $\pi'^{-1}(y_k^\ast)=\pi^{-1}(y_k^\ast)$.  Then,
    \[
    \begin{aligned}
        \pi[x_1^\ast\to y_1^\ast]\dots[x_{k}^\ast\to y_{k}^\ast](\pi^{-1}(y_k^\ast))
        &=\pi'[x_{k}^\ast\to y_{k}^\ast](\pi^{-1}(y_k^\ast))\\
        &=\pi'[x_{k}^\ast\to y_{k}^\ast](\pi'^{-1}(y_k^\ast))\\
        &=\pi'(x_k^\ast)=\pi(x_k).
    \end{aligned}
    \]
\end{itemize}
\end{proof}

\begin{customlemma}{\ref{lem:reprogram_good_tuples}}[Reprogramming on good tuples]\label{lem:reprogram_good_tuples_proof}
    Consider any permutation $\pi$ and $k$ pairs $p_1^\ast,\dots, p_k^\ast$ with $p_j^\ast=(x_j^\ast, y_j^\ast)$ for $j=1, \dots, k$.  Suppose the tuple of pairs $(p_1^\ast, \dots, p_k^\ast)$ is good w.r.t.\ $\pi$. Then we have: 
    \[
    \pi[\vec{x}^* \to \vec{y}^*](z)=
    \begin{cases}
    y_j^\ast & \text{if $z=x_j^\ast$ for some $j\in[k]$,}\\
    \pi(x_j^\ast) & \text{if $z=\pi^{-1}(y_j^\ast)$ for some $j\in[k]$,}\\
    \pi(z) & \text{otherwise.}
    \end{cases}
    \]
\end{customlemma}
\begin{proof}
We have to prove this case by case.

First, for the case $z=x_j^\ast$ for some $j\in [k]$, we need to show
\[
\pi[x_1^\ast\to y_1^\ast]\dots[x_k^\ast\to y_k^\ast](x_j^\ast)=y_j^\ast.
\]

In fact, Statement~1 of Lemma~\ref{lem:good_tuple_properties} already shows this equation for the case $j=k$, but since we can shuffle the order or reprogramming arbitrarily, thanks to Lemma~\ref{lem:commutativity_disjoint_pairs}, the general case also follows from Statement~1.

Second, for the case $z=\pi^{-1}(y_j^\ast)$ for some $j\in [k]$, we need to show
\[
\pi[x_1^\ast\to y_1^\ast]\dots[x_k^\ast\to y_k^\ast](\pi^{-1}(y_j^\ast))=\pi(x_j^\ast).
\]

Again, Statement~4 of Lemma~\ref{lem:good_tuple_properties}, together with Lemma~\ref{lem:commutativity_disjoint_pairs}, proves this equality.

Finally, for the case $z\neq x_j^\ast$ and $z\neq \pi^{-1}(y_j^\ast)$ for any $j\in[k]$, we need to show
\[
\pi[x_1^\ast\to y_1^\ast]\dots[x_k^\ast\to y_k^\ast](z)=\pi(z).
\]

Observe that we need only to prove that it is possible to get rid of just one reprogramming:
\[
\pi[x_1^\ast\to y_1^\ast]\dots[x_k^\ast\to y_k^\ast](z)=
\pi[x_1^\ast\to y_1^\ast]\dots[x_{k-1}^\ast\to y_{k-1}^\ast](z).
\]

From this, we get the desired equality by repeatedly removing one reprogramming at a time.

Let $\pi'=\pi[x_1^\ast\to y_1^\ast]\dots[x_{k-1}^\ast\to y_{k-1}^\ast]$.  Then we need to show
\[
\pi'[x_k^\ast\to y_k^\ast](z)=\pi'(z).
\]

But since $z\neq x_k^\ast$ and $z\neq \pi'^{-1}(y_k^\ast)$, this just comes from the definition of reprogramming.  (Note that Statement~3 of Lemma~\ref{lem:good_tuple_properties} says that $\pi'^{-1}(y_k^\ast)=\pi^{-1}(y_k^\ast)$.)   
\end{proof}

\section{Handling Interactive Setting} \label{app:interactive}

%\alex{Here is a sketch of the generalization of the lifting theorem for interactive games. At a high level, this follows the strategy of \cite{YZ21}.}

%\alex{
In this section, we extend our lifting theorem in order to show that the lifting also holds in the interactive setting, when $\cC$ and $\cA$ are allowed send multiple messages and their queries can depend on the interaction. 
At a high level, we can show this by suitably adapting the strategy of \cite{YZ21} to the permutation setting.
More formally, we can show that:
%}
\begin{theorem} \label{thm:lifting_interactive}
    Let $G$ be any search-type classically verifiable game (as defined in \cite{YZ21}) played with a Challenger that performs at most $k$ classical queries to a random permutation $\pi^* : X \rightarrow X$.
    Then for any adversary $\cA$ equipped with $q$ quantum queries against the game $G$, we can construct a simulator $\cB$ performing at most $k$ classical queries such that:
    $$ \Pr_{\pi^*}[\cB^{\pi^*} \text{ wins } G]  \geq \frac{\left(1 - \frac{k^2}{|X|}\right)}{(8q+1)^{2k}}\Pr_{\pi}[\cA^{\pi^*} \text{ wins } G] $$
\end{theorem}

\begin{proof}
    
The algorithm $\cB^{\pi^*}$ will follow exactly the simulator in the permutation measure-and-reprogram experiment in \Cref{def:quantum_simulator} 
%\revise{
for uniformly chosen $\pi$
%}
, with the only difference that in the interactive setting, $\cB$ will run the algorithm $\cA$ by forwarding all messages supposed to be sent to $\cC$ to the external challenger and forwarding all messages sent
back from the external challenger to $\cA$.

%\revise{
As argued in \cite{YZ21}, we can assume that the challenger $\cC$ is deterministic without loss of generality.
%}
%\takashi{I believe this is important since otherwise the relation is not well-defined.} \alex{Sorry I forgot to add that $\cC$ is deterministic, that's completely right.}
We define $\cC$'s view as consisting of the set of $k$ query response pairs together with the transcript consisting of the messages exchanged with $\cA$. More formally, we will denote the view by: $v := (\vec{x} = (x_1, ..., x_k) , \vec{y} = (y_1, ..., y_k), t)$, where 
%$y_i = \pi(x_i)$
$y_i = \pi^*(x_i)$
and $t$ denotes the transcript.

Additionally, we define the relation $R_{\sf view}$ corresponding to all accepting views of $\cC$, in other words, all possible views of the challenger that would result in $\cC$ outputting accept. More concretely, we say $v = (\vec{x}, \vec{y}, t) \in R_{\sf view}$ if and only if the following view verification algorithm ${\sf VerView}(\vec{x}, \vec{y}, t)$ accepts. \\
${\sf VerView}(\vec{x}, \vec{y}, t)$:
\begin{itemize}
    \item Run $\cC$ such that the messages supposed to be sent from $\cA$ and its query responses are consistent with the view $v$;
    \item If $\cC$ cannot be consistent with $v$, then output reject;
    \item Else output $\cC$'s output.
\end{itemize}

The goal is to apply the measure and reprogram lemma with the target relation instantiated as $R_{\sf view}$ and where the target algorithm will be an algorithm $S_{\sf interact}^{\pi'}$ simulating the interaction between $\cA$ and $\cC$. Specifically, for any permutation $\pi': X \rightarrow X$, we define $S_{\sf interact}^{\pi'}$ as follows. \\
$S_{\sf interact}^{\pi'}$:
\begin{itemize}
    \item $\cA$'s queries are forwarded to $\pi'$ and responded using $\pi'$;
    \item For every query $x_i$ of $\cC$ respond with the original 
    %$\pi(x_i)$
    $\pi^*(x_i)$
    , for $i \in [k]$;
    \item Output $\cC$'s queries $\vec{x} = (x_1, ..., x_k)$ and the transcript $t'$ between $\cA$ and $\cC$.
\end{itemize}

% Using \Cref{thm:quantum_lifting} for the algorithm $S_{\sf interact}$, for the relation $R_{\sf view}$ described above and for an arbitrary permutation $\pi'$, we have that there exist a classical algorithm $\cB'$ making exactly $k$ queries such that:
% \begin{align*}
% &\Pr_{\pi'} \left[
% (x_1,...,x_k,\pi'(x_1),...,\pi'(x_k),z) \in R_{\sf view}
% :(x_1,...,x_k,z)\leftarrow B^{\pi'} \right] \\
% &\ge
% \frac{\left(1 - \frac{k^2}{|X|}\right)}{(8q+1)^{2k}}\Pr_{\pi'} \left[
% (x_1,...,x_k,\pi'(x_1),...,\pi'(x_k),z)\in R_{\sf view}
% :(x_1,...,x_k,z)\leftarrow S_{\sf interact}^{\pi'}\right].
% \end{align*}

%\alex{We now need to connect the RHS with the probability of a general adversary $\cA$ to win the search game}

%From the proof of \Cref{thm:quantum_lifting} we know that for a random permutation $\pi$ there exists a simulator $S$ such that for any permutation $\pi^*$, for any relation $R$ and for any disjoint $\vec{x}^*$ and any $z$, and $\vec{y}^* = \pi^*(\vec{x}^*)$ we have: \takashi{Don't we need to assume $(\pi,\pi^*)\in G[\vec{x}^*]$?}
%\begin{align*}
%    &\Pr\left[
%\begin{array}{ll}
%(x^*_1,...,x^*_k,y^*_1,...,y^*_k,z)\in R \\
%~\land ~
%\forall j\in[k]~x_j=x^*_j\\
%\end{array}
%:(x_1,...,x_k,z)\leftarrow S[A, \pi,\pi^*]\right]\\
%&\ge  
%\frac{1}{(8q+1)^{2k}}
%\Pr\left[
%\begin{array}{ll}
%(x^*_1,...,x^*_k,y^*_1,...,y^*_k,z)\in R \\
%~\land ~
%\forall j\in[k]~x_j=x^*_j\\
%\end{array}
%:(x_1,...,x_k,z)\leftarrow A^{\pi[\vec{x}^* \rightarrow \vec{y}^*]}\right]. 
%\end{align*}

%\revise{
By \Cref{lemma:quantum_measure_reprogram} where we instantiate
$R$ with $R_{\sf view}$ and the adversary $\cA$ with $S_{\sf interact}$, 
for any distinct $\vec{x}^*$ and $(\pi,\pi^*)\in G[\vec{x}^*]$,
we have:
%}

\begin{align} \label{eq:sim_trans}
    &\Pr\left[
\begin{array}{ll}
(x^*_1,...,x^*_k,y^*_1,...,y^*_k,z)\in R_{\sf view} \\
~\land ~
\forall j\in[k]~x_j=x^*_j\\
\end{array}
:(x_1,...,x_k,z)\leftarrow S[S_{\sf interact}, \pi,\pi^*]\right]\\
&\ge  \nonumber
\frac{1}{(8q+1)^{2k}}
\Pr\left[
\begin{array}{ll}
(x^*_1,...,x^*_k,y^*_1,...,y^*_k,z)\in R_{\sf view} \\
~\land ~
\forall j\in[k]~x_j=x^*_j\\
\end{array}
:(x_1,...,x_k,z)\leftarrow S_{\sf interact}^{\pi[\vec{x}^* \rightarrow \vec{y}^*]}\right] 
\end{align}
where $y_i^* = \pi^*(x_i^*)$ for any $i \in [k]$.

Now let us examine the simulator $S_{\sf interact}^{\pi[\vec{x}^* \rightarrow \vec{y}^*]}$.
By construction, $S_{\sf interact}^{\pi[\vec{x}^* \rightarrow \vec{y}^*]}$ (where we instantiate $\pi' = {\pi[\vec{x}^* \rightarrow \vec{y}^*]}$) will simulate the interaction between the original adversary $\cA$ and the challenger $\cC$ of the game $G$, where now $\cA$'s oracle queries are simulated by $\pi[\vec{x}^* \rightarrow \vec{y}^*]$,
$\cC$'s queries are simulated by the original permutation 
$\pi^*$,
%$\pi$,
and $S_{\sf interact}^{\pi[\vec{x}^* \rightarrow \vec{y}^*]}$ will output $\cC$'s queries $\vec{x}$ and the transcript $z := t'$ (between $\cA$ and $\cC$).

Additionally, conditioned on $\vec{x} = \vec{x}^*$, $S_{\sf interact}^{\pi[\vec{x}^* \rightarrow \vec{y}^*]}$ will simulate the interaction between $\cA$ and $\cC$ where now $\cA$ and $\cC$ are going to query exactly the same oracle $\pi[\vec{x}^* \rightarrow \vec{y}^*]$, as for any $x^* \in \vec{x}^*$ we have: $\pi[\vec{x}^* \rightarrow \vec{y}^*](x^*) = \pi^*(x^*) = y^*$.
Moreover, conditioned on $\vec{x} = \vec{x}^*$, the condition $(\vec{x}, \pi^*(\vec{x}^*), z) \in  R_{\sf view}$ is equivalent to $\cA^{\pi[\vec{x}^* \rightarrow \pi^*(\vec{x}^*)]}$ wins $\cC^{\pi[\vec{x}^* \rightarrow \pi^*(\vec{x}^*)]}$ in the executed simulation of $S_{\sf interact}^{\pi[\vec{x}^* \rightarrow \pi^*(\vec{x}^*)]}$.

%\revise{
Therefore, we can rewrite \Cref{eq:sim_trans} as: %\revise{the probability in} the RHS of \Cref{eq:sim_trans} as:
\begin{align} \label{eq:sim_trans_rewrite}
    &\Pr\left[
\begin{array}{ll}
(x^*_1,...,x^*_k,y^*_1,...,y^*_k,z)\in R_{\sf view} \\
~\land ~
\forall j\in[k]~x_j=x^*_j\\
\end{array}
:(x_1,...,x_k,z)\leftarrow S[S_{\sf interact}, \pi,\pi^*]\right]\\
&\ge  \nonumber
\frac{1}{(8q+1)^{2k}}
\Pr\left[ \vec{x} = \vec{x}^* ~\land ~ \cA^{\pi[\vec{x}^* \rightarrow \vec{y}^*]} \text{ wins } G \text{ with } \cC^{\pi[\vec{x}^* \rightarrow \vec{y}^*]}\right]
\end{align}
where $\vec{x}$ denotes the queries made by $\cC$.
%}
\if0
\begin{align*}
    \Pr & \left[
\begin{array}{ll}
(x^*_1,...,x^*_k,\pi^*(x^*_1),...,\pi^*(x^*_k),z)\in R_{\sf view} \\
~\land ~
\forall j\in[k]~x_j=x^*_j\\
\end{array}
:(x_1,...,x_k,z)\leftarrow S_{\sf interact}^{\pi[\vec{x}^* \rightarrow \revise{\vec{y}^*}]}\right] \\
&= \Pr\left[ \vec{x} = \vec{x}^* ~\land ~ \cA^{\pi[\vec{x}^* \rightarrow \vec{y}^*]} \text{ wins } G \text{ with } \cC^{\pi[\vec{x}^* \rightarrow \revise{\vec{y}^*}]}\right]
\end{align*}
where \revise{$\vec{x}$} denotes the queries made by $\cC$.
\fi 
%If $\pi$ and $\pi^*$ are uniformly picked, then the permutation $\pi[\vec{x}^* \rightarrow \pi^*(\vec{x}^*)]$ is also uniform over all permutations for any fixed $\vec{x}^*$. \takashi{We are taking $(\pi,\pi^*)$ from $G[\vec{x}^*]$, not uniformly randomly. Shouldn't we use \Cref{lem:uniform}?} 

%\revise{
By \Cref{lem:uniform}, 
for any fixed $\vec{x}^*$,
if $\pi$ and $\pi^*$ are uniformly picked from $G[\vec{x}^*]$, then $\pi[\vec{x}^* \rightarrow \pi^*(\vec{x}^*)]$ is also uniform over all permutations.
%} 
Then, by taking the average over $(\pi, \pi^*) \in G[\vec{x}^*]$ and by summing over all $\vec{x}^*$ in \Cref{eq:sim_trans_rewrite}, we get:
\begin{align}\label{eq:upperbound_A_wins}
    & \sum_{\vec{x}^*} \Pr_{\pi, \pi^* \leftarrow G[\vec{x}^*]}\left[
\begin{array}{ll}
(x^*_1,...,x^*_k,y^*_1,...,y^*_k,z)\in R_{\sf view} \\
~\land ~
\forall j\in[k]~x_j=x^*_j\\
\end{array}
:(x_1,...,x_k,z)\leftarrow S[S_{\sf interact}, \pi,\pi^*]\right]\\
&\ge  \nonumber
\frac{1}{(8q+1)^{2k}}
\Pr_{\pi^*} \left[ \cA^{\pi^*} \text{ wins } G \text{ with } \cC^{\pi^*}\right].
\end{align}

On the other hand, we have:
\begin{align*}
&\Pr_{\pi^*}[\cB^{\pi^*} \text{ wins } G]\\
&=
\sum_{(x^*_1,...,x^*_k)}\Pr_{\pi,\pi^*}\left[
\begin{array}{ll}
(x^*_1,...,x^*_k,y^*_1,...,y^*_k,z)\in R_{\sf view} \\
~\land ~
\forall j\in[k]~x_j=x^*_j 
\end{array}
:(x_1,...,x_k,z)\leftarrow S[S_{\sf interact}, \pi, \pi^*]\right]\\
&\ge
\sum_{(x^*_1,...,x^*_k)}
\Pr_{\pi,\pi^*}[(\pi,\pi^*)\in G[\vec{x}^*]]\\
&\cdot 
\Pr_{(\pi,\pi^*)\leftarrow G[\vec{x}^*]}\left[
\begin{array}{ll}
(x^*_1,...,x^*_k,y^*_1,...,y^*_k,z)\in R_{\sf view} \\
~\land ~
\forall j\in[k]~x_j=x^*_j\} 
\end{array}
:(x_1,...,x_k,z)\leftarrow S[S_{\sf interact}, \pi, \pi^*]\right]\\
&\ge
\sum_{(x^*_1,...,x^*_k)}
\left(1-\frac{k^2}{|X|}\right)\\
&\cdot 
\Pr_{(\pi,\pi^*) \leftarrow G[\vec{x}^*]}\left[
\begin{array}{ll}
(x^*_1,...,x^*_k,y^*_1,...,y^*_k,z)\in R_{\sf view} \\
~\land ~
\forall j\in[k]~x_j=x^*_j\} 
\end{array}
:(x_1,...,x_k,z)\leftarrow S[S_{\sf interact}, \pi, \pi^*]\right]
\end{align*}
where the second inequality follows from \Cref{lem:bad_prob}.
Combining the above with \Cref{eq:upperbound_A_wins},
we conclude the proof of \Cref{thm:lifting_interactive}. 
\end{proof}

\section{Proof of \Cref{lem:quantum_correctness_of_reprogramming}\label{correctness_reprogramming} }
\begin{proof}[Proof of \Cref{lem:quantum_correctness_of_reprogramming}.]
By the assumption, $v_j\in [q]$. We analyze the behavior of $S[\cA,\pi,\pi^*]$ at the $v_j$-th query 
for each case: 
\paragraph{First case.}
In this case, $b_j=0$, and the measured query is a forward query $x'_{v_j}=x_j^{\mathsf{hit}}$. 
By the definition of $x_j^{\mathsf{hit}}$ (Definition 6), this means that $x'_{v_j}=x_j^*$. 
By the description of $S[\cA,\pi,\pi^*]$, the simulator queries $x'_{v_j}$ to $\pi^*$ to obtain 
$y'_{v_j}=\pi^*(x'_{v_j})=\pi^*(x_j^*)=y_j^*$.
Then the simulator reprograms $O$ to $O[x'_{v_j}\rightarrow y'_{v_j}]$ (before or after answering  $\cA$'s query according to $c_j$.) As shown above,  $x'_{v_j}=x^*_j$ and  $y'_{v_j}=y^*_j$. Thus, the claim of the lemma is satisfied. 

\paragraph{Second case.}
In this case, $b_j=0$, and the measured query is a backward query $y'_{v_j}=y_j^{\mathsf{hit}}$. 
By the definition of $y_j^{\mathsf{hit}}$ (Definition 6), this means that $y'_{v_j}=y_j^*$. 
By the description of $S[\cA,\pi,\pi^*]$, the simulator queries $y'_{v_j}$ to ${\pi^*}^{-1}$ to obtain 
$x'_{v_j}={\pi^*}^{-1}(y'_{v_j})={\pi^*}^{-1}(y_j^*)=x_j^*$.
Then the simulator reprograms $O$ to $O[x'_{v_j}\rightarrow y'_{v_j}]$ (before or after answering  $\cA$'s query according to $c_j$.) As shown above,  $x'_{v_j}=x^*_j$ and  $y'_{v_j}=y^*_j$. Thus, the claim of the lemma is satisfied. 

\paragraph{Third case.}
In this case, $b_j=1$, and the measured query is a forward query $x'_{v_j}=x_j^{\mathsf{miss}}$. 
By the definition of $x_j^{\mathsf{miss}}$ (Definition 6), this means that $x'_{v_j}=\pi^{-1}(y_j^*)$ and thus $\pi(x'_{v_j})=y^*_j$.  
By the description of $S[\cA,\pi,\pi^*]$, the simulator queries $\pi(x'_{v_j})$ to ${\pi^*}^{-1}$ to obtain 
${\pi^*}^{-1}(\pi(x'_{v_j}))
={\pi^*}^{-1}(y_j^*)=x_j^*$. 
Then the simulator reprograms $O$ to $O[{\pi^*}^{-1}(\pi(x'_{v_j}))\rightarrow \pi(x'_{v_j})]$ (before or after answering  $\cA$'s query according to $c_j$.) As shown above,  ${\pi^*}^{-1}(\pi(x'_{v_j}))=x^*_j$ and  $\pi(x'_{v_j})=y^*_j$. Thus, the claim of the lemma is satisfied. 

\paragraph{Fourth case.}
In this case, $b_j=1$, and the measured query is a backward query $y'_{v_j}=y_j^{\mathsf{miss}}$. 
By the definition of $y_j^{\mathsf{miss}}$ (Definition 6), this means that $y'_{v_j}=\pi(x_j^*)$ and thus $\pi^{-1}(y'_{v_j})=x^*_j$.  
By the description of $S[\cA,\pi,\pi^*]$, the simulator queries $\pi^{-1}(y'_{v_j})$ to $\pi^*$ to obtain 
$\pi^*(\pi^{-1}(y'_{v_j}))
=\pi^*(x_j^*)=y_j^*$. 
Then the simulator reprograms $O$ to $O[\pi^{-1}(y'_{v_j})\rightarrow \pi^*(\pi^{-1}(y'_{v_j}))]$ (before or after answering  $\cA$'s query according to $c_j$.) As shown above,  $\pi^{-1}(y'_{v_j})=x^*_j$ and  $\pi^*(\pi^{-1}(y'_{v_j}))=y^*_j$. Thus, the claim of the lemma is satisfied. 
\end{proof}

\end{document}